\documentclass[
a4paper,reqno]{amsart}
\usepackage[T1]{fontenc}
\usepackage[english]{babel}
\usepackage{amssymb,bbm,enumerate}
\usepackage{color,xcolor}

\usepackage[linktocpage=true,colorlinks=true, linkcolor=blue, citecolor=red, urlcolor=green]{hyperref}

\usepackage{tikz-cd}

\newcommand{\mc}[1]{\mathcal{#1}}
\newcommand{\mf}[1]{\mathfrak{#1}}
\newcommand{\mb}[1]{\mathbb{#1}}

\newcommand{\id}{\mathbbm{1}}
\newcommand{\tint}{{\textstyle\int}}

\newcommand{\xhat}{\widehat{\phantom{\mc V}}}
\newcommand{\yhat}{\!\!\!\!\widehat{\phantom{\mc V}}}

\DeclareMathOperator{\End}{End}

\DeclareMathOperator{\Tr}{Tr}
\DeclareMathOperator{\res}{Res}
\DeclareMathOperator{\ad}{ad}
\DeclareMathOperator{\im}{Im}

\DeclareMathOperator{\Ker}{Ker}

\DeclareMathOperator{\Res}{Res}

\DeclareMathOperator{\sign}{sign}

\theoremstyle{plain}
\newtheorem{theorem}{Theorem}[section]
\newtheorem{lemma}[theorem]{Lemma}
\newtheorem{proposition}[theorem]{Proposition}

\theoremstyle{definition}
\newtheorem{definition}[theorem]{Definition}
\newtheorem{example}[theorem]{Example}

\theoremstyle{remark}
\newtheorem{remark}[theorem]{Remark}

\numberwithin{equation}{section}

\definecolor{light}{gray}{.9}

\begin{document}

\title{Adler-Oevel-Ragnisco type operators and Poisson vertex algebras}

\author{Alberto De Sole}
\address{Dipartimento di Matematica \& INFN, Sapienza Universit\`a di Roma,
P.le Aldo Moro 2, 00185 Rome, Italy}
\email{desole@mat.uniroma1.it}
\urladdr{www1.mat.uniroma1.it/\$$\sim$\$desole}

\author{Victor G. Kac}
\address{Dept of Mathematics, MIT,
77 Massachusetts Avenue, Cambridge, MA 02139, USA}
\email{kac@math.mit.edu}

\author{Daniele Valeri}
\address{Dipartimento di Matematica \& INFN, Sapienza Universit\`a di Roma,
P.le Aldo Moro 5, 00185 Rome, Italy}
\email{daniele.valeri@uniroma1.it}


\begin{abstract}
The theory of triples of Poisson brackets and related integrable systems,
based on a classical $R$-matrix $R\in\End_{\mb F}(\mf g)$, where $\mf g$ is a finite dimensional 
associative algebra over a field $\mb F$ viewed as a Lie algebra,
was developed by Oevel-Ragnisco and Li-Parmentier \cite{OR89,LP89}.
In the present paper we develop an ``affine'' analogue of this theory by introducing the notion of a continuous
Poisson vertex algebra and constructing triples of Poisson $\lambda$-brackets. We introduce the corresponding
Adler type identities and apply them to integrability of hierarchies of Hamiltonian PDEs.
\end{abstract}

\maketitle

\begin{center}
\emph{To Claudio Procesi, a friend and a source of inspiration.}
\end{center}

%

\tableofcontents

\section{Introduction}
Following Semenov-Tian-Shansky \cite{STS83},
one defines a (classical) $R$-matrix over a Lie algebra $\mf g$ over a field $\mb F$
as an element $R\in\End_{\mb F}(\mf g)$ satisfying
the modified Yang-Baxter equation ($a,b\in\mf g$)
\begin{equation}\label{intro:eq1}
[R(a),R(b)]-R([R(a),b])-R([a,R(b)])+[a,b]=0
\,.
\end{equation}
The basic example is provided by a decomposition of $\mf g$ in a sum of two subalgebras
$\mf g_+$ and $\mf g_-$, such that $\mf g_+\cap\mf g_-=0$, by letting
\begin{equation}\label{intro:eq2}
R=\Pi_+-\Pi_-\,,
\quad\text{where }\Pi_{\pm}:\mf g\to\mf g_{\pm}\text{ are projections}
\,.
\end{equation}

Oevel and Ragnisco \cite{OR89} and independently Li and Parmentier \cite{LP89}
used this $R$-matrix in the case when $\mf g$ is a unital finite dimensional associative algebra,
endowed with a non-degenerate trace form $\Tr:\mf g\to\mb F$,
and viewed  as a Lie algebra with the bracket $[a,b]=a\circ b-b\circ a$, where
$a\circ b$ is the associative product in $\mf g$.
Namely, they construct a triple of Poisson brackets $\{\cdot\,,\cdot\}^R_i$,
$i=1,2,3$,  on $S(\mf g)$ as the coefficients
in the expansion ($a,b\in\mf g$)
\begin{equation}\label{intro:eq3}
\begin{split}
\{a,b\}^{R,\epsilon}
&=
\frac12 \sum_{i,j\in I}
(u^i\!+\!\epsilon\Tr(u^i))(u^j\!+\!\epsilon\Tr(u^j))
\big(
[a,R(u_i\circ b\circ u_j)]\!-\![b,R(u_i\circ a\circ u_j)]
\big)
\\
&=
\{a,b\}^{R}_3+2\epsilon \{a,b\}^{R}_2+\epsilon^2\{a,b\}^{R}_1
\,.
\end{split}
\end{equation}
Here $\{u_i\}_{i\in I}$ and $\{u^i\}_{i\in I}$ are dual bases of $\mf g$ with respect to the trace form
$\langle a|b\rangle=\Tr(a\circ b)$, and $\epsilon\in\mb F$.
Identifying $\mf g$ and $\mf g^*$ using the trace form, we obtain the following expression for arbitrary
$f,g\in S(\mf g)$:
\begin{equation}
\label{intro:eq4}
\begin{array}{r}
\displaystyle{
\vphantom{\Big(}
\{f,g\}^{R,\epsilon}(L)
=
\frac12\Tr\Big( L \circ [d_Lf,R((L+\epsilon\id)\circ d_Lg\circ (L+\epsilon\id))] \Big)
} \\
\displaystyle{
\vphantom{\Big(}
-
\frac12\Tr\Big( L \circ [d_Lg,R((L+\epsilon\id)\circ d_Lf\circ (L+\epsilon\id))] \Big)
\,,}
\end{array}
\end{equation}
where $L\in\mf g^*\cong\mf g$ and $d_Lf=\sum_{i\in I}\frac{\partial f}{\partial u_i}(L)u_i$.

Furthermore, they showed that the subalgebra of Casimirs $S(\mf g)^{\ad \mf g}$ is commutative with
respect to all three Poisson brackets. This leads to the following family of compatible Hamiltonian ODE
($C\in S(\mf g)^{\ad \mf g}$, $i,j,h,k\in I$):
\begin{equation}\label{intro:eq5}
\begin{array}{l}
\displaystyle{
\vphantom{\Big(}
\frac{du_j}{dt_{1}}
=
\frac12 
\sum_{i\in I}
\frac{\partial C}{\partial u_i}
\big(
[u_i,R(u_j)]-[u_j,R(u_i)]
\big)
\,,} \\
\displaystyle{
\vphantom{\Big(}
\frac{du_j}{dt_{2}}
=
\frac14 \sum_{i,h\in I}
\frac{\partial C}{\partial u_i}
u^h
\big(
[u_i,R(u_h\circ u_j+u_j\circ u_h)]-[u_j,R(u_h\circ u_i+u_i\circ u_h)]
\big)
\,,} \\
\displaystyle{
\vphantom{\Big(}
\frac{du_j}{dt_{3}}
=
\frac12 \sum_{i,h,k\in I}
\frac{\partial C}{\partial u_i}
u^hu^k
\big(
[u_i,R(u_h\circ u_j\circ u_h)]-[u_j,R(u_h\circ u_i\circ u_k)]
\big)
\,.}
\end{array}
\end{equation}
Taking the Casimirs ($k\in\mb Z_{\geq1}$)
\begin{equation}\label{intro:eq6}
C_k
=
\frac1k
\sum_{i_1,\dots,i_k}
u_{i_1}\dots u_{i_k} \Tr( u^{i_1}\circ\dots\circ u^{i_k} )
\,,
\end{equation}
we obtain the following triple Lenard-Magri scheme for the time evolutions $t_{k,i}$, $k\in\mb Z_{\geq1}$
and $i\in I$, of Hamiltonian
ODE with respect to the Poisson structure $\{\cdot\,,\cdot\}_j^R$, $j=1,2,3$:
\begin{equation}\label{intro:eq7}
\frac{du_j}{dt_{k+1,1}}
=
\frac{du_j}{dt_{k,2}}
=
\frac{du_j}{dt_{k-1,3}}
=
\frac12 
\sum_{i_1,\dots,i_{k}}
u_{i_1}\dots u_{i_{k}} [R(u^{i_1}\circ\dots\circ u^{i_{k}}),u_j]
\,.
\end{equation}

In the present paper we start by giving a self-contained exposition of the Oevel-Ragnisco (OR) theory
(see Sections \ref{sec:2}-\ref{sec:3}). 
In the subsequent Sections \ref{sec:4.1}-\ref{sec:ham} we extend this theory to the
infinite-dimensional case of Hamiltonian PDE.
(Some authors applied the theory to the infinite-dimensional case without any justification,
see e.g. \cite{BM94} and \cite{KO93}.)

We consider again a unital associative algebra $A$ with a non-degenerate trace form, but the relevant $R$-matrix
is over the associative algebra $\mf g=\mc F((\partial^{-1}))\otimes A$, where $\mc F$ is a differential
algebra over $\mb F$ of ``test functions'' and $\mc F((\partial^{-1}))$ is the algebra of pseudodifferential operators with
coefficient in $\mc F$. We require only the existence of a linear functional $\int:\mc F\to\mb F$, which
vanishes on $\partial\mc F$ and such that the bilinear form $\tint f g$ is non-degenerate in $\mc F$.
The three most important examples of $R$-matrices
are special cases of \eqref{intro:eq2} (for an infinite-dimensional $\mf g$).
Namely, for $k\in\mb Z_{\geq0}$ we consider the direct sum decomposition as left
$\mc F$-modules
$$
\mf g=\left(\mc F[\partial]\partial^k\otimes A\right)
\oplus
\left(\mc F[[\partial^{-1}]]\partial^{k-1}\otimes A\right)
\,,
$$
and denote by $\Pi_+$ and $\Pi_-$ the projections on the first and second summand respectively. The first
summand is always an associative subalgebra, but the second summand is an associative subalgebra only
for $k=0,\,1$, and a Lie subalgebra for $k=2$, provided that $A$ is commutative. 
The corresponding $R$-matrices are denoted by $R^{(0)}$,
$R^{(1)}$ and $R^{(2)}$ respectively. 

Fix dual bases $\{E_\alpha\}_{\alpha\in I}$ and $\{E^\alpha\}_{\alpha\in I}$ of $A$ with respect to the trace form,
and let, for $N\in\mb Z$, $\mc V_N$ be the algebra over $\mb F$ of differential polynomials in the indeterminates
$u_{p,\alpha}^{(n)}$, where $n\in\mb Z_{\geq0}$, $p\geq-N-1$, $\alpha\in I$,
with derivation $\partial$ defined by $\partial u_{p,\alpha}^{(n)}=u_{p,\alpha}^{(n+1)}$.
We denote by $\mc V_\infty\yhat$ the inverse limit 
$\displaystyle\lim_{\substack{\longleftarrow \\ N}}\mc V_N$
for the homomorphisms $\pi_N:\mc V_N\to\mc V_{N-1}$ defined by
$\pi_N(u_{-N-1,\alpha}^{(n)})=0$ and $\pi_N(u_{p,\alpha}^{(n)})=u_{p,\alpha}^{(n)}$
for $p>-N-1$.

In order to develop an infinite-dimensional version of the OR theory we need to introduce the notion of a 
continuous Poisson vertex algebra (PVA) structure on $\mc V_\infty\yhat$. Besides the usual
properties of a PVA $\lambda$-bracket $\{\cdot\,_\lambda\,\cdot\}$ on  $\mc V_\infty\yhat$
it should satisfy the continuity property:
for every $N\in\mb Z$ there exists sufficiently large $M\in\mb Z$ such that $\pi_M(f)=0$
or $\pi_M(g)=0$ imply that $\pi_N\{f_\lambda g\}=0$ (the Jacobi identity makes sense only under the assumption
of continuity).

The differential algebra $\mc V_\infty\yhat$ plays the same role in the ``affine'' OR theory as $S(\mf g)$
plays in the finite-dimensional OR theory. By analogy with the latter we write down formula
\eqref{eq:black11} for the bracket on the space $\mc V_\infty\yhat/\partial\mc V_\infty\yhat$,
which leads to formula \eqref{eq:bblack}, which defines the corresponding $\lambda$-bracket on $\mc V_\infty\yhat$.

We encode all the differential variables $u_{p,\alpha}$ of $\mc V_\infty\yhat$ in a generating series
$$
L(z)=\sum_{p\in\mb Z,\alpha\in I}u_{p,\alpha}z^{-p-1}\otimes E^\alpha
\in\mc V_\infty\yhat[[z,z^{-1}]]\otimes A
\,,
$$
and deduce from \eqref{eq:bblack} the following explicit formula for the $\lambda$-bracket
(cf. \eqref{20180414:eq4})
\begin{equation}\label{20180414:eq4-intro}
\begin{array}{l}
\displaystyle{
\vphantom{\Big(}
\{L_1(z)_\lambda L_2(w)\}^{R,\epsilon}
} \\
\displaystyle{
\vphantom{\Big(}
=
\frac12
\Omega
\bigg(
(L_1(w\!+\!\lambda\!+\!\partial)\!+\!\epsilon\id)
\big(\big|_{\zeta=z+\partial}\!L_1(z)\big)
R_\zeta\big(
\delta(\zeta\!-\!\xi)\big)
(\big|_{\xi=\zeta\!-\!z+\!w\!+\!\lambda\!+\!\partial}
L_2(w)\!+\!\epsilon\id)
} \\
\displaystyle{
\vphantom{\Big(}
-
(L_1(w+\lambda+\partial)+\epsilon\id)
R_z\big(
\delta(z-\xi)\big)
\big|_{\xi=w+\lambda+\partial}
L_2^*(\lambda-z)
(L_2(w)+\epsilon\id)
} \\
\displaystyle{
\vphantom{\Big(}
+
L_1(w+\lambda+\partial)
(L_1(z)+\epsilon\id)
R_w\big(
\delta(\zeta-w)\big)
\big(\big|_{\zeta=z-\lambda-\partial}
L_2^*(\lambda-z)+\epsilon\id\big)
} \\
\displaystyle{
\vphantom{\Big(}
-
(L_1(z)+\epsilon\id)
R_\xi\big(
\delta(\zeta-\xi)\big)
\big(\big|_{\zeta=z-\lambda-\partial}
L_2^*(\lambda-z)+\epsilon\id\big)
\big|_{\xi=w+\partial}
L_2(w)
\bigg)
\,,}
\end{array}
\end{equation}
where $\Omega=\sum_{\alpha\in I}E^\alpha\otimes E_\alpha$ and, 
as usual, $L_1(z)=L(z)\otimes\id$ and $L_2(w)=\id\otimes L(w)$.
Also, by $R_z(\delta(z-w))$ we denote the symbol of the pseudodifferential operator $R(\delta(\partial-w))$.
We call equation \eqref{20180414:eq4-intro} 
the $\epsilon$-Adler identity associated to the $R$-matrix $R$ since the 
coefficient of $2\epsilon$ in $\{\cdot\,_\lambda\,\}^{R,\epsilon}$ for $R=R^{(0)}$ and $A=gl_N$
coincides with the Adler identity for $gl_N$, which appeared in \cite{DSKV16} and \cite{DSKV17}.

We prove that, for $R=R^{(0)}$, $R^{(1)}$ and $R^{(2)}$, 
the $\epsilon$-Adler identity defines a 
continuous PVA $\lambda$-brackets on $\mc V_\infty\yhat$ and,
for $R=R^{(0)}$, it defines three compatible continuous PVA $\lambda$-brackets on $\mc V_\infty\yhat$
(see Theorem \ref{thm:main1}).

Due to the $\epsilon$-Adler identities, the elements $h_0=0$,
$h_n=-\frac{1}{n}\Res\Tr L^{n}(z)dz$, $n\in\mb Z_{\geq1}$,
are densities of Hamiltonian functionals in involution, defining a compatible hierarchy
of Hamiltonian PDE, satisfying the relations
\begin{equation}\label{intro:eq8}
\frac{dL(w)}{dt_{n}}
=
\{\tint h_{n},L(w)\}^{R,\epsilon}
=\frac12
[R((L+\epsilon\id)\circ L^{n-1}\circ (L+\epsilon\id)),L](w)
\,,\,\,n\in\mb Z_{\geq0},\,
\end{equation}
It follows that we have a triple Lenard-Magri relation
\begin{equation}\label{intro:eq9}
\begin{split}
\{\tint h_{n-1},L(w)\}^{R,3}&=\{\tint h_{n},L(w)\}^{R,2}
=\{\tint h_{n+1},L(w)\}^{R,1}
\\
&=\frac12[R (L^n),L](w)
\,,\quad
n\in\mb Z_{\geq1}\,.
\end{split}
\end{equation}
Equation \eqref{intro:eq9} is the affine analogue of equation \eqref{intro:eq7}.

\smallskip

Throughout the paper all vector spaces, Hom's and tensor products
are over a base field $\mb F$ of characteristic zero.

\subsubsection*{Acknowledgments} 
A. De Sole  was supported  by  the national PRIN project n. 2017YRA3LK.
V. G. Kac was supported by the Simons collaboration grant and the Bert and Ann Kostant fund.
A. De Sole and D. Valeri acknowledge the financial support of the
project MMNLP (Mathematical Methods in Non Linear Physics) of the
INFN.

\section{Classical $R$-matrix over a Lie algebra}\label{sec:2}

\begin{definition}[\cite{STS83}]\label{20180407:def}
A (classical) \emph{R-matrix} over a Lie algebra $\mf g$ is an endomorphism $R\in\End(\mf g)$
satisfying the \emph{modified Yang-Baxter equation}:
\begin{equation}\label{eq:mod-YB}
[R(a),R(b)]-R([R(a),b])-R([a,R(b)])+[a,b]=0
\,.
\end{equation}
\end{definition}
\begin{example}[\cite{STS83}]\label{ex:R}
Suppose that we have a direct sum decomposition (as vector spaces)
$\mf g=\mf g_+\oplus\mf g_-$,
where $\mf g_{\pm}\subset\mf g$ are two subalgebras of the Lie algebra $\mf g$,
and denote by $\Pi_{\pm}:\,\mf g\twoheadrightarrow\mf g_{\pm}$
the corresponding projections.
Then,
\begin{equation}\label{eq:R}
R:=\Pi_+-\Pi_-
\,,
\end{equation}
is an $R$-matrix over $\mf g$.
Indeed, denoting $a_{\pm}=\Pi_{\pm}(a)$ and $b_{\pm}=\Pi_{\pm}(b)$, we have
$$
[R(a),b]+[a,R(b)]
=
2[a_+,b_+]-2[a_-,b_-]
\,,
$$
so that
\begin{equation}\label{20180409:eqx}
R([R(a),b]+[a,R(b)])
=
2[a_+,b_+]+2[a_-,b_-]
\,.
\end{equation}
On the other hand,
\begin{equation}\label{20180409:eqy}
[R(a),R(b)]
=
[a_+,b_+]-[a_+,b_-]-[a_-,b_+]+[a_-,b_-]
\,.
\end{equation}
Equation \eqref{eq:mod-YB} follows immediately from \eqref{20180409:eqx} and \eqref{20180409:eqy}.
\end{example}
\begin{lemma}[\cite{STS83}]\label{20180405:rem1}
If $R$ is an $R$-matrix over the Lie algebra $\mf g$, then
\begin{equation}\label{20180407:eq1}
[a,b]_R:=[R(a),b]+[a,R(b)]
\end{equation}
is a Lie algebra bracket on $\mf g$.
\end{lemma}
\begin{proof}
The bracket \eqref{20180407:eq1} is obviously skewsymmetric,
so we only need to prove the Jacobi identity,
i.e. that the sum over cyclic permutations of 
$[a,[b,c]_R]_R$ vanishes.
By equation \eqref{eq:mod-YB} we have
$$
\begin{array}{l}
\displaystyle{
\vphantom{\Big(}
[a,[b,c]_R]_R
=
[R(a),[R(b),c]]+[R(a),[b,R(c)]]+
\big(
[a,R([R(b),c])+R([b,R(c)])]
\big)
} \\
\displaystyle{
\vphantom{\Big(}
=
[R(a),[R(b),c]]+[R(a),[b,R(c)]]+
\big(
[a,[R(b),R(c)]]+[a,[b,c]]
\big)
\,,}
\end{array}
$$
and the sum over cyclic permutations of the above sum is zero
due to the Jacobi identity for the commutator in $\mf g$.
\end{proof}
For example, if $R=\id_{\mf g}$, we recover the original Lie bracket of $\mf g$ multiplied by the factor $2$.
\begin{lemma}\label{20180409:lem1}
Let $\langle\cdot\,|\,\cdot\rangle$ be a non-degenerate, symmetric, invariant 
bilinear form on the Lie algebra $\mf g$,
let $R\in\End (\mf g)$ be an $R$-matrix over $\mf g$,
and let $R^*$ the adjoint of $R$ with respect to $\langle\cdot\,|\,\cdot\rangle$.
\begin{enumerate}[(a)]
\item
Each of the following identities is equivalent to the 
modified Yang-Baxter equation \eqref{eq:mod-YB} for $R$:
\begin{equation}\label{eq:mod-YB-a}
[R(a),R^*(b)]-R^*([R(a),b])+R^*([a,R^*(b)])-[a,b]=0
\,,
\end{equation}
and 
\begin{equation}\label{eq:mod-YB-b}
[R^*(a),R(b)]+R^*([R^*(a),b])-R^*([a,R(b)])-[a,b]=0
\,.
\end{equation}
\item
The antisymmetric part $\frac12(R-R^*)\in\End(\mf g)$ is also an $R$-matrix over $\mf g$
if and only if the following equation holds:
\begin{equation}\label{eq:mod-YB-c}
[R^*(a),R^*(b)]+R([R^*(a),b])+R([a,R^*(b)])+[a,b]=0
\,.
\end{equation}
\item
Equation \eqref{eq:mod-YB-c} implies the following identity:
\begin{equation}\label{eq:mod-YB-d}
\frac12 \sum_\sigma\sign(\sigma)
\langle [R^*(a),R^*(b)] | c \rangle
=
-\langle [a,b] | c \rangle
\,,
\end{equation}
where the sum is over all permutations of $a,b,c$ and $\sign(\sigma)$ is the sign of the permutation.
\end{enumerate}
\end{lemma}
In \eqref{eq:mod-YB-d}, and throughout the remainder of the paper,
in order to simplify notation,
we write 
\begin{equation}\label{eq:sums}
\sum_\sigma \sign(\sigma) f(a,b,c)\,,
\end{equation}
in place of $\sum_\sigma \sign(\sigma) f(\sigma(a),\sigma(b),\sigma(c))$.
\begin{proof}
Paring the modified Yang-Baxter equation \eqref{eq:mod-YB} with $c\in\mf g$,
we get
\begin{equation}\label{20180409:eq1}
\langle [R(a),R(b)] | c \rangle
-
\langle R([R(a),b]) | c \rangle
-
\langle R([a,R(b)]) | c \rangle
+
\langle [a,b] | c \rangle
=
0
\,.
\end{equation}
By the definition of $R^*$ and the invariance of the 
bilinear form $\langle\cdot\,|\,\cdot\rangle$, we have
$$
\begin{array}{l}
\displaystyle{
\vphantom{\Big(}
\langle [R(a),R(b)] | c \rangle
=
\langle a | R^*([R(b),c]) \rangle
\,,} \\
\displaystyle{
\vphantom{\Big(}
\langle R([R(a),b]) | c \rangle
=
\langle a | R^*([b,R^*(c)]) \rangle
\,} \\
\displaystyle{
\vphantom{\Big(}
\langle R([a,R(b)]) | c \rangle
=
\langle a | [R(b),R^*(c)] \rangle
\,} \\
\displaystyle{
\vphantom{\Big(}
\langle [a,b] | c \rangle
=
\langle a | [b,c] \rangle
\,.}
\end{array}
$$
Hence, \eqref{20180409:eq1} gives
$$
\langle a | R^*([R(b),c]) \rangle
-
\langle a | R^*([b,R^*(c)]) \rangle
-
\langle a | [R(b),R^*(c)] \rangle
+
\langle a | [b,c] \rangle
=
0
\,,
$$
which is the same as \eqref{eq:mod-YB-a},
with $b$ and $c$ in place of $a$ and $b$.
Equation \eqref{eq:mod-YB-b}
is obtained from \eqref{eq:mod-YB-a}
by exchanging the roles of $a$ and $b$ and using skewsymmetry.
This proves part (a).
Let us prove part (b).
Writing the modified Yang-Baxter equation \eqref{eq:mod-YB}
for the operator $\frac12(R-R^*)$, we get
\begin{equation}\label{20180409:eq2}
\begin{array}{l}
\displaystyle{
\vphantom{\Big(}
[R(a),R(b)]-[R(a),R^*(b)]-[R^*(a),R(b)]+[R^*(a),R^*(b)]
} \\
\displaystyle{
\vphantom{\Big(}
-R([R(a),b])+R^*([R(a),b])-R^*([R^*(a),b])+R([R^*(a),b])
} \\
\displaystyle{
\vphantom{\Big(}
-R([a,R(b)])-R^*([a,R^*(b)])+R^*([a,R(b)])+R([a,R^*(b)])
} \\
\displaystyle{
\vphantom{\Big(}
+4[a,b]=0
\,.}
\end{array}
\end{equation}
Hence,
in view of equations \eqref{eq:mod-YB}, \eqref{eq:mod-YB-a} and \eqref{eq:mod-YB-b},
equation \eqref{20180409:eq2} reduces to \eqref{eq:mod-YB-c},
proving (b).
Finally, we prove part (c).
We have, by the definition of $R^*$ and the invariance of the 
bilinear form $\langle\cdot\,|\,\cdot\rangle$,
$$
\begin{array}{l}
\displaystyle{
\vphantom{\Big(}
\frac12 \sum_\sigma\sign(\sigma)
\langle [R^*(a),R^*(b)] | c \rangle
} \\
\displaystyle{
\vphantom{\Big(}
=
\langle [R^*(a),R^*(b)] | c \rangle
+
\langle [R^*(b),R^*(c)] | a \rangle
+
\langle [R^*(c),R^*(a)] | b \rangle
} \\
\displaystyle{
\vphantom{\Big(}
=
\langle [R^*(a),R^*(b)] | c \rangle
+
\langle R([a,R^*(b)]) | c \rangle
+
\langle R([R^*(a),b]) | c \rangle
} \\
\displaystyle{
\vphantom{\Big(}
=
-\langle [a,b] | c \rangle
\,.}
\end{array}
$$
For the last equality we used \eqref{eq:mod-YB-c}.
\end{proof}

\section{Oevel-Ragnisco Poisson structures 
for finite dimensional associative algebras}\label{sec:3}

\subsection{The O-R construction}\label{sec:3.1}

Let $\mf g$ be a 
finite dimensional unital associative algebra, 
with unity $\id$,
associative product $\circ$, and the Lie bracket $[\cdot\,,\,\cdot]$ given by the commutator:
$[a,b]=a\circ b-b\circ a$.
Recall that a \emph{trace form} on $\mf g$ is a linear function
$\Tr(\cdot):\,\mf g\to\mb F$,
vanishing on commutators: $\Tr([a,b])=0$ for all $a,b\in\mf g$, 
and non-degenerate, 
in the sense that $\Tr(a\circ b)=0$ for all $b\in\mf g$
implies that $a=0$.
Any trace form has the cyclic property:
\begin{equation}\label{eq:cyclic}
\Tr(a_1\circ a_2\circ\dots\circ a_k)
=
\Tr(a_k\circ a_1\circ\dots\circ a_{k-1})
\,,
\end{equation}
hence the invariance property:
\begin{equation}\label{eq:invariance}
\Tr(a\circ[b,c])
=
\Tr([a,b]\circ c)
\,.
\end{equation}
We have the corresponding non-degenerate, symmetric, invariant bilinear form on $\mf g$
\begin{equation}\label{20180412:eq3}
\langle a|b \rangle=\Tr(a\circ b)
\,.
\end{equation}

The following construction is due independently to Oevel and Ragnisco \cite{OR89} 
and to Lie and Parmentier \cite{LP89}.
\begin{theorem}\label{20180408:thm}
Assume that $R\in\End(\mf g)$ is an $R$-matrix on the algebra $\mf g$.
Then, we have an $\epsilon$-family ($\epsilon\in\mb F$) of Poisson brackets on $S(\mf g)$,
defined by the following Lie brackets on $\mf g$ with values in $S(\mf g)$,
extended to $S(\mf g)$ by the Leibniz rules
($a,b,c\in\mf g$):
\begin{equation}\label{20180408:eq1}
\{a,b\}^{R,\epsilon}
=
\frac12 \sum_{i,j\in I}
(u^i+\epsilon\Tr(u^i))(u^j+\epsilon\Tr(u^j))
\big(
[a,R(u_i\circ b\circ u_j)]-[b,R(u_i\circ a\circ u_j)]
\big)
\,,
\end{equation}
where $\{u_i\}_{i\in I}$, $\{u^i\}_{i\in I}$ are bases of $\mf g$
dual with respect to the inner product $\langle\cdot\,|\,\cdot\rangle$ in \eqref{20180412:eq3},
i.e. such that $\Tr(u_i\circ u^j)=\delta_{i,j}$.
The $\epsilon$-family of Poisson brackets $\{\cdot\,,\,\cdot\}^{R,\epsilon}$
has the expansion
\begin{equation}\label{20180408:eq2}
\{\cdot\,,\,\cdot\}^{R,\epsilon}
=
\{\cdot\,,\,\cdot\}^{R}_3+2\epsilon \{\cdot\,,\,\cdot\}^{R}_2+\epsilon^2\{\cdot\,,\,\cdot\}^{R}_1
\,,
\end{equation}
where $\{\cdot\,,\,\cdot\}^{R}_i$, $i=1,2,3$ are the following brackets on $\mf g$
with values in $S(\mf g)$ ($a,b,c\in\mf g$):
\begin{equation}\label{20180408:eq3}
\begin{array}{l}
\displaystyle{
\vphantom{\Big(}
\{a,b\}^{R}_3
=
\frac12 \sum_{i,j\in I}
u^iu^j
\big(
[a,R(u_i\circ b\circ u_j)]-[b,R(u_i\circ a\circ u_j)]
\big)
\,,} \\
\displaystyle{
\vphantom{\Big(}
\{a,b\}^{R}_2
=
\frac14 \sum_{i\in I}
u^i
\big(
[a,R(u_i\circ b+b\circ u_i)]-[b,R(u_i\circ a+a\circ u_i)]
\big)
\,,} \\
\displaystyle{
\vphantom{\Big(}
\{a,b\}^{R}_1
=
\frac12 
\big(
[a,R(b)]-[b,R(a)]
\big)
=
\frac12[a,b]_R
\,.}
\end{array}
\end{equation}
The $1$-st and $3$-rd brackets $\{\cdot\,,\,\cdot\}^{R}_1$ and $\{\cdot\,,\,\cdot\}^{R}_3$
are Lie brackets, i.e. they define Poisson algebra structures on $S(\mf g)$.
If, moreover, $\frac12(R-R^*)$ 
is also an $R$-matrix on $\mf g$,
then the $2$-nd bracket $\{\cdot\,,\,\cdot\}^{R}_2$
is also a Lie bracket,
and all three brackets $\{\cdot\,,\,\cdot\}^{R}_i$, $i=1,2,3$, are compatible,
in the sense that any their linear combination is also a Lie bracket.
\end{theorem}
For the proof of Theorem \ref{20180408:thm}
we shall use the following elementary results:
\begin{lemma}\label{20180409:lem2}
The following identities hold,
for every $a,b,c,x\in\mf g$, 
\begin{enumerate}[(a)]
\item
$\displaystyle{
\frac12 \sum_\sigma \sign(\sigma) \Tr( [x\circ a\circ x,x\circ b\circ x] \circ [x,c] ) =0
}$,
\item
$\displaystyle{
\frac12 \sum_\sigma \sign(\sigma) \Tr( [x\!\circ\! a+a\!\circ\! x,x\!\circ\! b+b\!\circ\! x] \circ [x,c] )
=
- \Tr( [[x,a],[x,b]] \circ [x,c] )
}$,
\end{enumerate}
where we are using the notation \eqref{eq:sums}.
\end{lemma}
\begin{proof}
Both claims are straightforward. 
We provide here a proof for pedagogical reasons.
For claim (a), we have
\begin{equation}\label{20180409:eq3}
\begin{array}{l}
\displaystyle{
\vphantom{\Big(}
\frac12 \sum_\sigma \sign(\sigma) \Tr( [x\circ a\circ x,x\circ b\circ x] \circ [x,c] )
=
\Tr( [x\circ a\circ x,x\circ b\circ x] \circ [x,c] )
} \\
\displaystyle{
\vphantom{\Big(}
+
\Tr( [x\circ b\circ x,x\circ c\circ x] \circ [x,a] )
+
\Tr( [x\circ c\circ x,x\circ a\circ x] \circ [x,b] )
\,.}
\end{array}
\end{equation}
By the invariance of the trace \eqref{eq:invariance},
the second term in the RHS of \eqref{20180409:eq3}
is
\begin{equation}\label{20180409:eq4}
\Tr( [x\circ b\circ x,x\circ c\circ x] \circ [x,a] )
=
\Tr( [[x,a],x\circ b\circ x] \circ x\circ c\circ x )
\,,
\end{equation}
while the third term in the RHS of \eqref{20180409:eq3}
is
\begin{equation}\label{20180409:eq5}
\Tr( [x\circ c\circ x,x\circ a\circ x] \circ [x,b] )
=
\Tr( [x\circ a\circ x,[x,b]] \circ x\circ c\circ x )
\,.
\end{equation}
Combining \eqref{20180409:eq4} and \eqref{20180409:eq5}, we get
\begin{equation}\label{20180409:eq6}
\begin{array}{l}
\displaystyle{
\vphantom{\Big(}
\Tr\Big(
\big(
[[x,a],x\circ b\circ x] 
+
[x\circ a\circ x,[x,b]]
\big)\circ
x\circ c\circ x 
\Big)
} \\
\displaystyle{
\vphantom{\Big(}
=
\Tr\big(
[x,a\circ x\circ x\circ b-b\circ x\circ x\circ a]
\circ
x\circ c\circ x
\big)
} \\
\displaystyle{
\vphantom{\Big(}
=
-\Tr\big(
(a\circ x\circ x\circ b-b\circ x\circ x\circ a)
\circ
x\circ [x,c]\circ x
\big)
} \\
\displaystyle{
\vphantom{\Big(}
=
-\Tr\big(
[x\circ a\circ x,x\circ b\circ x]
\circ
[x,c]
\big)
\,.}
\end{array}
\end{equation}
Hence, the RHS of \eqref{20180409:eq3} vanishes, 
proving claim (a).
Similarly, we have
\begin{equation}\label{20180409:eq7}
\begin{array}{l}
\displaystyle{
\vphantom{\Big(}
\frac12 \sum_\sigma \sign(\sigma) 
\Tr( [x\!\circ\! a+a\!\circ\! x,x\!\circ\! b+a\!\circ\! x] \circ [x,c] ) 
} \\
\displaystyle{
\vphantom{\Big(}
=
\Tr( [x\!\circ\! a+a\!\circ\! x,x\!\circ\! b+b\!\circ\! x] \circ [x,c] )
+
\Tr( [x\!\circ\! b+b\!\circ\! x,x\!\circ\! c+c\!\circ\! x] \circ [x,a] ) 
} \\
\displaystyle{
\vphantom{\Big(}
+
\Tr( [x\!\circ\! c+c\!\circ\! x,x\!\circ\! a+a\!\circ\! x] \circ [x,b] )
\,.}
\end{array}
\end{equation}
By the invariance of the trace \eqref{eq:invariance}, the second term in the RHS of \eqref{20180409:eq7}
is
\begin{equation}\label{20180409:eq8}
\Tr( [x\!\circ\! b+b\!\circ\! x,x\!\circ\! c+c\!\circ\! x] \circ [x,a] )
=
\Tr( [[x,a],x\!\circ\! b+b\!\circ\! x] \circ (x\!\circ\! c+c\!\circ\! x) )
\,,
\end{equation}
while the third term in the RHS of \eqref{20180409:eq7}
is
\begin{equation}\label{20180409:eq9}
\Tr( [x\!\circ\! c+c\!\circ\! x,x\!\circ\! a+a\!\circ\! x] \circ [x,b] )
=
\Tr( [x\!\circ\! a+a\!\circ\! x,[x,b]] \circ (x\!\circ\! c+c\!\circ\! x) )
\,.
\end{equation}
Combining \eqref{20180409:eq8} and \eqref{20180409:eq9}, we get,
again by \eqref{eq:invariance},
\begin{equation}\label{20180409:eq10}
\begin{array}{l}
\displaystyle{
\vphantom{\Big(}
\Tr\Big(
\big(
[[x,a],x\!\circ\! b+b\!\circ\! x]
+
[x\!\circ\! a+a\!\circ\! x,[x,b]] 
\big)
\circ
(x\!\circ\! c+c\!\circ\! x)
\Big)
} \\
\displaystyle{
\vphantom{\Big(}
=
2\Tr\big(
[x,a\circ x\circ b-b\circ x\circ a]
\circ
(x\!\circ\! c+c\!\circ\! x)
\big)
} \\
\displaystyle{
\vphantom{\Big(}
=
-2\Tr\big(
(a\circ x\circ b-b\circ x\circ a)
\circ
(x\circ[x,c]+[x,c]\circ x)
\big)
} \\
\displaystyle{
\vphantom{\Big(}
=
-2\Tr\big(
(
x\circ a\circ x\circ b+a\circ x\circ b\circ x
-x\circ b\circ x\circ a-b\circ x\circ a\circ x
)
\circ
[x,c]
\big)
\,.}
\end{array}
\end{equation}
Finally, combining the first term in the RHS of \eqref{20180409:eq7}
and \eqref{20180409:eq10}, we get
$$
-\Tr(
[[x,a],[x,b]]
\circ
[x,c]
)
\,,
$$
proving claim (b).
\end{proof}
\begin{proof}[Proof of Theorem \ref{20180408:thm}]
The bracket $\{\cdot\,,\,\cdot\}^{R,\epsilon}$ is defined on $S(\mf g)$
by its value \eqref{20180408:eq1} on $a,b\in\mf g$, 
and it is extended to $S(\mf g)$ by the left and right Leibniz rules.
It is well known that, in this case, in order to prove the Lie algebra axioms 
for $\{\cdot\,,\,\cdot\}^{R,\epsilon}$,
it is enough to prove the skew-symmetry 
\begin{equation}\label{20180408:eq4}
\{a,b\}^{R,\epsilon}=-\{b,a\}^{R,\epsilon}
\,,
\end{equation}
and the Jacobi identity 
\begin{equation}\label{20180408:eq5}
\{a,\{b,c\}^{R,\epsilon}\}^{R,\epsilon}
+\text{ cycl. perm.'s } = 0
\,,
\end{equation}
on elements $a,b,c\in\mf g$.
The skew-symmetry \eqref{20180408:eq4}
is obvious by the definition \eqref{20180408:eq1}.
Let us prove the Jacobi identity \eqref{20180408:eq5}.
By \eqref{20180408:eq1} we have
\begin{equation}\label{20180408:eq7}
\begin{array}{l}
\displaystyle{
\vphantom{\Big(}
\{a,\{b,c\}^{R,\epsilon}\}^{R,\epsilon}
+
\text{ cycl. perm.'s }
} \\
\displaystyle{
\vphantom{\Big(}
=
\frac12 \sum_{h,k\in I}
\Big\{a,
(u^h + \epsilon\Tr(u^h))(u^k + \epsilon\Tr(u^k))
\Big(
[b,R(u_h \circ c \circ u_k)]
} \\
\displaystyle{
\vphantom{\Big(}
-
[c,R(u_h \circ b \circ u_k)]
\Big)
\Big\}^{R,\epsilon}
+
\text{ cycl. perm.'s }
} \\
\displaystyle{
\vphantom{\Big(}
=
\sum_{\sigma} \sign(\sigma)
\frac12 \sum_{h,k\in I}
\Big\{a,
(u^h+\epsilon\Tr(u^h))(u^k+\epsilon\Tr(u^k))
[b,R(u_h\circ c\circ u_k)]
\Big\}^{R,\epsilon}
\,.}
\end{array}
\end{equation}
In the RHS of \eqref{20180408:eq7} we are using the notation \eqref{eq:sums}.
We compute the bracket on the RHS of \eqref{20180408:eq7}
applying the Leibniz rule
and using \eqref{20180408:eq1}.
As a result we get
\begin{equation}\label{20180408:eq9}
\begin{array}{l}
\displaystyle{
\vphantom{\Big(}
\sum_{\sigma} \sign(\sigma)
\frac14 \sum_{i,j,h,k\in I}
(u^i+\epsilon\Tr(u^i))
(u^j+\epsilon\Tr(u^j))
} \\
\displaystyle{
\vphantom{\Big(}
\Bigg(
(u^k+\epsilon\Tr(u^k))
[a,R(u_i\circ u^h\circ u_j)]
[b,R(u_h\circ c\circ u_k)]
} \\
\displaystyle{
\vphantom{\Big(}
-
(u^k+\epsilon\Tr(u^k))
[u^h,R(u_i\circ a\circ u_j)]
[b,R(u_h\circ c\circ u_k)]
} \\
\displaystyle{
\vphantom{\Big(}
+
(u^h+\epsilon\Tr(u^h))
[a,R(u_i\circ u^k\circ u_j)]
[b,R(u_h\circ c\circ u_k)]
} \\
\displaystyle{
\vphantom{\Big(}
-
(u^h+\epsilon\Tr(u^h))
[u^k,R(u_i\circ a\circ u_j)]
[b,R(u_h\circ c\circ u_k)]
} \\
\displaystyle{
\vphantom{\Big(}
+
(u^h+\epsilon\Tr(u^h))(u^k+\epsilon\Tr(u^k))
[a,R(u_i \circ [b,R(u_h \circ c \circ u_k)] \circ u_j)]
} \\
\displaystyle{
\vphantom{\Big(}
-
(u^h+\epsilon\Tr(u^h))(u^k+\epsilon\Tr(u^k))
[[b,R(u_h \circ c \circ u_k)],R(u_i \circ a \circ u_j)]
\Bigg)
.}
\end{array}
\end{equation}
We need to prove that \eqref{20180408:eq9} vanishes.
Since it lies in the symmetric algebra $S(\mf g)\simeq\mb F[\mf g^*]$, 
in order to prove that \eqref{20180408:eq9} vanishes, 
it suffices to prove that it vanishes when evaluated at 
an arbitrary point $\Tr(x\circ\,\cdot\,)\in\mf g^*$.
By completeness, we have 
$$
\sum_{i\in I}\big(\Tr(u^i\circ x)+\epsilon\Tr(u^i)\big)u_i=x+\epsilon\id
\,.
$$
Hence, the vanishing of \eqref{20180408:eq9}
is equivalent to the vanishing of
\begin{equation}\label{20180408:eq10}
\begin{array}{l}
\displaystyle{
\vphantom{\Big(}
\sum_{\sigma} \sign(\sigma)
} \\
\displaystyle{
\vphantom{\Big(}
\Bigg(
\sum_{h\in I}
\Tr\big( x \circ [a,R((x+\epsilon\id)\circ u^h\circ (x+\epsilon\id))] \big)
\Tr\big( x \circ [b,R(u_h\circ c\circ (x+\epsilon\id))] \big)
} \\
\displaystyle{
\vphantom{\Big(}
-
\sum_{h\in I}
\Tr\big( x \circ [u^h,R((x+\epsilon\id)\circ a\circ (x+\epsilon\id))] \big)
\Tr\big( x \circ [b,R(u_h\circ c\circ (x+\epsilon\id))] \big)
} \\
\displaystyle{
\vphantom{\Big(}
+
\sum_{k\in I}
\Tr\big( x \circ [a,R((x+\epsilon\id)\circ u^k\circ (x+\epsilon\id))] \big)
\Tr\big( x \circ [b,R((x+\epsilon\id)\circ c\circ u_k)] \big)
} \\
\displaystyle{
\vphantom{\Big(}
-
\sum_{k\in I}
\Tr\big( x \circ [u^k,R((x+\epsilon\id)\circ a\circ (x+\epsilon\id))] \big)
\Tr\big( x \circ [b,R((x+\epsilon\id)\circ c\circ u_k)] \big)
} \\
\displaystyle{
\vphantom{\Big(}
+
\Tr\big( x \circ [a,R((x+\epsilon\id) \circ [b,R((x+\epsilon\id) \circ c \circ (x+\epsilon\id))] \circ (x+\epsilon\id))] \big)
} \\
\displaystyle{
\vphantom{\Big(}
-
\Tr\big( x \circ [[b,R((x+\epsilon\id) \circ c \circ (x+\epsilon\id))],R((x+\epsilon\id) \circ a \circ (x+\epsilon\id))] \big)
\Bigg)
\,,}
\end{array}
\end{equation}
for every $x\in\mf g$.
By definition that the trace form $\Tr(\,\cdot\,)$ vanishes
on the derived subalgebra $[\mf g,\mf g]$.
Hence, \eqref{20180408:eq10}
can be rewritten by letting $x+\epsilon\id$ appear everywhere,
and, in order to prove its vanishing for every $x\in\mf g$, 
we can set $\epsilon=0$.
Moreover, by the invariance of the trace \eqref{eq:invariance}, we have
$$
\Tr( x \circ [a,R(x\circ u^h\circ x)] )
=
\Tr( u^h \circ x\circ R^*([x,a])\circ x )
\,,
$$
and
$$
\Tr( x \circ [u^h,R(x\circ a\circ x)] )
=
-\Tr( u^h \circ [x,R(x \circ a\circ x)] )
\,.
$$
Hence, by completeness, the vanishing of \eqref{20180408:eq10} is equivalent to the vanishing of 
\begin{equation}\label{20180408:eq12}
\begin{array}{l}
\displaystyle{
\vphantom{\Big(}
\sum_{\sigma}\! \sign(\!\sigma\!)
} \\
\displaystyle{
\vphantom{\Big(}
\Bigg(
\Tr\big( x \!\circ\! [b,R(x\!\circ\! R^*([x,a])\!\circ\! x\!\circ\! c\!\circ\! x)] \big)
+
\Tr\big( x \!\circ\! [b,R([x,R(x\!\circ\! a\!\circ\! x)]\!\circ\! c\!\circ\! x)] \big)
} \\
\displaystyle{
\vphantom{\Big(}
+
\Tr\big( x \!\circ\! [b,R( x\!\circ\! c\!\circ\! x\!\circ\! R^*([x,a])\!\circ\! x )] \big)
+
\Tr\big( x \!\circ\! [b,R(x\!\circ\! c\!\circ\! [x,R(x\!\circ\! a\!\circ\! x)] )] \big)
} \\
\displaystyle{
\vphantom{\Big(}
+
\Tr\big( x \!\circ\! [a,R(x \!\circ\! [b,R(x \!\circ\! c \!\circ\! x)] \!\circ\! x)] \big)
-
\Tr\big( x \!\circ\! [[b,R(x \!\circ\! c \!\circ\! x)],R(x \!\circ\! a \!\circ\! x)] \big)
\Bigg)
\,.}
\end{array}
\end{equation}
The first summand in \eqref{20180408:eq12} is, by the invariance of the trace \eqref{eq:invariance},
\begin{equation}\label{20180408:eq13}
\Tr\big( x \circ [b,R(x\circ R^*([x,a])\circ x\circ c\circ x)] \big)
=
\Tr\big( [x,b] \circ R( x \circ R^*([x,a])\circ x\circ c\circ x ) \big)
\,,
\end{equation}
while the third summand in \eqref{20180408:eq12} is
\begin{equation}\label{20180408:eq14}
\begin{array}{l}
\displaystyle{
\vphantom{\Big(}
\Tr\big( x \!\circ\! [b,R( x\!\circ\! c\!\circ\! x\!\circ\! R^*([x,a])\!\circ\! x )] \big)
=
\Tr\big( [x,b] \!\circ\! R( x\!\circ\! c\!\circ\! x\!\circ\! R^*([x,a])\!\circ\! x ) \big)
} \\
\displaystyle{
\vphantom{\Big(}
=
\Tr\big( x \!\circ\! R^*([x,b])\!\circ\! x\!\circ\! c\!\circ\! x \!\circ\! R^*([x,a]) \big)
=
\Tr\big( [x,a] \!\circ\! R( x \!\circ\! R^*([x,b])\!\circ\! x\!\circ\! c\!\circ\! x ) \big)
\,.}
\end{array}
\end{equation}
Since \eqref{20180408:eq12} and \eqref{20180408:eq14}
are obtained one from another by exchanging $a$ with $b$,
their sum vanishes under the alternating sum over permutations.
Furthermore, the fifth summand in \eqref{20180408:eq12}
can be replaced, under the sum over permutation, by
$$
\Tr\big( x \circ [b,R(x \circ [c,R(x \circ a \circ x)] \circ x)] \big)
\,,
$$
and combining it with the second and fourth summands of \eqref{20180408:eq12},
we get
\begin{equation}\label{20180408:eq15}
\begin{array}{l}
\displaystyle{
\vphantom{\Big(}
\Tr\Big( [x,b] \circ
R\big( 
[x,R(x \!\circ\! a \!\circ\! x)]\circ c\circ x 
+
x\circ c\circ [x,R(x \!\circ\! a \!\circ\! x)]
} \\
\displaystyle{
\vphantom{\Big(}
+
x \circ [c,R(x \!\circ\! a \!\circ\! x)] \circ x
\big)
\Big)
=
\Tr\Big( [x,b] \circ
R\big( 
[x\circ c\circ x,R(x\circ a\circ x)] 
\big)
\Big)
} \\
\displaystyle{
\vphantom{\Big(}
=
\Tr\Big( x \circ \big[b,R\big([x\circ c\circ x,R(x\circ a\circ x)]\big)\big] \Big)
\,.}
\end{array}
\end{equation}
Under the alternating sum over permutations,
the RHS of \eqref{20180408:eq15} can be replaced by
\begin{equation}\label{20180408:eq15b}
\frac12
\Tr\Big( x \circ 
\big[b, - R\big([R(x\circ a\circ x),x\circ c\circ x]\big) -R\big([x\circ a\circ x,R(x\circ c\circ x)]\big) \big] \Big)
\,,
\end{equation}
while
the last summand in \eqref{20180408:eq12}
can be replaced by
\begin{equation}\label{20180408:eq16}
\begin{array}{l}
\displaystyle{
\vphantom{\Big(}
-
\frac12 
\Tr\big( x \!\circ\! [[b,R(x \!\circ\! c \!\circ\! x)],R(x \!\circ\! a \!\circ\! x)] \big)
+
\frac12
\Tr\big( x \!\circ\! [[b,R(x \!\circ\! a \!\circ\! x)],R(x \!\circ\! c \!\circ\! x)] \big)
} \\
\displaystyle{
\vphantom{\Big(}
=
\frac12 
\Tr\big( x \circ [b,[R(x \circ a \circ x),R(x \circ c \circ x)]] \big)
\,,}
\end{array}
\end{equation}
by the Jacobi identity.
Combining \eqref{20180408:eq15b} and \eqref{20180408:eq16},
we conclude that the vanishing of \eqref{20180408:eq12} is equivalent to the vanishing of
\begin{equation}\label{20180408:eq17}
\begin{array}{l}
\displaystyle{
\vphantom{\Big(}
\frac12
\sum_{\sigma} \sign(\sigma)
\Tr\Big( x \circ \Big[b, 
[R(x \circ a \circ x),R(x \circ c \circ x)]
-
R([R(x\circ a\circ x),x\circ c\circ x]) 
} \\
\displaystyle{
\vphantom{\Big(}
-
R([x\circ a\circ x,R(x\circ c\circ x)]) 
\Big] \Big)
} \\
\displaystyle{
\vphantom{\Big(}
=
-\frac12
\sum_{\sigma} \sign(\sigma)
\Tr\big( x \circ \big[b, 
[x \circ a \circ x,x \circ c \circ x]
\big] \big)
} \\
\displaystyle{
\vphantom{\Big(}
=
\frac12
\sum_{\sigma} \sign(\sigma)
\Tr\big(
[x \circ a \circ x,x \circ b \circ x] \circ [x,c]
\big)
\,.}
\end{array}
\end{equation}
For the first equality 
we used the modified Yang-Baxter equation \eqref{eq:mod-YB} on $R$.
By Lemma \ref{20180409:lem2}(a), the RHS of \eqref{20180408:eq17} vanishes,
proving the first claim of Theorem \ref{20180408:thm}.

Equations \eqref{20180408:eq2} and \eqref{20180408:eq3}
are immediately checked.
Letting $\epsilon=0$ in the Jacobi identity \eqref{20180408:eq5}
we get
$$
\{a,\{b,c\}^{R}_3\}^{R}_3
+\text{ cycl. perm.'s } = 0
\,,
$$
i.e. the $3$-rd bracket $\{\cdot\,,\,\cdot\}^R_3$ satisfies the Lie algebra axioms,
while taking the coefficient of $\epsilon^4$ in \eqref{20180408:eq5}
we get 
$$
\{a,\{b,c\}^{R}_1\}^{R}_1
+\text{ cycl. perm.'s } = 0
\,,
$$
i.e. the $1$-st bracket $\{\cdot\,,\,\cdot\}^R_1$
satisfies the Lie algebra axioms as well.
Moreover, taking the coefficient of $\epsilon$ and $\epsilon^3$ 
in the Jacobi identity \eqref{20180408:eq5},
we get 
$$
\{a,\{b,c\}^{R}_2\}^{R}_3+\{a,\{b,c\}^{R}_3\}^{R}_2
+\text{ cycl. perm.'s } = 0
\,,
$$
and
$$
\{a,\{b,c\}^{R}_1\}^{R}_2+\{a,\{b,c\}^{R}_2\}^{R}_1
+\text{ cycl. perm.'s } = 0
\,,
$$
i.e. the compatibility between the $2$-nd and $3$-rd brackets,
and between the $1$-st and $2$-nd brackets respectively
(provided that the $2$-nd bracket is a Lie algebra bracket).
Taking the coefficient of $\epsilon^2$ in \eqref{20180408:eq5},
we get
$$
\{a,\{b,c\}^{R}_1\}^{R}_3+\{a,\{b,c\}^{R}_3\}^{R}_1
+\{a,\{b,c\}^{R}_2\}^{R}_2
+\text{ cycl. perm.'s } = 0
\,,
$$
which shows that the $1$-st and $3$-rd Lie brackets are compatible
if and only if the $2$-nd bracket satisfies the Jacobi identity.

To complete the proof of the Theorem, we are left to prove the last assertion,
i.e. that the $2$-nd bracket satisfies the Jacobi identity
\begin{equation}\label{20180408:eq6}
\{a,\{b,c\}^{R}_2\}^{R}_2
+\text{ cycl. perm.'s } = 0
\,,
\end{equation}
provided that $\frac12(R-R^*)$ is an $R$-matrix over the Lie algebra $\mf g$.
By the definition \eqref{20180408:eq3} of the second bracket $\{\cdot\,,\,\cdot\}^R_2$,
and the Leibniz rules, we have
\begin{equation}\label{20180409:eq11}
\begin{array}{l}
\displaystyle{
\vphantom{\Big(}
\{a,\{b,c\}^{R}_2\}^{R}_2
+\text{ cycl. perm.'s }
} \\
\displaystyle{
\vphantom{\Big(}
=
\sum_\sigma\sign(\sigma)
\frac14 \sum_{j\in I}
\{a,
u^j
[b,R(u_j\circ c+c\circ u_j)]
\}^{R}_2
} \\
\displaystyle{
\vphantom{\Big(}
=
\sum_\sigma\sign(\sigma)
\frac1{16} \sum_{i,j\in I}
\Big(
u^i
[a,R(u_i\circ u^j+u^j\circ u_i)]
[b,R(u_j\circ c+c\circ u_j)]
} \\
\displaystyle{
\vphantom{\Big(}
-
u^i
[u^j,R(u_i\circ a+a\circ u_i)]
[b,R(u_j\circ c+c\circ u_j)]
} \\
\displaystyle{
\vphantom{\Big(}
+
u^iu^j
[a,R(u_i\circ [b,R(u_j\circ c+c\circ u_j)]+[b,R(u_j\circ c+c\circ u_j)]\circ u_i)]
} \\
\displaystyle{
\vphantom{\Big(}
-
u^iu^j
[[b,R(u_j\circ c+c\circ u_j)],R(u_i\circ a+a\circ u_i)]
\Big)
\,,}
\end{array}
\end{equation}
where, as before, we use the notation \eqref{eq:sums} for the alternating sums
over permutations of $a,b,c$.
As in \eqref{20180408:eq10},
in order to prove that \eqref{20180409:eq11} vanishes
we evaluate it at a generic point $\Tr( x \circ \,\cdot\, )\in\mf g^*$.
As a result, we get that the Jacobi identity \eqref{20180408:eq6}
is equivalent to the vanishing of
\begin{equation}\label{20180409:eq12}
\begin{array}{l}
\displaystyle{
\vphantom{\Big(}
\sum_\sigma\sign(\sigma)
\Big(
\sum_{j\in I}
\Tr\big( x \circ [a,R(x\circ u^j+u^j\circ x)] \big)
\Tr\big( x \circ [b,R(u_j\circ c+c\circ u_j)] \big)
} \\
\displaystyle{
\vphantom{\Big(}
-
\sum_{j\in I}
\Tr\big( x \circ [u^j,R(x\circ a+a\circ x)] \big)
\Tr\big( x \circ [b,R(u_j\circ c+c\circ u_j)] \big)
} \\
\displaystyle{
\vphantom{\Big(}
+
\Tr\Big( x \circ \Big[a,R\big(x\circ [b,R(x\circ c+c\circ x)]+[b,R(x\circ c+c\circ x)]\circ x\big) \Big] \Big)
} \\
\displaystyle{
\vphantom{\Big(}
-
\Tr\Big( x \circ \Big[[b,R(x\circ c+c\circ x)],R(x\circ a+a\circ x)\Big] \Big)
\Big)
\,.}
\end{array}
\end{equation}
By the invariance of the trace \eqref{eq:invariance}
and the completeness of the dual bases $\{u_i\}_{i\in I}$, $\{u^i\}_{i\in I}$,
the first summand in \eqref{20180409:eq12} is
\begin{equation}\label{20180409:eq13}
\begin{array}{l}
\displaystyle{
\vphantom{\Big(}
\sum_{j\in I}
\Tr\big( x \circ [a,R(x\circ u^j+u^j\circ x)] \big)
\Tr\big( x \circ [b,R(u_j\circ c+c\circ u_j)] \big)
} \\
\displaystyle{
\vphantom{\Big(}
=
\Tr\Big( R^*([x,b]) \circ
\big(
R^*([x,a]) \circ x \circ c
+
x \circ R^*([x,a]) \circ c
} \\
\displaystyle{
\vphantom{\Big(}
+
c \circ R^*([x,a]) \circ x
+
c \circ x \circ R^*([x,a])
\big)
\Big)
,}
\end{array}
\end{equation}
while 
the second summand in \eqref{20180409:eq12} is
\begin{equation}\label{20180409:eq14}
\begin{array}{l}
\displaystyle{
\vphantom{\Big(}
\sum_{j\in I}
\Tr\big( x \circ [u^j,R(x\circ a+a\circ x)] \big)
\Tr\big( x \circ [b,R(u_j\circ c+c\circ u_j)] \big)
} \\
\displaystyle{
\vphantom{\Big(}
=
-
\Tr\Big( R^*([x,b]) \circ 
\big(
[x,R(x\circ a+a\circ x)]\circ c+c\circ [x,R(x\circ a+a\circ x)] 
\big)
\Big)
\,.}
\end{array}
\end{equation}
Moreover, under the alternating sum over permutations,
we can replace the third summand in \eqref{20180409:eq12} by
\begin{equation}\label{20180409:eq15}
\Tr\Big(
R^*([x,b]) \circ 
\big( 
x\circ [c,R(x\circ a+a\circ x)]+[c,R(x\circ a+a\circ x)]\circ x 
\big)
\Big)
\,,
\end{equation}
and the fourth summand in \eqref{20180409:eq12} by
\begin{equation}\label{20180409:eq16}
\frac12
\Tr\big( [x,b] \circ [R(x\circ a+a\circ x),R(x\circ c+c\circ x)] \big) 
\,.
\end{equation}
(Here we used the Jacobi identity for the commutator $[\cdot\,,\,\cdot]$ on $\mf g$.)
Combining \eqref{20180409:eq12}, \eqref{20180409:eq13}, \eqref{20180409:eq14},
\eqref{20180409:eq15} and \eqref{20180409:eq16},
we get that the Jacobi identity \eqref{20180408:eq6}
is equivalent to the vanishing of
\begin{equation}\label{20180409:eq17}
\begin{array}{l}
\displaystyle{
\vphantom{\Big(}
\sum_\sigma\sign(\sigma)
\Bigg(
\Tr\Big( 
R^*([x,b]) 
\circ
} \\
\displaystyle{
\vphantom{\Big(}
\circ
\big(
R^*([x,a]) \!\circ\! x \!\circ\! c
+
x \!\circ\! R^*([x,a]) \!\circ\! c
+
c \!\circ\! R^*([x,a]) \!\circ\! x
+
c \!\circ\! x \!\circ\! R^*([x,a])
\big)
\Big)
} \\
\displaystyle{
\vphantom{\Big(}
+
\Tr\Big(
R^*([x,b]) \circ
\big(
[x,R(x\circ a+a\circ x)]\circ c+c\circ [x,R(x\circ a+a\circ x)]
\big)
\Big)
} \\
\displaystyle{
\vphantom{\Big(}
+
\Tr\Big(
R^*([x,b])
\circ
\big(
x\circ [c,R(x\circ a+a\circ x)]+[c,R(x\circ a+a\circ x)]\circ x
\big)
\Big)
} \\
\displaystyle{
\vphantom{\Big(}
+\frac12
\Tr\Big(
[x,b]
\circ 
[R(x\circ a+a\circ x),R(x\circ c+c\circ x)]
\Big)
\Bigg)
\,.}
\end{array}
\end{equation}
Note that
\begin{equation}\label{20180409:eq18}
\begin{array}{l}
\displaystyle{
\vphantom{\Big(}
\Tr\Big(
R^*([x,b]) 
\circ
\big(
x \circ R^*([x,a]) \circ c
+
c \circ R^*([x,a]) \circ x
\big)
\Big)
} \\
\displaystyle{
\vphantom{\Big(}
=
\Tr\Big(
R^*([x,b])
\circ
x \circ R^*([x,a]) \circ c
\Big)
+
\Tr\Big(
R^*([x,a])
\circ
x\circ R^*([x,b])\circ c
\Big)
\,,}
\end{array}
\end{equation}
and this expression is symmetric with respect to the exchange of $a$ and $b$,
hence it vanishes under the alternating sum over permutations.
Moreover, we have
\begin{equation}\label{20180409:eq19}
\begin{array}{l}
\displaystyle{
\vphantom{\Big(}
[x,R(x\circ a+a\circ x)]\circ c+c\circ [x,R(x\circ a+a\circ x)] 
} \\
\displaystyle{
\vphantom{\Big(}
+
x\circ [c,R(x\circ a+a\circ x)]+[c,R(x\circ a+a\circ x)]\circ x 
} \\
\displaystyle{
\vphantom{\Big(}
=
[x\circ c+c\circ x,R(x\circ a+a\circ x)]
\,.}
\end{array}
\end{equation}
Hence, by \eqref{20180409:eq18} and \eqref{20180409:eq19},
we can rewrite \eqref{20180409:eq17} as
\begin{equation}\label{20180409:eq20}
\begin{array}{l}
\displaystyle{
\vphantom{\Big(}
\sum_\sigma\sign(\sigma)
\bigg(
\Tr\Big( 
R^*([x,b]) 
\circ
\big(
R^*([x,a]) \!\circ\! x \!\circ\! c
+
c \!\circ\! x \!\circ\! R^*([x,a])
\big)
\Big)
} \\
\displaystyle{
\vphantom{\Big(}
+
\Tr\Big( 
R^*([x,b])
\circ
[x\circ c+c\circ x,R(x\circ a+a\circ x)]
\Big)
} \\
\displaystyle{
\vphantom{\Big(}
+\frac12
\Tr\Big( 
[x,b] 
\circ 
[R(x\circ a+a\circ x),R(x\circ c+c\circ x)]
\Big)
\bigg)
\,.}
\end{array}
\end{equation}
The first summand in \eqref{20180409:eq20} is
\begin{equation}\label{20180409:eq21}
\begin{array}{l}
\displaystyle{
\vphantom{\Big(}
\sum_\sigma\sign(\sigma)
\bigg(
\Tr\Big( 
R^*([x,b]) 
\circ
R^*([x,a]) \!\circ\! x \!\circ\! c
\Big)
+
\Tr\Big(
R^*([x,b]) 
\circ
c \!\circ\! x \!\circ\! R^*([x,a])
\Big)
\bigg)
} \\
\displaystyle{
\vphantom{\Big(}
=
\sum_\sigma\!\sign(\sigma)\!
\bigg(\!
\Tr\Big(
R^*([x,b]) 
\!\circ\!
R^*([x,a]) \!\circ\! x \!\circ\! c
\Big)
-
\Tr\Big(
R^*([x,a]) 
\!\circ\! c \!\circ\! x
\!\circ\!
R^*([x,b])
\Big)
\bigg)
} \\
\displaystyle{
\vphantom{\Big(}
=
\sum_\sigma\sign(\sigma)
\Tr\Big(
R^*([x,b]) 
\circ 
R^*([x,a]) 
\circ
[x,c]
\Big)
} \\
\displaystyle{
\vphantom{\Big(}
=
-\frac12\sum_\sigma\sign(\sigma)
\Tr\Big(
[R^*([x,a]),R^*([x,b])] 
\circ
[x,c]
\Big)
} \\
\displaystyle{
\vphantom{\Big(}
=
\Tr\big(
[[x,a],[x,b]] \circ [x,c]
\big)
\,.}
\end{array}
\end{equation}
For the last equality we used Lemma \ref{20180409:lem1}(c).
The second and third summands in \eqref{20180409:eq20} combined
give
\begin{equation}\label{20180409:eq22}
\begin{array}{l}
\displaystyle{
\vphantom{\Big(}
\sum_\sigma\sign(\sigma)
\Tr\Big( 
[x,b] 
\circ
\big(
R([x \circ c+c \circ x,R(x \circ a+a \circ x)])
} \\
\displaystyle{
\vphantom{\Big(}
\qquad
+\frac12
[R(x \circ a+a \circ x),R(x \circ c+c \circ x)] 
\big)
\Big)
} \\
\displaystyle{
\vphantom{\Big(}
=
\frac12
\sum_\sigma\sign(\sigma)
\Tr\Big(
[x,b]
\circ
\big(
[R(x \circ a+a \circ x),R(x \circ c+c \circ x)] 
} \\
\displaystyle{
\vphantom{\Big(}
\qquad
-
R([R(x \circ a+a \circ x),x \circ c+c \circ x])
-
R([x \circ a+a \circ x,R(x \circ c+c \circ x)])
\big)
\Big)
} \\
\displaystyle{
\vphantom{\Big(}
=
-\frac12
\sum_\sigma\sign(\sigma)
\Tr( [x,b] \circ 
[x \circ a+a \circ x,x \circ c+c \circ x] 
)
} \\
\displaystyle{
\vphantom{\Big(}
=
-\Tr( [[x,a],[x,b]] \circ [x,c] )
\,.}
\end{array}
\end{equation}
For the second equality we used the assumption \eqref{eq:mod-YB} on $R$,
while for the third equality we used Lemma \ref{20180409:lem2}(b).
Combining \eqref{20180409:eq21} and \eqref{20180409:eq22} we get $0$, 
proving the claim.
\end{proof}
\begin{remark}\label{20180408:rem1}
By \eqref{20180408:eq3},
the $1$-bracket $\{\cdot\,,\,\cdot\}^R_1$ coincides (up to a factor $\frac12$)
with the Lie bracket $[\cdot\,,\,\cdot]_R$ of $\mf g$ defined by Lemma \ref{20180405:rem1},
hence the associated Poisson bracket of $S(\mf g)$
corresponds to the Kirillov-Kostant Poisson structure on $\mf g^*$
with respect to Lie bracket $[\cdot\,,\,\cdot]_R$.
In particular, this structure only uses the Lie bracket of $\mf g$ and the $R$-matrix $R$,
and not the associative product of $\mf g$.
\end{remark}
\begin{remark}\label{20180408:rem2}
Identifying $S(\mf g)$ with the algebra of polynomial functions on $\mf g^*$,
we can write down the $\epsilon$-family of Poisson structures of $\mf g^*$,
corresponding to the Poisson brackets \eqref{20180408:eq1} by ($\epsilon\in\mb F$):
\begin{equation}\label{20180412:eq1}
\{f,g\}^{R,\epsilon}(L)
=
\sum_{i,j\in I}
\frac{\partial f}{\partial u_i}(L)
\frac{\partial g}{\partial u_j}(L)
\{u_i,u_j\}^{R,\epsilon} (L)
\,,\,\, L\in\mf g^*
\,,
\end{equation}
where $f$ and $g$ are polynomial functions on $\mf g^*$.
They are given by
\begin{equation}\label{eq:ragn3}
\begin{array}{r}
\displaystyle{
\vphantom{\Big(}
\{f,g\}^{R,\epsilon}(L)
=
\frac12\Tr\Big( L \circ [d_Lf,R((L+\epsilon\id)\circ d_Lg\circ (L+\epsilon\id))] \Big)
} \\
\displaystyle{
\vphantom{\Big(}
-
\frac12\Tr\Big( L \circ [d_Lg,R((L+\epsilon\id)\circ d_Lf\circ (L+\epsilon\id))] \Big)
\,,\,\,
\epsilon\in\mb F
\,,}
\end{array}
\end{equation}
where 
\begin{equation}\label{20180412:eq2}
d_Lf=\sum_{i\in I}\frac{\partial f}{\partial u_i}(L)u_i
\,.
\end{equation}
These are the same Poisson structures which appeared in \cite{OR89}.
In order to make sense of equation \eqref{eq:ragn3} we need to identify $\mf g^*\simeq\mf g$
via the inner product \eqref{20180412:eq3}.
Indeed, under this identification, $L\in\mf g^*$ can be thought of as an element of $\mf g$,
so it makes sense to take products $L\circ a$ or $a\circ L$ for $a\in\mf g$.
\end{remark}
\begin{remark}\label{20180411:rem}
Assuming that both $R$ and $\frac12(R-R^*)$ are $R$-matrices over $\mf g$,
by Theorem \ref{20180408:thm}
the second bracket $\{\cdot\,,\,\cdot\}^R_2$ is a Poisson bracket on $S(\mf g)$,
and it is obtained as the coefficient of $2\epsilon$ in \eqref{eq:ragn3}:
$$
\begin{array}{r}
\displaystyle{
\vphantom{\Big(}
\{f,g\}^{R}_2(L)
=
\frac14\Tr\big( L \circ [d_Lf,R(L\circ d_Lg+d_Lg\circ L)] \big)
} \\
\displaystyle{
\vphantom{\Big(}
-
\frac14\Tr\big( L \circ [d_Lg,R(L\circ d_Lf+d_Lf\circ L)] \big)
\,.}
\end{array}
$$
It has the following equivalent form:
\begin{equation}\label{20180411:eq1}
\begin{array}{l}
\displaystyle{
\vphantom{\Big(}
\{f,g\}^{R}_2(L)
} \\
\displaystyle{
\vphantom{\Big(}
=
\frac14\Tr\big( L \circ [d_Lf,(R-R^*)(L\circ d_Lg)] \big)
-
\frac14\Tr\big( L \circ [d_Lg,(R-R^*)(d_Lf\circ L)] \big)
} \\
\displaystyle{
\vphantom{\Big(}
+
\frac14\Tr\big( L\circ d_Lf \circ (R+R^*)(d_Lg\circ L) \big)
-
\frac14\Tr\big( L\circ d_Lg \circ (R+R^*)(d_Lf\circ L) \big)
\,,}
\end{array}
\end{equation}
which, for a skewadjoint $R$-matrix $R$,
reduces to \cite[Eq.(22)]{STS83}.
\end{remark}

\subsection{Hamiltonian equations and the triple Lenard-Magri scheme}

Recall that, given a Poisson structure on $\mf g^*$,
i.e. a Poisson bracket $\{\cdot\,,\,\cdot\}$ 
on the algebra of polynomial functions on $\mf g^*$,
and a Hamiltonian function $h$ on $\mf g^*$,
the corresponding Hamiltonian equation is,
in coordinates $x_j$, the following system of evolution equations
\begin{equation}\label{eq:ham1}
\frac{dx_j(t)}{dt}
=
\{h,u_j\}(L(t))
\,,\,\,j\in I
\,.
\end{equation}
It describes the time evolution of the point $L(t)=\sum_{i\in I}x_i(t)u^i\in\mf g^*$.
By Leibniz rule we get the corresponding evolution equation
for a function $f(L)$ on $\mf g^*$:
$\frac{df(L(t))}{dt}=\{h,f\}(L)$.
Using the identification of the symmetric algebra $S(\mf g)$ 
with the algebra of polynomial functions on $\mf g^*$,
we thus get the corresponding \emph{Hamiltonian equation}
on $S(\mf g)$:
\begin{equation}\label{eq:ham2}
\frac{df}{dt}
=
\{h,f\}
\,\,,\,\,\,\,
f\in S(\mf g)
\,.
\end{equation}

In particular, if $R$ and $\frac12(R-R^*)$ are $R$-matrices over $\mf g$,
we have, by Theorem \ref{20180408:thm},
the three Poisson brackets $\{\cdot\,,\,\cdot\}^R_i$, $i=1,2,3$,
and therefore, for every Hamiltonian function $h\in S(\mf g)$,
we have the corresponding three evolution equations:
\begin{equation}\label{eq:ham3}
\begin{array}{l}
\displaystyle{
\vphantom{\Big(}
\frac{du_j}{dt_{1}}
=
\frac12 
\sum_{i\in I}
\frac{\partial h}{\partial u_i}
\big(
[u_i,R(u_j)]-[u_j,R(u_i)]
\big)
\,,} \\
\displaystyle{
\vphantom{\Big(}
\frac{du_j}{dt_{2}}
=
\frac14 \sum_{i,h\in I}
\frac{\partial h}{\partial u_i}
u^h
\big(
[u_i,R(u_h\circ u_j+u_j\circ u_h)]-[u_j,R(u_h\circ u_i+u_i\circ u_h)]
\big)
\,,} \\
\displaystyle{
\vphantom{\Big(}
\frac{du_j}{dt_{3}}
=
\frac12 \sum_{i,h,k\in I}
\frac{\partial h}{\partial u_i}
u^hu^k
\big(
[u_i,R(u_h\circ u_j\circ u_h)]-[u_j,R(u_h\circ u_i\circ u_k)]
\big)
\,.}
\end{array}
\end{equation}

Recall that a \emph{Casimir} of $\mf g$
is an element $f\in S(\mf g)$ invariant with respect to the adjoint action of $\mf g$,
i.e. such that
\begin{equation}\label{eq:casimir}
\{x,f\}
:=
(\ad(x))(f)
=
\sum_{i\in I}\frac{\partial f}{\partial u_i}[x,u_i]
=
0\,\,,\text{ for all }\,\, x\in\mf g
\,.
\end{equation}
\begin{lemma}[\cite{OR89}]\label{20180410:lem}
If $R\in\End(\mf g)$ is an $R$-matrix over the Lie algebra $\mf g$,
and if $f,g\in S(\mf g)$ are Casimirs of $\mf g$,
then they Poisson commute with respect to the whole 
$\epsilon$-family of Poisson structures defined by \eqref{20180408:eq1}:
$\{f,g\}^{R,\epsilon}=0$, for every $\epsilon$. 
\end{lemma}
\begin{proof}
By \eqref{20180408:eq1} and \eqref{20180412:eq1}, we have
\begin{equation}\label{20180410:eq1}
\begin{array}{l}
\displaystyle{
\vphantom{\Big(}
\{f,g\}^{R,\epsilon}
=
\frac12 \sum_{j,h,k\in I}
\frac{\partial g}{\partial u_j}
(u^h+\epsilon\Tr(u^h))(u^k+\epsilon\Tr(u^k))
\big\{f,R(u_h\circ u_j\circ u_k)\big\}
} \\
\displaystyle{
\vphantom{\Big(}
-
\frac12 \sum_{i,h,k\in I}
\frac{\partial f}{\partial u_i}
(u^h+\epsilon\Tr(u^h))(u^k+\epsilon\Tr(u^k))
\big\{g,R(u_h\circ u_i\circ u_k)\big\}
\,.}
\end{array}
\end{equation}
If $f$ is a Casimir of $\mf g$, 
the first term of the RHS of \eqref{20180410:eq1} vanishes by \eqref{eq:casimir},
while if $g$ is a Casimir of $\mf g$, 
the second term of the RHS of \eqref{20180410:eq1} vanishes.
\end{proof}
If we take the Hamiltonian function to be a Casimir element $C\in S(\mf g)$,
the three evolution equation \eqref{eq:ham3}
greatly simplify thanks to equation \eqref{eq:casimir}.
They become
\begin{equation}\label{eq:ham4}
\begin{array}{l}
\displaystyle{
\vphantom{\Big(}
\frac{du_j}{dt_{1}}
=
\frac12 
\sum_{i\in I}
\frac{\partial C}{\partial u_i}
[R(u_i),u_j]
\,,} \\
\displaystyle{
\vphantom{\Big(}
\frac{du_j}{dt_{2}}
=
\frac12 \sum_{i,h\in I}
\frac{\partial C}{\partial u_i}
u^h
[R(u_i\circ u_h),u_j]
\,,} \\
\displaystyle{
\vphantom{\Big(}
\frac{du_j}{dt_{3}}
=
\frac12 \sum_{i,j,h,k\in I}
\frac{\partial C}{\partial u_i}
u^hu^k
[R(u_i\circ u_h\circ u_k),u_j]
\,.}
\end{array}
\end{equation}

An infinite collection of Casimirs is the following:
\begin{equation}\label{eq:ham5}
C_k
=
\frac1k
\sum_{i_1,\dots,i_k}
u_{i_1}\dots u_{i_k} \Tr( u^{i_1}\circ\dots\circ u^{i_k} )
\,,\,\,
k\geq 1
\,.
\end{equation}
Indeed, it is immediate to check that \eqref{eq:casimir} holds for all elements $C_k$.
Moreover, we have the following identity:
\begin{equation}\label{eq:ham5b}
\sum_{i\in I}\frac{\partial C_k}{\partial u_i}\otimes u_i
=
\sum_{i_1,\dots,i_{k-1}}
u_{i_1}\dots u_{i_{k-1}} \otimes (u^{i_1}\circ\dots\circ u^{i_{k-1}})
\,.
\end{equation}
Note that the functions on $\mf g^*$ corresponding to the Casimirs \eqref{eq:ham5}
are $C_k(L)=\frac1k\Tr(L^{\circ k})$, and we have $d_LC_k=L^{\circ(k-1)}$.

It immediately follows from \eqref{eq:ham4} and \eqref{eq:ham5b},
that the following ``triple Lenard-Magri scheme'' holds:
denoting $t_{k,i}$ the time evolution with respect to the Poisson structure $\{\cdot\,,\,\cdot\}^R_i$
and the Hamiltonian function $C_k$, we have
\begin{equation}\label{eq:ham6}
\begin{array}{l}
\displaystyle{
\vphantom{\Big(}
\frac{du_j}{dt_{k+1,1}}
=
\frac{du_j}{dt_{k,2}}
=
\frac{du_j}{dt_{k-1,3}}
} \\
\displaystyle{
\vphantom{\Big(}
=
\frac12 
\sum_{i_1,\dots,i_{k}}
u_{i_1}\dots u_{i_{k}} [R(u^{i_1}\circ\dots\circ u^{i_{k}}),u_j]
\,.}
\end{array}
\end{equation}

%

\section{Algebraic setup:
the algebra $\mc V_\infty\!\!\widehat{}$\,\, and continuous PVA structures}
\label{sec:4.1}

\subsection{The algebra $A$}

Throughout the rest of the paper 
we let $A$ be a finite dimensional associative algebra over $\mb F$,
with a unit $\id$,
and with a non-degenerate trace form $\Tr(\cdot):\,A\to\mb F$
(recall the definition at the beginning of Section \ref{sec:3}).
The typical example is the algebra $A=\End (V)$ of endomorphisms
of a finite-dimensional vector space $V$,
with the usual trace form $\Tr_V(XY)$.
We fix dual bases $\{E_\alpha\}_{\alpha\in I}$, $\{E^\alpha\}_{\alpha\in I}$ of $A$:
\begin{equation}\label{eq:motiv7}
\Tr(E^\alpha E_\beta) = \delta_{\alpha,\beta}
\,.
\end{equation}

\subsection{The differential algebra $\mc V_\infty\!\!\widehat{}\,\,$}

Consider the infinite set of variables
\begin{equation}\label{20180417:eq2}
u_{p,\alpha}
\,\,\text{ for }\,\,
p\in\mb Z,\,\alpha\in I
\,.
\end{equation}
The reason for considering such set of variables will be clear from the
discussion in Section \ref{sec:4},
where we present the construction of the ``affine'' O-R Poisson structures.
Consider the increasing sequence of algebras of differential polynomials ($N\in\mb Z$)
\begin{equation}\label{20180417:eq1}
\dots\subset\mc V_N\subset\mc V_{N+1}\subset\dots
\subset
\mc V_\infty
\,,
\end{equation}
where
\begin{equation}\label{20180417:eq3}
\mc V_N
=
\mb F\Big[u_{p,\alpha}^{(n)}
\,\Big|\,
\substack{p\in \geq-N-1 \\ \alpha\in I \\ n\in\mb Z_{\geq0}}
\Big]
\,\,\text{ for }\,\, N\in\mb Z
\,\,,\text{ and }\,\,
\mc V_\infty
=
\mb F\Big[u_{p,\alpha}^{(n)}
\,\Big|\,
\substack{p\in\mb Z \\ \alpha\in I \\ n\in\mb Z_{\geq0}}
\Big]
\,.
\end{equation}
These are differential algebras, with derivation $\partial:\,\mc V_N\to\mc V_N$
defined by $\partial u_{p,\alpha}^{(n)}=u_{p,\alpha}^{(n+1)}$.
As usual, we denote $u_{p,\alpha}=u_{p,\alpha}^{(0)}$.

We have the corresponding sequence of projection maps
\begin{equation}\label{20180417:eq4}
\begin{tikzcd}
\dots & \arrow[->>]{l} \mc V_N & \arrow[->>]{l} \mc V_{N+1} & \arrow[->>]{l} \dots
& \arrow[->>]{l} \arrow[->>,bend right]{ll}{\pi_{N+1}} \arrow[->>,swap,bend right]{lll}{\pi_N} \mc V_\infty
\,,
\end{tikzcd}
\end{equation}
where $\pi_N$ is the differential algebra homomorphism
defined by setting $u_{p,\alpha}=0$ for $p<-N-1$ and for all $\alpha\in I$.
We then have the corresponding inverse limit algebra
\begin{equation}\label{20180417:eq5}
\mc V_\infty\yhat
=
\lim_{\substack{\longleftarrow \\ N}}
\mc V_N
\,.
\end{equation}
Its elements are infinite sums
\begin{equation}\label{20180417:eq6}
f=\sum_{s=0}^\infty f_s
\,\,\text{ with }\,\,
f_s\in\mc V_\infty
\,,
\end{equation}
with the property that, for all $N\in\mb Z$,
\begin{equation}\label{20180417:eq7}
\pi_N(f_s)=0
\,\,\text{ for }\,\, s>>0
\,.
\end{equation}
In other words,
for every $N\in\mb Z$, $\pi_N(f)$ becomes a finite sum of elements in $\mc V_N$.
\begin{proposition}\label{20220503:prop1}
$\mc V_\infty\yhat$ is a differential algebra extension of $\mc V_\infty$,
with uniquely defined derivations 
$$
\frac{\partial}{\partial u_{p,\alpha}^{(n)}}
:\,\mc V_\infty\yhat\to\mc V_\infty\yhat
\,\,\text{ for }\,\, p\in\mb Z,\,\alpha\in I,\,n\in\mb Z_{\geq0}
\,,
$$
extending the usual partial derivatives on $\mc V_\infty$ and 
such that 
$$
\Big[
\frac{\partial}{\partial u_{p,\alpha}^{(n)}},\partial
\Big]
=
\frac{\partial}{\partial u_{p,\alpha}^{(n-1)}}
\,.
$$
Moreover, we have uniquely defined maps
$$
\frac{\delta}{\delta u_{p,\alpha}}:\,\mc V_\infty\yhat\to\mc V_\infty\yhat
\,\,\text{ for }\,\, p\in\mb Z,\,\alpha\in I
\,,
$$
extending the usual variational derivatives on $\mc V_\infty$.
These variational derivatives $\frac{\delta}{\delta u_{p,\alpha}}$ vanish on total derivatives,
so they induce linear maps on the quotient space of ``continuous'' local functionals:
\begin{equation*}
\frac{\delta}{\delta u_{p,\alpha}}
\,:\,\,
\mc V_\infty\yhat/\partial\mc V_\infty\yhat
\to\mc V_\infty\yhat
\,.
\end{equation*}
\end{proposition}
\begin{proof}
Clearly, $\mc V_\infty\yhat$ is an algebra extension of $\mc V_\infty$,
and the derivation $\partial$ uniquely extends to a derivation of $\mc V_\infty\yhat$
defined by $\partial f=\sum_{s=0}^\infty \partial f_s\,\in\mc V_\infty\yhat$,
for $f$ as in \eqref{20180417:eq6}.

Next, we show that the partial derivatives $\frac{\partial}{\partial u_{p,\alpha}^{(n)}}$
uniquely extend to well-defined derivations
$\frac{\partial}{\partial u_{p,\alpha}^{(n)}}:\,\mc V_\infty\yhat\to\mc V_\infty\yhat$,
defined, for $f=\sum_{s=0}^\infty f_s$ as in \eqref{20180417:eq6}-\eqref{20180417:eq7}, by
\begin{equation}\label{20180418:eq13}
\frac{\partial f}{\partial u_{p,\alpha}^{(n)}}
=
\sum_{s=0}^\infty
\frac{\partial f_s}{\partial u_{p,\alpha}^{(n)}}
\,\in\mc V_\infty\yhat 
\,.
\end{equation}
For this, we need to check that the RHS of \eqref{20180418:eq13} lies in $\mc V_\infty\yhat$.
By the definition of the projection maps $\pi_N$,
we have, for $\bar f\in\mc V_\infty$
\begin{equation}\label{20180418:eq12}
\pi_N\big(
\frac{\partial \bar f}{\partial u_{p,\alpha}^{(n)}}
\big)
=
\frac{\partial}{\partial u_{p,\alpha}^{(n)}}
\pi_N(\bar f)
\,\,\text{ if }\,\,
p\geq -N-1
\,.
\end{equation}
For $N\in\mb Z$, let $\tilde{N}=\max\{N,-p-1\}$.
By \eqref{20180417:eq7}
there exists $S_{\tilde{N}}\in\mb Z_{\geq0}$ such that
$\pi_{\tilde{N}}(f_s)=0$ for all $s> S_{\tilde{N}}$.
Since $p\geq -\tilde{N}-1$, we have, for $s> S_{\tilde{N}}$,
$$
\pi_{\tilde{N}}\big(\frac{\partial f_s}{\partial u_{p,\alpha}^{(n)}}\big)
=
\frac{\partial}{\partial u_{p,\alpha}^{(n)}}\pi_{\tilde{N}}(f_s)=0
\,,
$$
which implies 
$\pi_N\big(\frac{\partial f_s}{\partial u_{p,\alpha}^{(n)}}\big)=0$, since $N\leq\tilde{N}$.
Hence,
applying $\pi_N$ to the RHS of \eqref{20180418:eq13} we get
a finite sum in $\mc V_N$:
$$
\sum_{s=0}^{S_{\tilde{N}}}
\pi_N\big(
\frac{\partial f_s}{\partial u_{p,\alpha}^{(n)}}
\big)
\,.
$$
The facts that the maps  $\frac{\partial}{\partial u_{p,\alpha}^{(n)}}$
are derivations of the commutative associative product of $\mc V_\infty\yhat$,
and that they satisfy the usual commutation relations with $\partial$,
follow by construction, since the same properties hold in $\mc V_\infty$.

In the same way one can prove that the variational derivatives in $\mc V_\infty$
$$
\frac{\delta \bar f}{\delta u_{p,\alpha}}
:=
\sum_{n\in\mb Z_{\geq0}}(-\partial)^n\frac{\partial\bar f}{\partial u_{p,\alpha}^{(n)}}
\,\,,\,\,\,\,\bar f\in\mc V_\infty
\,,
$$
uniquely extend to linear maps $\frac{\delta}{\delta u_{p,\alpha}}:\,\mc V_\infty\yhat\to\mc V_\infty\yhat$,
given, for $f$ as in \eqref{20180417:eq6}-\eqref{20180417:eq7}, by
\begin{equation}\label{eq:varder}
\frac{\delta f}{\delta u_{p,\alpha}}
=
\sum_{s=0}^\infty
\frac{\delta f_s}{\delta u_{p,\alpha}}
\,\in\mc V_\infty\yhat
\,.
\end{equation}
The last assertion of the proposition is obvious.
\end{proof}

\subsection{Continuous differential and pseudodifferential operators}

We have the corresponding increasing sequence of the algebras
of differential operators
\begin{equation}\label{20180417:eq1b}
\dots\subset\mc V_N[\partial]\subset\mc V_{N+1}[\partial]\subset\dots
\subset
\mc V_\infty[\partial]
\,,
\end{equation}
with the corresponding projection maps (commuting with $\partial$)
\begin{equation}\label{20180417:eq4b}
\pi_N\,:\,\,\mc V_\infty[\partial]\twoheadrightarrow\mc V_N[\partial]
\,,\,\, N\in\mb Z
\,.
\end{equation} 
Hence, we can consider the inverse limit
\begin{equation}\label{20180417:eq8}
\mc V_\infty[\partial]\xhat
=
\lim_{\substack{\longleftarrow \\ N}}
\mc V_N[\partial]
\,.
\end{equation}
We will call an element $P(\partial)\in\mc V_\infty[\partial]\xhat$
a \emph{continuous} differential operator over $\mc V_\infty$.
It is, by definition, an infinite sum
\begin{equation}\label{20180417:eq6b}
P(\partial)=\sum_{s=0}^\infty P_s(\partial)
\,\,\text{ with }\,\,
P_s(\partial)\in\mc V_\infty[\partial]
\,,
\end{equation}
with the property that, for all $N\in\mb Z$,
\begin{equation}\label{20180417:eq7b}
\pi_N(P_s(\partial))=0
\,\,\text{ for }\,\, s>>0
\,.
\end{equation}
Hence, for every $N\in\mb Z$, $\pi_N(P(\partial))$ is a well-defined element of $\mc V_N[\partial]$.
Note that $\mc V_{\infty}[\partial]\xhat$ is of course larger than $\mc V_\infty\yhat[\partial]$,
as $P(\partial)$ in \eqref{20180417:eq6b} might have unbounded powers of $\partial$.

Clearly, $\mc V_\infty[\partial]\xhat$ is an algebra extension of $\mc V_\infty[\partial]$:
given $P(\partial)=\sum_{s=0}^\infty P_s(\partial)$ and $Q(\partial)=\sum_{t=0}^\infty Q_t(\partial)$
as in \eqref{20180417:eq6b}-\eqref{20180417:eq7b},
their $\circ$ product is
\begin{equation}\label{20180418:eq4}
P(\partial)\circ Q(\partial)
=
\sum_{s,t=0}^\infty P_s(\partial)\circ Q_t(\partial)
\,,
\end{equation}
which lies in $\mc V_\infty[\partial]\xhat$ since, for $N\in\mb Z$, we have
$\pi_N(P_s(\partial)\circ Q_t(\partial))
=\pi_N(P_s(\partial))\circ\pi_N(Q_t(\partial))=0$
for all but finitely many values of $s$ and $t$.
For the same reason, we have a natural action of a continuous differential operator
$P(\partial)\in\mc V_\infty[\partial]\xhat$
on an element $f\in\mc V_\infty\yhat$ in the obvious way:
if $f$ is as in \eqref{20180417:eq6}-\eqref{20180417:eq7}
and $P(\partial)$ is as in \eqref{20180417:eq6b}-\eqref{20180417:eq7b},
then $P(\partial)f=\sum_{s,t=0}^\infty P_s(\partial)f_t$,
and this sum becomes finite once we apply the projection map $\pi_N$.

Passing from differential operators to symbols,
we have the corresponding sequence of projection maps
$\pi_N:\,\mc V_\infty[\lambda]\twoheadrightarrow\mc V_N[\lambda]$,
$N\in\mb Z$, commuting with $\lambda$,
and the associated inverse limit
\begin{equation}\label{20180417:eq8e}
\mc V_\infty[\lambda]\xhat
=
\lim_{\substack{\longleftarrow \\ N}}
\mc V_N[\lambda]
\,.
\end{equation}
It is an algebra extension of $\mc V_\infty[\lambda]$,
and taking symbols summand by summand in \eqref{20180417:eq6b}
we have the corresponding symbol map 
$\mc V_\infty[\partial]\xhat\stackrel{\sim}{\longrightarrow}\mc V_\infty[\lambda]\xhat$.
Similarly, we denote by $\mc V_\infty[\lambda,\mu]\xhat$ the inverse limit
of the sequence of projection maps of the algebras of polynomials in two variables,
$\pi_N:\,\mc V_\infty[\lambda,\mu]\twoheadrightarrow\mc V_N[\lambda,\mu]$.
Obviously, if $P(\lambda)\in\mc V_\infty[\lambda]\xhat$,
then $P(\lambda+\mu)\in\mc V_\infty[\lambda,\mu]\xhat$.

\medskip

In the same way as for polynomials, 
we can extend the projection maps $\pi_N$
to the algebras of pseudodifferential operators, or Laurent series 
(i.e. the symbols of pseudodifferential operators):
\begin{equation}\label{20180417:eq1c}
\pi_N\,:\,\,\mc V_\infty((\partial^{-1}))\twoheadrightarrow\mc V_N((\partial^{-1}))
\,\,\text{ or }\,\,
\pi_N\,:\,\,\mc V_\infty((z^{-1}))\twoheadrightarrow\mc V_N((z^{-1}))
\,,
\end{equation}
letting $\pi_N$ commute with $\partial$ or $z$.
We have the associated inverse limits
\begin{equation}\label{20180417:eq8c}
\mc V_\infty((\partial^{-1}))\xhat
=
\lim_{\substack{\longleftarrow \\ N}}
\mc V_N((\partial^{-1}))
\,\,\text{ and }\,\,
\mc V_\infty((z^{-1}))\xhat
=
\lim_{\substack{\longleftarrow \\ N}}
\mc V_N((z^{-1}))
\,.
\end{equation}
In particular, a \emph{continuous} pseudodifferential operator over $\mc V_\infty$
is an infinite sum
\begin{equation}\label{20180418:eq2}
P(\partial)
=
\sum_{s=0}^\infty
P_s(\partial)
\,\in\mc V_\infty((\partial^{-1}))\xhat
\,,
\end{equation}
with $P_s(\partial)\in\mc V_\infty((\partial^{-1}))$
such that, for every $N\in\mb Z$, 
$\pi_N(P_s(\partial))=0$ for all but finitely many values of $s$.
As for differential operators, $\mc V_\infty((\partial^{-1}))\xhat$ is an algebra extension
of $\mc V_\infty((\partial^{-1}))$,
with $\circ$ product defined as in \eqref{20180418:eq4}.

We can define the adjoint $P^*(\partial)$
of a continuous pseudodifferential operator $P(\partial)\in\mc V_\infty((\partial^{-1}))\xhat$
by taking the adjoint of each summand in \eqref{20180418:eq2}.
Also, 
the \emph{residue} of a continuous pseudodifferential operator $P(\partial)$ is defined in the
obvious way:
if $P(\partial)\in\mc V_\infty((\partial^{-1}))\xhat$ is as in \eqref{20180418:eq2},
then
\begin{equation}\label{20180418:eq5}
\res_\partial P(\partial)
=
\sum_{s=0}^\infty \res_\partial(P_s(\partial))
\,\in\mc V_\infty\yhat
\,.
\end{equation}

\medskip

Throughout the paper, we will use the following standard notation:
for a continuous pseudodifferential operator  $P(\partial)\in\mc V_\infty((\partial^{-1}))\xhat$ 
and elements $f,g\in\mc V_\infty\yhat$,
we let
\begin{equation}\label{eq:notation}
P(z+x)\big(\big|_{x=\partial}f\big)g
=
gP(z+\partial)f
\,.
\end{equation}
In other words, if $P(\partial)=\sum_{s=0}^\infty \sum_{n=-\infty}^{N_s}p_{s,n}\partial^n$,
$f=\sum_{t=0}^\infty f_t$ and $g=\sum_{r=0}^\infty g_r$, then
$$
P(z+x)\big(\big|_{x=\partial}f\big)g
=
\sum_{s,t,r=0}^\infty\sum_{n=-\infty}^{N_s}\sum_{k\in\mb Z_{\geq0}}\binom{n}{k}
p_{s,n}f_t^{(k)}g_r z^{n-k}\in\mc V_\infty((z^{-1}))\xhat
\,.
$$

\medskip

%
For an algebra $\mc V$, we denote by $\mc V((z^{-1},w^{-1}))$ the algebra:
\begin{equation}\label{20180418:eq1}
\mc V((z^{-1},w^{-1}))
:=
\mc V[[z^{-1},w^{-1}]][z,w]
\,.
\end{equation}
As above, we consider the sequence of projection maps
$\pi_N:\,\mc V_\infty((z^{-1},w^{-1}))\twoheadrightarrow\mc V_N((z^{-1},w^{-1}))$,
and the corresponding inverse limit
\begin{equation}\label{20180417:eq8d}
\mc V_\infty((z^{-1},w^{-1}))\xhat
=
\lim_{\substack{\longleftarrow \\ N}}
\mc V_N((z^{-1},w^{-1}))
\,.
\end{equation}

\medskip

The key object in the forthcoming Sections
will be the following \emph{Lax operator}:
\begin{equation}\label{20180417:eq9}
L(z)
:=
\sum_{p\in\mb Z,\alpha\in I}
u_{p,\alpha}z^{-p-1}E^\alpha
\,\in
\mc V_\infty((z^{-1}))\xhat\otimes A
\,.
\end{equation}
It can be viewed as the generating series of all the variables $u_{p,\alpha}$, $p\in\mb Z,\alpha\in I$.
Here and further we omit the tensor product sign
between elements of $\mc V_\infty\yhat$ (or any its polynomial or Laurent series extension)
and elements of the algebra $A$.

\subsection{Continuous PVA $\lambda$-brackets on $\mc V_\infty\!\!\widehat{}\,\,$}

Recall from \cite{BDSK09} the definition of $\lambda$-bracket
and PVA $\lambda$-bracket on a differential algebra $\mc V$.
Here we introduce its continuous analogue on $\mc V_\infty\yhat$.
\begin{definition}\label{20180418:def}
A \emph{continuous} $\lambda$-\emph{bracket} on $\mc V_\infty\yhat$
is a bilinear over $\mb F$ map 
\begin{equation}\label{eq:pva1}
\{\cdot\,_\lambda\,\cdot\}
\,:\,\,
\mc V_\infty\yhat\times\mc V_\infty\yhat
\,\longrightarrow\,
\mc V_\infty[\lambda]\xhat
\,,
\end{equation}
satisfying the following axioms:
\begin{enumerate}[(i)]
\item
continuity: for every $N\in\mb Z$, there exists $M\in\mb Z$ (sufficiently large) 
such that
($f,g\in\mc V_\infty\yhat$):
\begin{equation}\label{eq:pvalim}
\pi_M(f)=0
\text{ or }
\pi_M(g)=0
\,\Rightarrow\,
\pi_N\{f_\lambda g\}=0
\,;
\end{equation}
\item
sesquilinearity ($f,g\in\mc V_\infty\yhat$):
\begin{equation}\label{eq:pva4}
\{\partial f_\lambda g\}
=
-\lambda \{f_\lambda g\}
\,\,,\,\,\,\,
\{f_\lambda \partial g\}
=
(\lambda+\partial) \{f_\lambda g\}
\,;
\end{equation}
\item
Leibniz rules ($f,g,h\in\mc V_\infty\yhat$):
\begin{equation}\label{eq:pva2}
\begin{array}{l}
\displaystyle{
\vphantom{\big(}
\{f_\lambda gh\}
=
\{f_\lambda g\}h+\{f_\lambda h\}g
\,;} \\
\displaystyle{
\vphantom{\big(}
\{fg_\lambda h\}
=
\{f_{\lambda+x} h\}\big|_{x=\partial}g+\{g_\lambda h\}\big|_{x=\partial}f
\,.}
\end{array}
\end{equation}
\end{enumerate}
In the second Leibniz rule we use the notation introduced in \eqref{eq:notation}.

A \emph{continuous Poisson vertex algebra} (PVA) $\lambda$-\emph{bracket} on $\mc V_\infty\yhat$
is a continuous $\lambda$-bracket
satisfying the following two extra axioms:
\begin{enumerate}[(i)]
\setcounter{enumi}{3}
\item
skewsymmetry ($f,g\in\mc V_\infty\yhat$):
\begin{equation}\label{eq:pva5}
\{f_\lambda g\}
=
-\big|_{x=\partial}\{g_{-\lambda-x}f\}
\,\in\mc V_\infty[\lambda]\xhat
\,;
\end{equation}
\item
Jacobi identity ($f,g,h\in\mc V_\infty\yhat$):
\begin{equation}\label{eq:pva3}
\{f_\lambda\{g_\mu h\}\}
-
\{g_\mu\{f_\lambda h\}\}
=
\{\{f_\lambda g\}_{\lambda+\mu}h\}
\,\in\mc V_\infty[\lambda,\mu]\xhat
\,.
\end{equation}
\end{enumerate}
\end{definition}
The Jacobi identity requires some explanation.
By definition of $\mc V_\infty[\lambda]\xhat$, we have
$$
\{g_\mu h\}=\sum_{t=0}^\infty Q_t(\mu)
\,,
$$
where $Q_t(\mu)\in\mc V_\infty[\mu]$ 
and, for every $M\in\mb Z$, we have $\pi_M(Q_t(\mu))=0$
for every $t> T_M$.
Moreover, for every given $t\in\mb Z_{\geq0}$, we have
$$
\{f_\lambda Q_t(\mu)\}=\sum_{s=0}^\infty P_{s,t}(\lambda,\mu)
\,,
$$
where $P_{s,t}(\lambda,\mu)\in\mc V_\infty[\lambda,\mu]$
and, for every $N\in\mb Z$, we have $\pi_N(P_{s,t}(\mu))=0$
for all $s>S_{t,N}$.
Then,
$$
\{f_\lambda \{g_\mu h\}\}=\sum_{s,t=0}^\infty P_{s,t}(\lambda,\mu)
\,,
$$
and we need to explain why this infinite sum lies in $\mc V_\infty[\lambda,\mu]\xhat$.
Fix $N\in\mb Z$,
and, by the continuous assumption (i),
let $M\in\mb Z$ be such that \eqref{eq:pvalim} holds.
Note that, for $t>T_M$,
we have by assumption that $\pi_M(Q_t(\mu))=0$,
and therefore by the continuity axiom \eqref{eq:pvalim} $\pi_N(\{f_\lambda Q_t(\mu)\})=0$.
Therefore, 
$$
\pi_N\big(\{f_\lambda \{g_\mu h\}\}\big)
=
\sum_{t=0}^{T_M}\pi_N(\{f_\lambda Q_t(\mu)\})
=
\sum_{t=0}^{T_M}\sum_{s=0}^{S_{t,N}}\pi_N(P_{s,t}(\lambda,\mu))
\,,
$$
which is a finite sum, thus lying in $\mc V_N[\lambda,\mu]$, as needed.
Similarly, all three terms in the Jacobi identity \eqref{eq:pva3}
lie in $\mc V_\infty[\lambda,\mu]\xhat$,
so the Jacobi identity makes sense.

\medskip

We want to show that, as for the usual Poisson vertex algebra \cite{BDSK09},
also a continuous PVA $\lambda$-bracket is uniquely defined by its values
\begin{equation}\label{20180418:eq6}
\{{u_{p,\alpha}}_\lambda{u_{q,\beta}}\}
\,\in\mc V_\infty[\lambda]\xhat
\,\,,\,\,\,\, p,q\in\mb Z,\,\alpha,\beta\in I
\,,
\end{equation}
on the set of generators \eqref{20180417:eq2}.
Recall that, by the definition of the projection maps $\pi_N$ in \eqref{20180417:eq4},
$\pi_N(u_{p,\alpha})=0$ for all $p<-N-1$.
Hence, the continuity condition \eqref{eq:pvalim}
implies, on the $\lambda$-brackets \eqref{20180418:eq6},
that, for every $N\in\mb Z$, there exists $M\in\mb Z$ such that
\begin{equation}\label{20180418:eq7}
\pi_N\big(
\{{u_{p,\alpha}}_\lambda{u_{q,\beta}}\}
\big)
=
0
\,\,\text{ if either }\,\,
p<-M-1
\,\,\text{ or }\,\,
q<-M-1
\,.
\end{equation}
We claim that if the continuity conditions on generators \eqref{20180418:eq7} hold,
then there is a unique way to extend the $\lambda$-bracket on generators \eqref{20180418:eq6} 
to a continuous $\lambda$-bracket on $\mc V_\infty\yhat$,
by the following Master Formula \cite{BDSK09}:
if $f=\sum_{s=0}^\infty f_s,\,g=\sum_{t=0}^\infty g_t\,\in\mc V_\infty\yhat$
are as in \eqref{20180417:eq6}-\eqref{20180417:eq7}, then
\begin{equation}\label{20180418:eq8}
\{f_\lambda g\}
=
\sum_{s,t=0}^\infty
\sum_{p,q\in\mb Z}
\sum_{\alpha,\beta\in I}
\sum_{m,n\in\mb Z_{\geq0}}
\frac{\partial g_t}{\partial u_{q,\beta}^{(n)}}
(\lambda+\partial)^n
\{{u_{p,\alpha}}_{\lambda+\partial}{u_{q,\beta}}\}_\to
(-\lambda-\partial)^m
\frac{\partial f_s}{\partial u_{p,\alpha}^{(m)}}
\,.
\end{equation}
Indeed, let us check that the RHS of \eqref{20180418:eq8} lies in $\mc V_\infty[\lambda]\xhat$.
In other words, we fix $N$ and we need to show that, after applying $\pi_N$,
the RHS of \eqref{20180418:eq8} becomes a finite sum.
Let $M\in\mb Z$ be such that \eqref{20180418:eq7} holds
and let $\tilde{N}=\max\{M,N\}$.
By the definition of $\mc V_\infty\yhat$,
there exists $S$ such that
\begin{equation}\label{20180418:eq9}
\pi_{\tilde{N}}(f_s)=0
\,,\,\,
\pi_{\tilde{N}}(g_t)=0
\,\,\text{ for all }\,\,
s,t\geq S.
\end{equation}
Moreover, if either $p<-\tilde{N}-1\leq-M-1$ or $q<-\tilde{N}-1$,
we have, by \eqref{20180418:eq7},
$$
\pi_N\big(
\{{u_{p,\alpha}}_\lambda{u_{q,\beta}}\}
\big)
=
0
\,.
$$
On the other hand, if $p\geq-\tilde{N}-1$ and $s>S$,
we have, by \eqref{20180418:eq12},
$$
\pi_{\tilde{N}}\big(
\frac{\partial f_s}{\partial u_{p,\alpha}}
\big)
=
\frac{\partial}{\partial u_{p,\alpha}}
\pi_{\tilde{N}}(f_s)
=0
\,\,\text{ which implies }\,\,
\pi_N\big(
\frac{\partial f_s}{\partial u_{p,\alpha}}
\big)
=
0
\,.
$$
And analogously, 
if $q\geq-\tilde{N}-1$ and $t>S$,
we have
$\pi_N\big(\frac{\partial g_t}{\partial u_{q,\beta}}\big)=0$.
Hence, applying $\pi_N$ to the RHS of \eqref{20180418:eq8},
we are left with the finite sum in $\mc V_N[\lambda]$
$$
\!\sum_{s,t=0}^S\!
\!\sum_{p,q=-\tilde{N}-1}^\infty\!
\sum_{\alpha,\beta\in I}\!
\sum_{m,n\in\mb Z_{\geq0}}\!
\!\!\!
\pi_N\!\big(\frac{\partial g_t}{\partial u_{q,\beta}^{(n)}}\big)
(\lambda\!+\!\partial)^n
\pi_N\!\big(\{{u_{p,\alpha}}_{\lambda+\partial}{u_{q,\beta}}\}\big)_\to
(\!-\!\lambda\!-\!\partial)^m
\pi_N\!\big(\frac{\partial f_s}{\partial u_{p,\alpha}^{(m)}}\big)
\,.
$$
(It is a finite sum since, for every $s$ and $t$, $f_s,g_t\in\mc V_\infty$
are polynomials in the variables $u_{p,\alpha}^{(n)}$.)

The proof that the Master Formula \eqref{20180418:eq8}
satisfies the sesquilinearity axioms \eqref{eq:pva4}
and the Leibniz rules \eqref{eq:pva2} is as in the usual PVA case \cite{BDSK09}.
Likewise one can show, as in the usual case, that the continuous $\lambda$-bracket given
by \eqref{20180418:eq8} is a continuous PVA $\lambda$-bracket,
i.e. the skewsymmetry axiom \eqref{eq:pva5} and the Jacobi identity \eqref{eq:pva3} hold,
if and only if they hold on generators:
skewsymmetry:
\begin{equation}\label{eq:pva5b}
\{{u_{p,\alpha}}_\lambda {u_{q,\beta}}\}
=
-\big|_{x=\partial}\{{u_{q,\beta}}_{-\lambda-x}{u_{p,\alpha}}\}
\,\in\mc V_\infty[\lambda]\xhat
\,,
\end{equation}
and Jacobi identity:
\begin{equation}\label{eq:pva3b}
\{{u_{p,\alpha}}_\lambda\{{u_{q,\beta}}_\mu {u_{r,\gamma}}\}\}
-
\{{u_{q,\beta}}_\mu\{{u_{p,\alpha}}_\lambda {u_{r,\gamma}}\}\}
=
\{\{{u_{p,\alpha}}_\lambda {u_{q,\beta}}\}_{\lambda+\mu}{u_{r,\gamma}}\}
\,\in\mc V_\infty[\lambda,\mu]\xhat
\,.
\end{equation}
Summarizing the above observations, we have the following Theorem, whose details of the proof are
left to the reader:
\begin{theorem}\label{20180418:thm1}
Every choice for the $\lambda$-brackets 
$\{{u_{p,\alpha}}_\lambda{u_{q,\beta}}\}\in\mc V_\infty[\lambda]\xhat$
among the generators $u_{p,\alpha}$, $p\in\mb Z$, $\alpha\in I$,
of $\mc V_\infty\yhat$,
satisfying the continuity conditions \eqref{20180418:eq7},
extends uniquely to a continuous $\lambda$-bracket on $\mc V_\infty\yhat$,
and it is given by the Master Formula \eqref{20180418:eq8}.
Moreover, the continuous $\lambda$-bracket \eqref{20180418:eq8}
is a continuous PVA $\lambda$-bracket on $\mc V_\infty\yhat$
if and only if skewsymmetry and Jacobi identity hold on generators,
i.e. \eqref{eq:pva5b} and \eqref{eq:pva3b} hold.
\end{theorem}

\medskip
\subsection{Continuous PVA $\lambda$-bracket on $\mc V_\infty\!^{\widehat{}}$
in terms of the generating series $L(z)$}

The generating series $L(z)\in\mc V_\infty((z^{-1}))\xhat\otimes A$
defined by \eqref{20180417:eq9}
encodes all the generators $\{u_{p,\alpha}\}_{p\in\mb Z,\alpha\in I}$ of $\mc V_\infty\yhat$.
Hence, all the $\lambda$-brackets $\{{u_{p,\alpha}}_\lambda{u_{q,\beta}}\}$
among the generators can be encoded in
\begin{equation}\label{20180418:eq10}
\{L_1(z)_\lambda L_2(w)\}
=
\sum_{p,q\in\mb Z}
\sum_{\alpha,\beta\in I}
\{{u_{p,\alpha}}_\lambda{u_{q,\beta}}\} z^{-p-1}w^{-q-1} E^\alpha\otimes E^\beta
\,,
\end{equation}
where
\begin{equation}\label{eq:yang-notation}
L_1(z)=L(z)\otimes\id
\,\,\text{ and }\,\,
L_2(z)=\id\otimes L(z)
\,\,\text{ in }\,\,
\mc V_\infty((z^{-1}))\xhat\otimes A^{\otimes2}
\,.
\end{equation}
(As usual, we omit the tensor product sign for the factors in $\mc V_\infty((z^{-1}))$.)
The continuity conditions \eqref{20180418:eq7}
translate into saying that,
for every $N\in\mb Z$, there exists $M\in\mb Z$ such that
\begin{equation}\label{20180418:eq11b}
\pi_N\big(\{L_1(z)_\lambda L_2(w)\})
\in
(\mc V_N[\lambda])[[z^{-1},w^{-1}]]z^Mw^M\otimes A^{\otimes2}
\,.
\end{equation}
In other words,
\begin{equation}\label{20180418:eq11}
\{L_1(z)_\lambda L_2(w)\}
\,\in
(\mc V_\infty[\lambda])((z^{-1},w^{-1}))\xhat\otimes A^{\otimes2}
\,,
\end{equation}
Moreover, all the skewsymmetry conditions \eqref{eq:pva5b} are encoded in
the single identity
\begin{equation}\label{eq:pva5c}
\{L_1(z)_\lambda L_2(w)\}
=
-\big|_{x=\partial}
\{L_2(w)_{-\lambda-x}L_1(z)\}
\,,
\end{equation}
in $(\mc V_\infty[\lambda])((z^{-1},w^{-1}))\xhat\otimes A^{\otimes2}$,
while the Jacobi identities \eqref{eq:pva3b} are encoded in
\begin{equation}\label{eq:pva3c}
\{L_1(z)_\lambda\{L_2(w)_\mu L_3(v)\}\}
-
\{L_2(w)_\mu\{L_1(z)_\lambda L_3(v)\}\}
=
\{\{L_1(z)_\lambda L_2(w)\}_{\lambda+\mu} L_3(v)\}
\,,
\end{equation}
in $(\mc V_\infty[\lambda,\mu])((z^{-1},w^{-1},v^{-1}))\xhat\otimes A^{\otimes3}$,
where $L_1(z), L_2(z), L_3(z)\in\mc V_\infty((z^{-1}))\xhat\otimes A^{\otimes3}$
are defined as in \eqref{eq:yang-notation}.
We can thus translate Theorem \ref{20180418:thm1}
in terms of generating series as follows:
\begin{theorem}\label{20180418:thm2}
A continuous PVA $\lambda$-bracket on $\mc V_\infty\yhat$
is uniquely determined by an element
$$
\{L_1(z)_\lambda L_2(w)\}
\,\in
(\mc V_\infty[\lambda])((z^{-1},w^{-1}))\xhat\otimes A^{\otimes2}
\,,
$$
satisfying the skewsymmetry condition \eqref{eq:pva5c}
and the Jacobi identity \eqref{eq:pva3c}.
\end{theorem}

\subsection{Continuous local Poisson bracket on 
$\mc V_\infty\!\!\widehat{}\,\,/\partial\mc V_\infty\!\!\widehat{}\,\,$}\label{sec:locPoisson}

As in the usual PVA case, a continuous PVA $\lambda$-bracket on $\mc V_\infty\yhat$
induces a Lie algebra bracket on 
the space of local functionals 
$$
\{\cdot\,,\,\cdot\}:\,
\big(\mc V_\infty\yhat/\partial\mc V_\infty\yhat\big)
\times
\big(\mc V_\infty\yhat/\partial\mc V_\infty\yhat\big)
\to
\big(\mc V_\infty\yhat/\partial\mc V_\infty\yhat\big)
\,,
$$
given by 
$$
\{\tint f,\tint g\}=\tint \{f_\lambda g\}\big|_{\lambda=0}
\,.
$$
This bracket is clearly a well-defined element of $\mc V_\infty\yhat$,
and, as in the usual PVA case, 
the Lie algebra axioms are an immediate consequence
of the skewsymmetry and Jacobi identity axioms \eqref{eq:pva5}-\eqref{eq:pva3}
of a continuous PVA $\lambda$-bracket.

By Theorem \ref{20180418:thm1},
any continuous $\lambda$-bracket on $\mc V_\infty\yhat$
is given by the Master Formula \eqref{20180418:eq8}.
It follows that the induced local Poisson bracket on the space of local functionals is,
for $f,g\in\mc V_\infty\yhat$,
\begin{equation}\label{eq:motiv22}
\{\tint f,\tint g\}
=
\sum_{p,q\in\mb Z,\alpha,\beta\in I}
\int
\frac{\delta g}{\delta u_{q,\beta}}
\{{u_{p,\alpha}}{\,_{}}_\partial\, {u_{q,\beta}}\}_{\to}
\frac{\delta f}{\delta u_{p,\alpha}}
\,\in\mc V_\infty\yhat/\partial\mc V_\infty\yhat
\,,
\end{equation}
where the arrow means that $\partial$ is moved to the right,
and the variational derivatives are defined in \eqref{eq:varder}.

\section{Affinization of the O-R construction and the three Adler-Oevel-Ragnisco (AOR) identites}
\label{sec:4}

We want to find the ``infinite-dimensional analogue'' of the Oevel-Ragnisco (O-R)
Poisson structures \eqref{20180408:eq1}.

\subsection{O-R construction in finite dimension: summary}

Let us first summarize the finite-dimensional construction presented in Section \ref{sec:3}.
The starting point is a finite dimensional associative algebra $\mf g$ over $\mb F$,
with a non-degenerate trace form $\Tr(\cdot):\,\mf g\to\mb F$
(recall the definition at the beginning of Section \ref{sec:3}),
and an $R$-matrix $R\in\End(\mf g)$.
If we fix dual bases $\{u_i\}_{i=1}^N$, $\{u^i\}_{i=1}^N$,
and identify $\mf g^*\simeq\mf g$ via the bilinear form \eqref{20180412:eq3},
an arbitrary element $L\in\mf g^*$ is
\begin{equation}\label{eq:motiv1}
L=\sum_{i=1}^Nx_iu^i
\,,\,\, x_i\in\mb F
\,,
\end{equation}
and the pairing $\langle\cdot\,|\,\cdot\rangle$ is, in coordinates,
\begin{equation}\label{eq:motiv2}
\langle L | a \rangle 
=
\Tr(L \circ a)
=
\sum_{i=1}^Nx_i\alpha_i
\,\,,\,\,\,\,
\text{ where }
a=\sum_{j=1}^N \alpha_ju_j\,\in\mf g
\,.
\end{equation}
The O-R construction
provides a Poisson algebra structure on the algebra of polynomial functions on $\mf g^*$,
which can be identified with $S(\mf g)$, the symmetric algebra over $\mf g$:
\begin{equation}\label{eq:motiv0}
\Big\{\text{polyn. functions on } \mf g^*\Big\}
\,\simeq\,
S(\mf g)
\,.
\end{equation}
This identification is clear:
an element $f\in S(\mf g)$, which can be expanded as
\begin{equation}\label{eq:motiv3}
f = \sum \text{coeff. } u_1^{k_1}\dots u_N^{k_N}\,\in S(\mf g)
\,,
\end{equation}
corresponds to the polynomial function on $\mf g^*$ given, in coordinates, by
the same polynomial of $x_1,\dots,x_N$:
\begin{equation}\label{eq:motiv4}
f(L) = \sum \text{coeff. } x_1^{k_1}\dots x_N^{k_N}\,\in\mb F
\,,
\end{equation}
if $L\in\mf g^*$ is as in \eqref{eq:motiv1}.
The differential of the function $f$ at a point $L\in\mf g^*$ 
is defined as the element $d_Lf\in\mf g$ such that
\begin{equation}\label{eq:motiv5}
f(L+\varepsilon Y)
=
f(L)+\varepsilon \Tr( Y \circ d_Lf )+O(\varepsilon^2)
\,.
\end{equation}
We can use Taylor's formula
and the equation \eqref{eq:motiv2}, 
to get the following explicit formula for the differential of $f\in S(\mf g)$:
\begin{equation}\label{eq:motiv6}
d_Lf
=
\sum_{i=1}^N
\frac{\partial f}{\partial u_i}(L) u_i
\,\in\mf g
\,.
\end{equation}

At this point, we consider the O-R Poisson bracket \eqref{eq:ragn3}:
$$
\begin{array}{r}
\displaystyle{
\vphantom{\Big(}
\{f,g\}^{R,\epsilon}(L)
=
\frac12\Tr( L \circ [d_Lf,R((L+\epsilon\id)\circ d_Lg\circ (L+\epsilon\id))] )
} \\
\displaystyle{
\vphantom{\Big(}
-
\frac12\Tr( L \circ [d_Lg,R((L+\epsilon\id)\circ d_Lf\circ (L+\epsilon\id))] )
\,.}
\end{array}
$$
We expand $L$ via \eqref{eq:motiv1}, 
$\id=\sum_{i=1}^N\Tr(u_i)u^i$,
and the differentials $d_Lf$ and $d_Lg$ via \eqref{eq:motiv6}.
As a result, we get
\begin{align*}
\{f,g\}^{R,\epsilon}(L)
&=
\frac12
\sum_{i,j,h,k,\ell=1}^N
\frac{\partial f}{\partial u_i}(L)
\frac{\partial g}{\partial u_j}(L)
(x_h+\epsilon\Tr(u_h))(x_k+\epsilon\Tr(u_k))x_\ell\times \\
& \times\Tr \big(
u^\ell \circ\big(
[u_i,R(u^h\circ u_j\circ u^k)]
-
[u_j,R(u^h\circ u_i\circ u^k)] 
\big)
\big)
\,.
\end{align*}
Identifying the polynomial functions on $\mf g^*$ and the elements of $S(\mf g)$
as in \eqref{eq:motiv3}-\eqref{eq:motiv4},
this corresponds to the Poisson bracket on $S(\mf g)$ given by
\begin{align*}
\{f,g\}^{R,\epsilon}
& =
\frac12
\sum_{i,j,h,k=1}^N
\frac{\partial f}{\partial u_i}
\frac{\partial g}{\partial u_j}
(u^h+\epsilon\Tr(u^h))
(u^k+\epsilon\Tr(u^k))\times \\
& \times\big(
[u_i,R(u_h\circ u_j\circ u_k)]
-
[u_j,R(u_h\circ u_i\circ u_k)]
\big)
\,,
\end{align*}
or, equivalently, to the Lie algebra bracket \eqref{20180408:eq1} on $\mf g$.

\subsection{Setup for the construction in the affine case}

We now proceed to describe the ``affine analogue'' of the O-R construction,
which we shall call the Adler-Oevel-Ragnisco (AOR) construction.
We let $\mc F$ be an algebra over $\mb F$ of ``test functions''.
By this we mean an algebra of functions $f(x)$ in one (space) variable $x\in M$,
which we can integrate: $\tint_M f(x) dx\,\in\mb F$,
and which we can differentiate: $f'(x)=\frac{\partial f(x)}{\partial x}\in\mc F$.
The only assumptions on the linear map $\int_M \, dx:\,\mc F\to\mb F$, 
are the following. 
First, we require the validity of the fundamental theorem of calculus, 
which has the form
\begin{equation}\label{eq:motiv8b}
\int_M f'(x) dx
=
0
\,\,
\text{ for every }
f(x)\in\mc F
\,,
\end{equation}
(we are assuming that the manifold $M$ has no boundaries).
In particular
we have the rule of integration by parts:
\begin{equation}\label{eq:motiv8}
\int_M f(x)g'(x) dx 
=
-\int_M f'(x)g(x) dx
\,,\,\,
\text{ for every }
f(x),g(x)\in\mc F
\,.
\end{equation}
Moreover, we require the non-degeneracy condition:
\begin{equation}\label{eq:motiv8c}
\int_M f(x)g(x) dx 
=
0
\,\,\text{ for all }\,\,
g(x)\in\mc F
\,\,\text{ implies }\,\,
f(x)=0
\,.
\end{equation}
A typical example of such an algebra of test functions
is the algebra of smooth functions on the circle $S^1$.

\medskip

The starting point of the AOR construction 
is the following associative algebra:
\begin{equation}\label{eq:motiv9}
\mf g
=
\mc F((\partial^{-1}))\otimes A
\,,
\end{equation}
where $A$ is, as in Section \ref{sec:4.1},
a finite dimensional associative algebra with a trace form $\Tr$, 
and $\mc F((\partial^{-1}))$ is the associative algebra
of pseudodifferential operators over $\mc F$.
The associative product $\circ$ on it is defined by the rule
\begin{equation}\label{eq:motiv10}
\partial^p\circ f(x)
=
\sum_{n\in\mb Z_{\geq0}}
\binom{p}{n}
f^{(n)}(x)\partial^{p-n}
\,,\,\,
p\in\mb Z,\,f(x)\in\mc F
\,,
\end{equation}
where $f^{(n)}(x)=\frac{\partial^n f(x)}{\partial x^n}$.
For simplicity, when writing an element of $\mf g$ we drop the tensor product sign:
for $P(x;\partial)\in\mc F((\partial^{-1}))$ and $X\in A$, 
we let $P(x;\partial)X$ be the corresponding element of $\mf g$.
The associative product of $\mf g$, defined componentwise, will be denoted by $\circ$:
\begin{equation}\label{eq:motiv12}
(P(x;\partial)X)\circ (Q(x;\partial)Y)
:=
(P(x;\partial)\circ Q(x;\partial)) (XY)
\,,
\end{equation}
for
$P(x;\partial),Q(x;\partial)\in\mc F((\partial^{-1}))$, $X,Y\in A$.
We define the following trace form $\langle\,\cdot\,\rangle:\,\mf g\to\mb F$
on $\mf g$ (recall the definition at the beginning of Section \ref{sec:3}):
\begin{equation}\label{eq:motiv11}
\langle P(x;\partial) X\rangle
:=
\int_M \res_\partial P(x;\partial) dx\, \Tr(X)
\,,
\end{equation}
where, as usual, the residue $\res_\partial P(x;\partial)$ 
of a pseudodifferential operator $P(x;\partial) = \sum_{p<\infty} f_p(x)\partial^{-p-1}$ 
denotes the coefficient $f_0(x)$ of $\partial^{-1}$.
It is easy to check, using integration by parts \eqref{eq:motiv8},
that the linear function $\langle\,\cdot\,\rangle$ defined by \eqref{eq:motiv11}
vanishes on commutators, 
and it is non-degenerate by \eqref{eq:motiv8c},
hence it is indeed a trace form on $\mf g$.

Using the pairing associated to this trace form (cf. \eqref{20180412:eq3}),
we can identify the associative algebra $\mf g$
with the (restricted) dual $\mf g^*$.
Namely, the generic point $L\in\mf g^*$ is written, in coordinates, as
\begin{equation}\label{eq:motiv13}
L
=
\sum_{p\in\mb Z,\,\alpha\in I}
x_{p,\alpha}(x)\partial^{-p-1}\, E^\alpha
\,\in\mf g^*
\,,
\end{equation}
where the coordinate functions $x_{p,\alpha}(x)\in\mc F$ are $0$ for $p<<0$.
We can also write the pairing associated to \eqref{eq:motiv11} in coordinates,
thus obtaining the ``affine analogue'' of formula \eqref{eq:motiv2}:
\begin{equation}\label{eq:motiv14}
\langle L \circ a \rangle 
=
\sum_{p\in\mb Z,\,\alpha\in I}
\int_M x_{p,\alpha}(x) y_{p\alpha}(x) dx
\,,
\end{equation}
where
\begin{equation}\label{eq:motiv15}
a=\sum_{q\in\mb Z,\beta\in I}\partial^q\circ y_{q,\beta}(x) \, E_\beta
\,\in\mf g
\,.
\end{equation}
(Here the ``dual'' coordinate functions $y_{q,\beta}(x)\in\mc F$ are $0$ for $q>>0$.)

Next, we want to describe the correct space of functions on $\mf g^*$,
on which we will define the AOR Poisson bracket.
We consider the space of \emph{local polynomial functionals} on $\mf g^*$.
These are functions $F:\mf g^*\to\mb F$ which, 
for $L\in\mf g^*$ as in \eqref{eq:motiv13},
have the form
\begin{equation}\label{eq:motiv16}
F(L)
=
\int_M
f(L(x))\, dx
=
\int_M
f\big(\big\{x_{p,\alpha}^{(n)}(x)\big\}_{p\in\mb Z,\alpha\in I,n\in\mb Z_{\geq0}}\big)\, dx
\,,
\end{equation}
where the \emph{density function} $f$ 
is a differential polynomial in the variables $x_{p,\alpha}$, $p\in\mb Z,\,\alpha\in I$
(and $L=L(x)$ is as in \eqref{eq:motiv13}).
In fact, we only need $f$ to be a polynomial when it is computed at a given point $L$,
i.e. when $x_{p,\alpha}=0$ for $p<<0$.

Here the inverse limit algebra $\mc V_\infty\yhat$, defined by \eqref{20180417:eq5},
comes into play.
Indeed, 
any $f\in\mc V_\infty\yhat$,
when evaluated at the point $L(x)$ as in \eqref{eq:motiv13},
i.e. at $u_{p,\alpha}^{(n)}=\frac{\partial^n x_{p,\alpha}(x)}{\partial x^n}\in\mc F$ ($=0$ for $p<<0$),
becomes a finite sum, i.e. a well-defined element of $\mc F$.
Hence 
$F(L)=\int_M f(L(x)) dx\in\mb F$ is well defined.

Clearly, when we evaluate $f\in\mc V_\infty\yhat$ at the point $L(x)\in\mf g^*$,
i.e. at $u_{p,\alpha}^{(n)}=\frac{\partial^n x_{p,\alpha}(x)}{\partial x^n}$,
the derivation $\partial$ corresponds to the derivative in $x$:
\begin{equation}\label{eq:motiv18}
(\partial f)(L(x))
=
\frac{\partial f(L(x))}{\partial x}
\,\in\mc F
\,.
\end{equation}
In particular, by \eqref{eq:motiv8b} we have 
$\tint_M(\partial f)(L(x)) dx = 0$.
Hence, 
we can identify 
\begin{equation}\label{eq:motiv0b}
\Big\{\text{local polyn. functionals on } \mf g^*\Big\}
\,\simeq\,
\mc V_\infty\yhat/\partial\mc V_\infty\yhat
\,,
\end{equation}
by associating the element $\tint f\in\mc V_\infty\yhat/\partial\mc V_\infty\yhat$,
where $f\in\mc V_\infty\yhat$,
with the local polynomial functional $F:\,\mf g^*\to\mb F$
given by \eqref{eq:motiv16}.
This is the ``affine analogue'' of the identification \eqref{eq:motiv0}.

Next, we need to find a formula for the differential $d_LF$
of a local functional $F=\tint f\in\mc V_\infty\yhat/\partial\mc V_\infty\yhat$
at a point $L\in\mf g^*$.
Recalling \eqref{eq:motiv5},
we let $d_LF\in\mf g$ be defined by
\begin{equation}\label{eq:motiv18b}
F(L+\varepsilon Y)
=
F(L)+\varepsilon \langle Y\circ d_LF \rangle
+O(\varepsilon^2)
\,.
\end{equation}
Let $L(x)\in\mf g^*$ be as in \eqref{eq:motiv13}
and let $Y(x)=\sum_{p\in\mb Z,\alpha\in I}y_{p,\alpha}(x)\partial^{-p-1}E^\alpha$,
with $y_{p,\alpha}(x)\in\mc F$ vanishing for $p<<0$.
By Taylor's formula and integration by parts, we have
$$
\begin{array}{l}
\displaystyle{
\vphantom{\Big(}
F(L+\varepsilon Y)
=
\int_M
f\big(L(x)+\varepsilon Y(x))dx
} \\
\displaystyle{
\vphantom{\Big(}
\simeq
\int_M
\Big(
f\big(L(x))
+\varepsilon
\sum_{p\in\mb Z,\alpha\in I,m\in\mb Z_{\geq0}}
\frac{\partial f}{\partial u_{p,\alpha}^{(m)}}(L(x))y_{p,\alpha}^{(m)}(x)
\Big)
\,dx
} \\
\displaystyle{
\vphantom{\Big(}
\simeq
F(L)
+\epsilon
\sum_{p\in\mb Z,\alpha\in I}
\int_M
y_{p,\alpha}(x)
\sum_{m\in\mb Z_{\geq0}}
(-\frac{\partial}{\partial x})^m
\frac{\partial f}{\partial u_{p,\alpha}^{(m)}}
(L(x))
dx
\,.}
\end{array}
$$
Recalling \eqref{eq:motiv14}-\eqref{eq:motiv15},
we are thus lead to define
\begin{equation}\label{eq:motiv19}
d_L F
=
\sum_{p\in\mb Z,\alpha\in I}
\partial^p\circ
\frac{\delta F}{\delta u_{p,\alpha}}
(L(x))
\,E_\alpha
\,\in\mf g
\,,
\end{equation}
where the variational derivatives of $F=\tint f\in\mc V_\infty\yhat/\partial\mc V_\infty\yhat$
are defined by \eqref{eq:varder}.

\subsection{$R$-matrices over $\mf g$}\label{sec:R}

In order to introduce the affine analogue of the OR bracket,
we need to fix an $R$-matrix 
over $\mf g=\mc F((\partial^{-1}))\otimes A$
viewed as a Lie algebra,
i.e. a linear map 
\begin{equation}\label{eq:black8}
R\,:\,\,
\mc F((\partial^{-1}))\otimes A\,\longrightarrow\,\mc F((\partial^{-1}))\otimes A
\,,
\end{equation}
satisfying the modified Yang-Baxter equation \eqref{eq:mod-YB}.

We construct $R$-matrices on the Lie algebra $\mf g=\mc F((\partial^{-1}))\otimes A$
as special cases of Example \ref{ex:R}.
Note that $\mc F[\partial]\partial^k\subset\mc F((\partial^{-1}))$
is an associative (hence Lie) subalgebra of $\mc F((\partial^{-1}))$ for every $k\geq0$,
while
$\mc F[[\partial^{-1}]]\partial^k\subset\mc F((\partial^{-1}))$
is an associative (hence Lie) subalgebra of $\mc F((\partial^{-1}))$ for every $k\leq0$,
and it is a Lie subalgebra of $\mc F((\partial^{-1}))$ for $k\leq1$.
For arbitrary $k\in\mb Z$, we have the direct sum decomposition 
as left $\mc F$-modules
\begin{equation}\label{20180416:eq1}
\mc F((\partial^{-1}))
=
\mc F[\partial]\partial^k\oplus\mc F[[\partial^{-1}]]\partial^{k-1}
\,,
\end{equation}
and we denote by 
$\Pi_{\geq k}:\,\mc F((\partial^{-1}))\twoheadrightarrow\mc F[\partial]\partial^k$
and 
$\Pi_{<k}:\,\mc F((\partial^{-1}))\twoheadrightarrow\mc F[[\partial^{-1}]]\partial^{k-1}$
the corresponding projection maps.
Hence, according to Example \ref{ex:R}
we have the following $R$-matrices 
over the Lie algebra $\mf g=\mc F((\partial^{-1}))\otimes A$:
\begin{enumerate}[(i)]
\item
$R^{(0)}=(\Pi_{\geq0}-\Pi_{<0})\otimes\id\,\big(=\frac12(R^{(0)}-(R^{(0)})^*)\big)$;
\begin{equation}\label{eq:affineR}
\vspace{-20pt}
\end{equation}
\item
$R^{(1)}=(\Pi_{\geq1}-\Pi_{<1})\otimes\id$;
\item
$R^{(2)}=\Pi_{\geq2}-\Pi_{<2}$, for $A=\mb F$.
\end{enumerate}
Note that only the first of these three examples is such that $\frac12(R-R^*)$
is an $R$-matrix.
These examples of $R$-matrices have been considered in \cite{KO93}.
\begin{remark}\label{rem:affineR}
For $k=1,2$,
we could also replace the subalgebras 
$\mc F[\partial]\partial^k$ and $\mc F[[\partial^{-1}]]\partial^{k-1}$,
in the decomposition \eqref{20180416:eq1}
with their adjoints $\partial^k\circ \mc F[\partial]$ and $\partial^{k-1}\circ \mc F[[\partial^{-1}]]$,
to get two new $R$-matrices.
\end{remark}

Throughout the remainder of Section \ref{sec:4}
we shall assume that $R$
is one of the $R$-matrices $R^{(0)}$, $R^{(1)}$, $R^{(2)}$ listed above,
and we will focus most of our attention to the case $R=R^{(0)}$.

Note that in all examples \eqref{eq:affineR}(i)-(iii) $R$
acts as the identity on the $A$-factor of $\mf g=\mc F((\partial^{-1}))\otimes A$, 
and it is left $\mc F$-linear
($f\in\mc F$, $P(\partial)\in\mc F((\partial^{-1}))$, $X\in A$):
\begin{equation}\label{eq:black2+}
R(fP(\partial) X)
=
fR(P(\partial)) X
\,.
\end{equation}
As a consequence, the dual $R^*$ 
(with respect to the pairing $\langle\cdot\,\circ\,\cdot\rangle$)
is right $\mc F$-linear:
\begin{equation}\label{eq:black3}
R^*(P(\partial)\circ f X)
=
R^*(P(\partial))\circ f\,X
\,.
\end{equation}
Moreover we have
\begin{equation}\label{eq:black4}
R(\partial^n\id)=r_n\partial^n\id\,\in\mb F\partial^n \otimes A
\,,
\quad\text{for }n\in\mb Z
\,.
\end{equation}
(In fact, $r_n=\pm1$ for all $n$, in all examples \eqref{eq:affineR}(i)-(iii).)

%
%

Clearly, the constants $r_n$, $n\in\mb Z$, uniquely determine $R$,
and $R$ can be uniquely extended, by left $\mc V_\infty\yhat$-linearity,
to a map
$R\,:\,\,
\mc V_\infty((\partial^{-1}))\xhat\otimes A\,\longrightarrow\,\mc V_\infty((\partial^{-1}))\xhat\otimes A$,
or, in terms of symbols,
to a map
\begin{equation}\label{eq:black7}
R_\xi\,:\,\,
\mc V_\infty((\xi^{-1}))\xhat\otimes A\,\longrightarrow\,\mc V_\infty((\xi^{-1}))\xhat\otimes A
\,.
\end{equation}

Let us introduce the formal $\delta$-function
\begin{equation}\label{eq:delta}
\delta(z-w)=\sum_{n\in\mb Z}z^nw^{-n-1}
\,.
\end{equation}
Recall that it is defined by the following properties
\begin{equation}\label{eq:delta1}
a(z)\delta(z-w)
=
a(w)\delta(z-w)
\,,
\end{equation}
and 
\begin{equation}\label{eq:delta2}
\Res_za(z)\delta(z-w)
=
a(w)
\,.
\end{equation}
In the sequel we consider the following generating series for the pseudodifferential operators in \eqref{eq:black4} (and their symbols):
\begin{equation}\label{Rdelta}
R_\xi(\delta(z-\xi))=\sum_{n\in\mb Z} r_n
z^{-n-1} \id \in\mb F((\xi^{-1}))[[z,z^{-1}]]\otimes A
\,,
\end{equation}
and similarly for $R^*_\xi(\delta(z-\xi))$. Using the property \eqref{eq:delta1} of the $\delta$-function, the properties \eqref{eq:black2+}-\eqref{eq:black3} of $R$ and $R^*$,
and denoting by $\lambda$ the action of $\partial$ on $\Theta$,
we have, for $P(\partial),Q(\partial)\in\mc V_\infty((\partial^{-1}))\xhat$:
\begin{equation}\label{20180416:eq2}
\begin{array}{l}
\displaystyle{
\vphantom{\Big(}
R\big(P(\partial)\delta(z-\partial)\circ\Theta Q(\partial)\big)
=
R\big(P(z)\big(\big|_{\lambda=\partial}\Theta\big)
\big(\big|_{\zeta=\partial}Q^*(\lambda-z)\big)\delta(z-\lambda-\zeta-\partial)\big)
} \\
\displaystyle{
\vphantom{\Big(}
=
\big(\big|_{\lambda=\partial}\Theta\big)
P(z)
R_\xi\big(
\delta(z-\lambda-\zeta-\xi)\big)
\big(\big|_{\zeta=\partial}Q^*(\lambda-z)\big)
\big|_{\xi=\partial}
\,,}
\end{array}
\end{equation}
and
\begin{equation}\label{20180416:eq3}
\begin{array}{l}
\displaystyle{
\vphantom{\Big(}
R^*\big(P(\partial)\delta(z-\partial)\circ\Theta Q(\partial)\big)
=
R^*\big(
\delta(z+\zeta-\partial)\circ
\big(\big|_{\zeta=\partial}P(z)\big)
\big(\big|_{\lambda=\partial}\Theta\big)
Q^*(\lambda-z)
\big)
} \\
\displaystyle{
\vphantom{\Big(}
=
\big(\big|_{\lambda=\partial}\Theta\big)
\big(\big|_{\zeta=\partial}P(z)\big)
R^*_\xi\big(
\delta(z+\zeta-\xi)\big)
\big|_{\xi=\zeta+\lambda+\partial}
\circ
Q^*(\lambda-z)
\,.}
\end{array}
\end{equation}
Here and further, when negative powers of a sum of variables appear,
if we do not specify how to expand,
it means that there is a unique way to make sense of it
requiring that when $\partial$ acts on functions
it can appear only in non-negative powers;
the same for the symbol $\lambda$ in equations \eqref{20180414:eq4} and below,
or $\mu$ in Section \ref{sec:4.3} and below.
For example, in the RHS of \eqref{20180416:eq2}
when negative powers of $z-\lambda-\zeta$ appear,
they must be expanded in the region $|z|>|\lambda+\zeta|$,
since $\lambda$ acts as $\partial$ applied to $\Theta$ and $\zeta$ acts as $\partial$ applied
to the coefficients of $Q^*(\lambda-z)$.
Finally, in terms of generating series, we have the following relation between $R$ and $R^*$.
\begin{lemma}\label{20181105:lem1a}
Let $R$ be one of the $R$-matrices $R^{(0)}$, $R^{(1)}$ or $R^{(2)}$ from \eqref{eq:affineR}.
Then
\begin{equation}\label{20181105:eq1}
R_w(\delta(z-w))=R^*_z(\delta(z-w))
\,.
\end{equation}
\end{lemma}
\begin{proof}
Of course, one can prove the claim by writing explicitly $R_w(\delta(z-w))$ separately in the three
cases $R^{(0)}$, $R^{(1)}$ and $R^{(2)}$.
We provide here a unified arguments which only uses the properties
\eqref{eq:black2+}-\eqref{eq:black4} of these $R$-matrices.
Let $P(\partial)X,Q(\partial)Y\in\mc V_\infty((\partial^{-1}))\xhat\otimes A$. Using the definition \eqref{eq:motiv11} 
of the trace form $\langle\,\cdot\,\rangle$ we have
\begin{equation}
\begin{split}\label{20181105:eq2}
&\langle R(P(\partial)\delta(z-\partial))X\circ Q(\partial)\delta(w-\partial)Y\rangle
\\
&=\int \res_\partial R(P(\partial)\delta(z-\partial))\circ Q(\partial)\delta(w-\partial)\Tr(XY)
\\
&
=\int\res_\partial P(z)R_\xi(\delta(z-\xi))\big|_{\xi=\partial}\circ Q(\partial)\delta(w-\partial)\Tr(XY)
\\
&=\int P(z)R_{\xi}(\delta(z-\xi))\big|_{\xi=w+\partial}Q(w)\Tr(XY)
\,.
\end{split}
\end{equation}
In the second identity we used equation \eqref{20180416:eq2} and in the last identity we used equation \eqref{eq:delta2}.
Moreover, using again the definition \eqref{eq:motiv11} 
of the trace form $\langle\,\cdot\,\rangle$, we also have 
\begin{equation}
\begin{split}\label{20181105:eq3}
&\langle P(\partial)\delta(z-\partial)X\circ R^*(Q(\partial)\delta(w-\partial))Y\rangle
\\
&=\int \res_\partial P(\partial)\delta(z-\partial)\circ R^*(Q(\partial)\delta(w-\partial))\Tr(XY)
\\
&
=\int\res_\partial P(\partial)\delta(z-\partial)\circ (\big|_{\zeta=\partial}Q(w))R_\xi^*(\delta(w+\zeta-\xi))\big|_{\xi=\zeta+\partial}\Tr(XY)
\\
&=\int P(z)R^*_z(\delta(z-w-\xi))\big|_{\xi=\partial}Q(w)
\Tr(XY)\,.
\end{split}
\end{equation}
In the second identity we used equation \eqref{20180416:eq3} and in the last identity we used equation \eqref{eq:delta2}
and the facts that $\delta(w+\xi-z)=\delta(z-w-\xi)$ and $\partial\left(R^*_z(\delta(z-w-\xi))\right)=0$.
From equations \eqref{20181105:eq2} and \eqref{20181105:eq3}, the definition of adjoint operator 
and the non-degeneracy of the trace we get that
$$ 
\int P(z)\left( R_{\xi}(\delta(z-\xi))\big|_{\xi=w+\partial} - R^*_z(\delta(z-w-\xi))\big|_{\xi=\partial}\right) Q(w)=0\,,
$$
for every $P(\partial),Q(\partial)\in\mc V_\infty((\partial^{-1}))\xhat$. Hence, equation
\eqref{20181105:eq1} follows by \cite[Lemma 1.36]{BDSK09}.

\end{proof}

\subsection{AOR Poisson brackets}

Note that, if $L$ is as in \eqref{eq:motiv13}, then
\begin{equation}\label{eq:black9}
L+\epsilon\id
=
\sum_{p\in\mb Z,\alpha\in I}
(x_{p,\alpha}(x)+\epsilon\,\eta_{p,\alpha})\partial^{-p-1}E^\alpha
\,,
\end{equation}
where $\eta_{p,\alpha}``=\frac1{\text{Vol}(M)}\langle\partial^pE_\alpha\rangle$ '' is defined by the property
\begin{equation}\label{eq:black10}
\sum_{p\in\mb Z,\alpha\in I}
\eta_{p,\alpha}\partial^{-p-1}E^\alpha=1
\,,
\end{equation}
the unit element of $\mf g$ (given by the tensor product of the function $1\in\mc F((\partial^{-1}))$
and the unit element $\id\in A$).
We next compute the O-R $\epsilon$-Poisson brackets $\{\int f,\int g\}^{R,\epsilon}(L(x))$,
corresponding to \eqref{eq:ragn3},
for one of the $R$-matrices \eqref{eq:affineR}(i)-(iii).
In order to compute the RHS of \eqref{eq:ragn3}
we use the definition \eqref{eq:motiv11} 
of the trace form $\langle\,\cdot\,\rangle$ on $\mc F((\partial^{-1}))\otimes A$,
and $d_L\tint f$ and $d_L\tint g$ as in \eqref{eq:motiv19}.
As a result, we get
\begin{equation}\label{eq:black11}
\begin{split}
\{\tint f,\tint g\}^{R,\epsilon}(L(x))
& =
\frac12
\Big\langle L \circ [d_L(\tint f),R((L+\epsilon\id)\circ d_L(\tint g)\circ (L+\epsilon\id))] \Big\rangle
\\
& -
\frac12
\Big\langle L \circ [d_L(\tint g),R((L+\epsilon\id)\circ d_L(\tint f)\circ (L+\epsilon\id))] \Big\rangle
\\
& =
\frac12
\Big\langle d_L(\tint g)\circ (L+\epsilon\id) \circ R^*[L ,d_L(\tint f)] \circ (L+\epsilon\id) \Big\rangle
\\
& +
\frac12
\Big\langle d_L(\tint g) \circ [L , R((L+\epsilon\id)\circ d_L(\tint f)\circ (L+\epsilon\id))] \Big\rangle
\,,
\end{split}
\end{equation}
where we used the cyclic property of the trace form.
By expanding $d_L(\tint f)$ and $d_L(\tint g)$ as in \eqref{eq:motiv19},
the RHS of \eqref{eq:black11} becomes
\begin{align*}
\frac12
\sum_{\substack{p,q\in\mb Z \\ \alpha,\beta\in I}}
\Big\langle
& \partial^q\circ
\frac{\delta\tint g}{\delta u_{q,\beta}}(L(x))
E_\beta
\circ
\Big(
(L+\epsilon\id) \circ R^*
\Big[L ,
\partial^p\circ \frac{\delta\tint f}{\delta u_{p,\alpha}}(L(x)) E_\alpha
\Big] \circ (L+\epsilon\id) 
\\
& +
\Big[L , 
R\Big(
(L+\epsilon\id)\circ 
\partial^p\circ \frac{\delta\tint f}{\delta u_{p,\alpha}}(L(x)) E_\alpha
\circ (L+\epsilon\id)
\Big)
\Big] 
\Big)
\Big\rangle
\,.
\end{align*}
Recall the general formula \eqref{eq:motiv22} relating the Poisson bracket 
on $\mc V_\infty\yhat/\partial\mc V_\infty\yhat$
to the PVA $\lambda$-bracket on $\mc V_\infty\yhat$.
Recall also that, in the identification \eqref{eq:motiv0b} of local functionals on $\mf g^*$
with elements of $\mc V_\infty\yhat/\partial\mc V_\infty\yhat$,
we simply replace the coordinate functions $x_{p,\alpha}(x)\in\mc F$ 
with the corresponding differential variables $u_{p,\alpha}\in\mc V_\infty\yhat$.
We therefore expand $L+\epsilon\id$ as in \eqref{eq:black9}
and use the definition \eqref{eq:motiv11} of the trace form,
to deduce the following formula defining the $\lambda$-brackets on $\mc V_\infty\yhat$
corresponding to the O-R Poisson brackets \eqref{eq:ragn3}
(where $\Theta\in\mc V_\infty\yhat$):
\begin{equation}\label{eq:bblack}
\begin{array}{l}
\displaystyle{
\vphantom{\Big(}
\{{u_{p,\alpha}} _\partial {u_{q,\beta}}\}^{R,\epsilon}_\to \Theta
=
\frac12
\sum_{\substack{i,j,k\in\mb Z \\ \gamma,\delta,\zeta\in I}}
} \\
\displaystyle{
\vphantom{\Big(}
\Bigg\{
\res_\partial\Big(
u_{j,\delta}\!(\epsilon)\partial^{-j-1}
\!\!
\circ
R^*\Big(
u_{k,\zeta}\partial^{-k-1+p}
\!\!\circ
\Theta
\Big)
\!\circ\!
u_{i,\gamma}\!(\epsilon)\partial^{-i-1+q}
\Big)
\Tr\big(
E_\beta
E^\delta
E^\zeta
E_\alpha
E^\gamma
\big)
} \\
\displaystyle{
\vphantom{\Big(}
-
\res_\partial\!\Big(
u_{j,\delta}\!(\epsilon)\partial^{-j-1}
\!\!\circ
R^*\Big(
\partial^p
\!\circ
\Theta\,
u_{k,\zeta}\partial^{-k-1}\!
\Big)
\!\circ\!
u_{i,\gamma}\!(\epsilon)\partial^{-i-1+q}
\Big)
\!\Tr\!\big(
E_\beta
E^\delta
E_\alpha
E^\zeta
E^\gamma
\big)
} \\
\displaystyle{
\vphantom{\Big(}
+
\res_\partial\!\Big(
u_{k,\zeta}\partial^{-k-1}\!\!
\circ
R\Big(
u_{i,\gamma}\!(\epsilon)\partial^{-i-1+p}\!\!
\circ
\Theta
\circ
u_{j,\delta}\!(\epsilon)\partial^{-j-1}
\Big)
\partial^q
\Big)
\Tr\big(
E_\beta
E^\zeta
E^\gamma
E_\alpha
E^\delta
\big)
} \\
\displaystyle{
\vphantom{\Big(}
-
\res_\partial\!\Big(
R\Big(
u_{i,\gamma}\!(\epsilon)\partial^{-i-1+p}
\!\!\circ
\Theta
u_{j,\delta}\!(\epsilon)\partial^{-j-1}
\Big)
\!\circ\!
u_{k,\zeta}\partial^{-k-1+q}
\Big)
\Tr\big(
E_\beta
E^\gamma
E_\alpha
E^\delta
E^\zeta
\big)
\Bigg\}
}
\end{array}
\end{equation}
where we introduced the notation $u_{i,\gamma}\!(\epsilon)=u_{i,\gamma}+\epsilon\eta_{i,\gamma}$.
Here we used that $R$ acts as the identity on the second factor of $\mf g=\mc F((\partial^{-1}))\otimes A$,
as remarked before equation \eqref{eq:black2+}.

\begin{remark}\label{20180419:rem}
From equation \eqref{eq:bblack}
it is not clear how to check the continuity condition \eqref{20180418:eq7}
only using the properties \eqref{eq:black2+}-\eqref{eq:black4}.
In fact, 
we will be able to check continuity only using the explicit expressions
of $R_z(\delta(z-w)$ for the three $R$-matrices $R^{(0)}$, $R^{(1)}$ and $R^{(2)}$
that we are considering in the present section.
\end{remark}

\subsection{Generating series and $\epsilon$-Adler identities}

We encode all the variables $u_{p,\alpha}\in\mc V_\infty\yhat$ in a generating series
as follows:
\begin{equation}\label{20180414:eq2}
L(z)
=
\sum_{p\in\mb Z,\alpha\in I}
u_{p,\alpha}z^{-p-1} E^\alpha
\,\in\mc V_\infty((z^{-1}))\xhat\otimes A
\,.
\end{equation}
Then, all $\lambda$-brackets $\{{u_{p,\alpha}} _\lambda {u_{q,\beta}}\}$
are encoded in 
$$
\{L_1(z)_\lambda L_2(w)\}^{R,\epsilon}\,\in\,
(\mc V_\infty\yhat[[\lambda]])[[z,z^{-1},w,w^{-1}]]\otimes A^{\otimes2}
\,,
$$
where we use the notation \eqref{eq:yang-notation} for $L_1(z)$ and $L_2(w)$.
In fact, multiplying both sides of \eqref{eq:bblack}
by $z^{-p-1}w^{-q-1}E^\alpha\otimes E^\beta$
and summing over $p,q\in\mb Z$ and $\alpha,\beta\in I$, we get
\begin{equation}\label{20180414:eq3}
\begin{array}{l}
\displaystyle{
\vphantom{\Big(}
\{L_1(z)_\partial L_2(w)\}^{R,\epsilon}_\to\Theta
=
\sum_{p,q\in\mb Z,\alpha,\beta\in I}
\{{u_{p,\alpha}} _\partial {u_{q,\beta}}\}^{R,\epsilon}_\to\Theta\,
z^{-p-1}w^{-q-1}E^\alpha\otimes E^\beta
} \\
\displaystyle{
\vphantom{\Big(}
=
\frac12
\Omega
\res_\partial\bigg(
(L_1(\partial)+\epsilon\id)
\circ
R^*\big(
L_1(\partial)
\delta(z-\partial)
\circ
\Theta
\big)
\circ
(L_2(\partial)+\epsilon\id)
\delta(w-\partial)
} \\
\displaystyle{
\vphantom{\Big(}
-
(L_1(\partial)+\epsilon\id)
\circ
R^*\Big(
\delta(z-\partial)
\circ
\Theta\,
L_2(\partial)
\Big)
\circ
(L_2(\partial)+\epsilon\id)
\delta(w-\partial)
} \\
\displaystyle{
\vphantom{\Big(}
+
L_1(\partial)
\circ
R\Big(
(L_1(\partial)+\epsilon\id)
\delta(z-\partial)
\circ
\Theta
\circ
(L_2(\partial)+\epsilon\id)
\Big)
\delta(w-\partial)
} \\
\displaystyle{
\vphantom{\Big(}
-
R\Big(
(L_1(\partial)+\epsilon\id)
\circ
\delta(z-\partial)
\circ
\Theta
(L_2(\partial)+\epsilon\id)
\Big)
\circ
L_2(\partial)
\delta(w-\partial)
\bigg)
\,.}
\end{array}
\end{equation}
Here, $\id$ stands for $\id\otimes\id\in A\otimes A$,
and we denote
\begin{equation}\label{eq:omega}
\Omega=\sum_{\alpha\in I}E_\alpha\otimes E^\alpha
\,\in A\otimes A
\,.
\end{equation}
It satisfies the following basic property
\begin{equation}\label{eq:switch}
\Omega(X\otimes Y)=(Y\otimes X)\Omega
\,\,\text{ for all }\,\,
X,Y\in A
\,,
\end{equation}
which is easily checked by \eqref{eq:motiv7}.
Moreover, for $a=A(\partial)X,b=B(\partial)Y\in\mc V_\infty\yhat$, we are denoting $a_1\circ b_2=A(\partial)B(\partial)\otimes X\otimes Y$.
Using \eqref{20180416:eq2}, \eqref{20180416:eq3} and \eqref{20181105:eq1}
we can rewrite \eqref{20180414:eq3} as
\begin{equation}\label{20180414:eq4}
\begin{array}{l}
\displaystyle{
\vphantom{\Big(}
\{L_1(z)_\lambda L_2(w)\}^{R,\epsilon}
} \\
\displaystyle{
\vphantom{\Big(}
=
\frac12
\Omega
\bigg(
(L_1(w\!+\!\lambda\!+\!\partial)\!+\!\epsilon\id)
\big(\big|_{\zeta=z+\partial}\!L_1(z)\big)
R_\zeta\big(
\delta(\zeta\!-\!\xi)\big)
(\big|_{\xi=\zeta\!-\!z+\!w\!+\!\lambda\!+\!\partial}
L_2(w)\!+\!\epsilon\id)
} \\
\displaystyle{
\vphantom{\Big(}
-
(L_1(w+\lambda+\partial)+\epsilon\id)
R_z\big(
\delta(z-\xi)\big)
\big|_{\xi=w+\lambda+\partial}
L_2^*(\lambda-z)
(L_2(w)+\epsilon\id)
} \\
\displaystyle{
\vphantom{\Big(}
+
L_1(w+\lambda+\partial)
(L_1(z)+\epsilon\id)
R_w\big(
\delta(\zeta-w)\big)
\big(\big|_{\zeta=z-\lambda-\partial}
L_2^*(\lambda-z)+\epsilon\id\big)
} \\
\displaystyle{
\vphantom{\Big(}
-
(L_1(z)+\epsilon\id)
R_\xi\big(
\delta(\zeta-\xi)\big)
\big(\big|_{\zeta=z-\lambda-\partial}
L_2^*(\lambda-z)+\epsilon\id\big)
\big|_{\xi=w+\partial}
L_2(w)
\bigg)
\,.}
\end{array}
\end{equation}
This formula, that we call the $\epsilon$-\emph{Adler identity} associated to the $R$-matrix $R$,
encodes the whole PVA structure of $\mc V_\infty\yhat$
associated to the affine analogue of the O-R Poisson brackets.

If we expand as in \eqref{20180408:eq2}:
\begin{equation}\label{20180408:eq2b}
\{\cdot\,_\lambda\,\cdot\}^{R,\epsilon}
=
\{\cdot\,_\lambda\,\cdot\}^{R}_3
+2\epsilon \{\cdot\,_\lambda\,\cdot\}^{R}_2
+\epsilon^2\{\cdot\,_\lambda\,\cdot\}^{R}_1
\,,
\end{equation}
we get the $3$-\emph{Adler identity}
\begin{equation}\label{eq:3adler}
\begin{array}{l}
\displaystyle{
\vphantom{\Big(}
\{L_1(z)_\lambda L_2(w)\}^{R}_3
} \\
\displaystyle{
\vphantom{\Big(}
=
\frac12
\Omega
\bigg(
L_1(w\!+\!\lambda\!+\!\partial)
\big(\big|_{\zeta=z+\partial}L_1(z)\big)
R_\zeta\big(
\delta(\zeta-\xi)\big)
\big|_{\xi=\zeta-z+w+\lambda+\partial}
L_2(w)
} \\
\displaystyle{
\vphantom{\Big(}
-
L_1(w+\lambda+\partial)
R_z\big(
\delta(z-\xi)\big)
\big|_{\xi=w+\lambda+\partial}
L_2^*(\lambda-z)
L_2(w)
} \\
\displaystyle{
\vphantom{\Big(}
+
L_1(w+\lambda+\partial)
L_1(z)
R_w\big(
\delta(\zeta-w)\big)
\big|_{\zeta=z-\lambda-\partial}
L_2^*(\lambda-z)
} \\
\displaystyle{
\vphantom{\Big(}
-
L_1(z)
R_\xi\big(
\delta(\zeta-\xi)\big)
\big(\big|_{\zeta=z-\lambda-\partial}
L_2^*(\lambda-z)\big)
\big|_{\xi=w+\partial}
L_2(w)
\bigg)
\,.}
\end{array}
\end{equation}
the $2$-Adler identity:
\begin{equation}\label{eq:2adler}
\begin{array}{l}
\displaystyle{
\vphantom{\Big(}
\{L_1(z)_\lambda L_2(w)\}^{R}_2
} \\
\displaystyle{
\vphantom{\Big(}
=
\frac14
\Omega
\bigg(
L_1(w\!+\!\lambda\!+\!\partial)
\big(\big|_{\zeta=z+\partial}L_1(z)\big)
R_\zeta\big(
\delta(\zeta-\xi)\big)
\big|_{\xi=\zeta-z+w+\lambda}
} \\
\displaystyle{
\vphantom{\Big(}
+
\big(\big|_{\zeta=z+\partial}L_1(z)\big)
R_\zeta\big(
\delta(\zeta-\xi)\big)
\big|_{\xi=\zeta-z+w+\lambda+\partial}
L_2(w)
} \\
\displaystyle{
\vphantom{\Big(}
-
L_1(w+\lambda+\partial)
R_z\big(
\delta(z-\xi)\big)
\big|_{\xi=w+\lambda+\partial}
L_2^*(\lambda-z)
} \\
\displaystyle{
\vphantom{\Big(}
-
R_z\big(
\delta(z-\xi)\big)
\big|_{\xi=w+\lambda+\partial}
L_2^*(\lambda-z)
L_2(w)
} \\
\displaystyle{
\vphantom{\Big(}
+
L_1(w+\lambda+\partial)
L_1(z)
R_w\big(
\delta(\zeta-w)\big)\big|_{\zeta=z-\lambda}
} \\
\displaystyle{
\vphantom{\Big(}
+
L_1(w+\lambda+\partial)
R_w\big(
\delta(\zeta-w)\big)
\big(\big|_{\zeta=z-\lambda-\partial}
L_2^*(\lambda-z)\big)
} \\
\displaystyle{
\vphantom{\Big(}
-
L_1(z)
R_\xi\big(
\delta(\zeta-\xi)\big)\big|_{\zeta=z-\lambda}
\big|_{\xi=w+\partial}
L_2(w)
} \\
\displaystyle{
\vphantom{\Big(}
-
R_\xi\big(
\delta(\zeta-\xi)\big)
\big(\big|_{\zeta=z-\lambda-\partial}
L_2^*(\lambda-z)\big)
\big|_{\xi=w+\partial}
L_2(w)
\bigg)
\,,}
\end{array}
\end{equation}
and the $1$-Adler identity:
\begin{equation}\label{eq:1adler}
\begin{array}{l}
\displaystyle{
\vphantom{\Big(}
\{L_1(z)_\lambda L_2(w)\}^{R}_1
=
\frac12
\Omega
} \\
\displaystyle{
\vphantom{\Big(}
\times\!\bigg(\!
\big(\big|_{\zeta=z+\partial}
L_1(z)
\big)
R_\zeta\big(
\delta(\zeta\!-\!\xi)\big)
\big|_{\xi=\zeta-z+w+\lambda}
-
R_z\big(
\delta(z-\xi)\big)
\big|_{\xi=w+\lambda+\partial}
L_2^*(\lambda-z)
} \\
\displaystyle{
\vphantom{\Big(}
+
L_1(w+\lambda+\partial)
R_w\big(
\delta(\zeta-w)
\big)\big|_{\zeta=z-\lambda}
-
R_\xi\big(
\delta(\zeta-\xi)\big)
\big|_{\zeta=z-\lambda}
\big|_{\xi=w+\partial}
L_2(w)
\bigg)
\,.}
\end{array}
\end{equation}

\subsection{The Adler identities for the standard $R$-matrix $R=R^{(0)}$}

Next, we specialize the Adler identities \eqref{eq:3adler}-\eqref{eq:1adler}
for the $R$-matrix $R^{(0)}$ in \eqref{eq:affineR}(i).

Recall that the $\delta$-function \eqref{eq:delta} admits the decomposition
\begin{equation}\label{eq:delta3}
\delta(z-w)
=
\iota_z(z-w)^{-1}
-
\iota_w(z-w)^{-1}
\,,
\end{equation}
where $\iota_z$ denotes the geometric expansion in the domain $|z|>>0$,
i.e. $\iota_z(z-w)^{-1}=\sum_{n\geq0}z^{-n-1}w^n$,
while $\iota_w$ denotes the geometric expansion in the domain $|w|>>0$,
i.e. $\iota_w(z-w)^{-1}=-\sum_{n\geq0}z^{n}w^{-n-1}$.

For $R=R^{(0)}=\Pi_{\geq0}-\Pi_{<0}$, we have
\begin{equation}\label{20180416:eq4}
R^{(0)}=-(R^{(0)})^*
\,\text{ and }\,
R^{(0)}_w(\delta(z-w))=
\iota_z(z-w)^{-1}+\iota_w(z-w)^{-1}
\,.
\end{equation}
Hence, in this case the $\epsilon$-Adler identity \eqref{20180414:eq4} becomes
\begin{equation}\label{eq:3adler-eps}
\begin{array}{l}
\displaystyle{
\vphantom{\Big(}
\{L_1(z)_\lambda L_2(w)\}^{R^{(0)},\epsilon}
} \\
\displaystyle{
\vphantom{\Big(}
=
\Omega
\bigg(
(L_1(w+\lambda+\partial)+\epsilon\id)
(z-w-\lambda-\partial)^{-1}
L_2^*(\lambda-z)(L_2(w)+\epsilon\id)
} \\
\displaystyle{
\vphantom{\Big(}
-(L_1(w+\lambda+\partial)+\epsilon\id)
L_1(z)
(z-w-\lambda-\partial)^{-1}
(L_2(w)+\epsilon\id)
} \\
\displaystyle{
\vphantom{\Big(}
+
L_1(w+\lambda+\partial)
(L_1(z)+\epsilon\id)
(z-w-\lambda-\partial)^{-1}
(L_2^*(\lambda-z)+\epsilon\id)
}
\\
\displaystyle{
\vphantom{\Big(}
-
(L_1(z)+\epsilon\id)
(z-w-\lambda-\partial)^{-1}
(L_2^*(\lambda-z)+\epsilon\id)
L_2(w)
\bigg)
\,.}
\end{array}
\end{equation}
Here $(z-w-\lambda-\partial)^{-1}$
can be interpreted as either its $\iota_z$ expansion, or its $\iota_w$ expansion:
both choices give the same answer.
Indeed, if we replace everywhere $(z-w-\lambda-\partial)^{-1}$
by $\delta(z-w-\lambda-\partial)$,
the RHS of \eqref{eq:3adler-0} vanishes by \eqref{eq:delta1}.
More explicitly, the $3$-Adler identity \eqref{eq:3adler} becomes
\begin{equation}\label{eq:3adler-0}
\begin{array}{l}
\displaystyle{
\vphantom{\Big(}
\{L_1(z)_\lambda L_2(w)\}^{(0)}_3
} \\
\displaystyle{
\vphantom{\Big(}
=
\Omega
\bigg(
L_1(w\!+\!\lambda\!+\!\partial)
L_1(z)
(z-w-\lambda-\partial)^{-1}
\big(
L_2^*(\lambda-z)-L_2(w)
\big)
} \\
\displaystyle{
\vphantom{\Big(}
+
\big(
L_1(w+\lambda+\partial)-L_1(z)
\big)
(z-w-\lambda-\partial)^{-1}
L_2^*(\lambda-z)
L_2(w)
\bigg)
\,.}
\end{array}
\end{equation}
Similarly, the $2$-Adler identity \eqref{eq:2adler} becomes
\begin{equation}\label{eq:2adler-0}
\begin{array}{l}
\displaystyle{
\vphantom{\Big(}
\{L_1(z)_\lambda L_2(w)\}^{(0)}_2
=
\Omega
\bigg(
L_1(w\!+\!\lambda\!+\!\partial)
(z-w-\lambda-\partial)^{-1}
L_2^*(\lambda-z)
} \\
\displaystyle{
\vphantom{\Big(}
-
L_1(z)
(z-w-\lambda-\partial)^{-1}
L_2(w)
\bigg)
\,,}
\end{array}
\end{equation}
and the $1$-Adler identity \eqref{eq:1adler} becomes
\begin{equation}\label{eq:1adler-0}
\begin{array}{l}
\displaystyle{
\vphantom{\Big(}
\{L_1(z)_\lambda L_2(w)\}^{(0)}_1
=
\Omega
\bigg(
\big(
L_1(w+\lambda)-L_1(z)
\big)
(z-w-\lambda)^{-1}
} \\
\displaystyle{
\vphantom{\Big(}
+
(z-w-\lambda-\partial)^{-1}
\big(
L_2^*(\lambda-z)-L_2(w)
\big)
\bigg)
\,.}
\end{array}
\end{equation}
Equation \eqref{eq:2adler-0} is the same as the Adler identity for $\mf{gl}_N$
which first appeared in \cite{DSKV16} and \cite{DSKV17}.

\subsection{The Adler identities corresponding to $R=R^{(1)}$}

Next, we specialize the Adler identities \eqref{eq:3adler}-\eqref{eq:1adler}
for the $R$-matrix $R^{(1)}$ in \eqref{eq:affineR}(ii).
We have
\begin{equation}\label{20180416:eq5}
R^{(1)}
=
\Pi_{\geq1}-\Pi_{<1}
=
R^{(0)}-2\Pi_0
\,,
\end{equation}
where $\Pi_0:\,\mb F((z^{-1}))\to\mb F$ denotes the projection to the coefficient of $z^0$,
and
$$
(R^{(1)})^*
=
-R^{(0)}-2\Pi_{-1}
\,.
$$
Applying \eqref{20180416:eq5},
we can compute the $3$-Adler identity \eqref{eq:3adler} for $R=R^{(1)}$,
to get
\begin{equation}\label{eq:3adler-1}
\begin{array}{l}
\displaystyle{
\vphantom{\Big(}
\{L_1(z)_\lambda L_2(w)\}^{(1)}_3
=
\{L_1(z)_\lambda L_2(w)\}^{(0)}_3
} \\
\displaystyle{
\vphantom{\Big(}
+
\Omega
\bigg(
-
L_1(w+\lambda+\partial)
\iota_w(w+\lambda+\partial)^{-1}
L_1(z)L_2(w)
} \\
\displaystyle{
\vphantom{\Big(}
+
L_1(w+\lambda+\partial)
\iota_w(w+\lambda+\partial)^{-1}
L_2^*(\lambda-z)
L_2(w)
} \\
\displaystyle{
\vphantom{\Big(}
-
L_1(w+\lambda+\partial)
L_1(z)
\iota_z(z-\lambda-\partial)^{-1}
L_2^*(\lambda-z)
} \\
\displaystyle{
\vphantom{\Big(}
+
L_1(z)
\big(
\iota_z(z-\lambda-\partial)^{-1}
L_2^*(\lambda-z)
\big)
L_2(w)
\bigg)
\,.}
\end{array}
\end{equation}
Since $\frac12(R^{(1)}-(R^{(1)})^*)$ does not satisfy the modified Yang-Baxter equation \eqref{eq:mod-YB},
the corresponding $2$-Adler identity will not define a PVA structure on $\mc V_\infty\yhat$,
while the $1$-st Adler identity will. It is
\begin{equation}\label{eq:1adler-1}
\begin{array}{l}
\displaystyle{
\vphantom{\Big(}
\{L_1(z)_\lambda L_2(w)\}^{(1)}_1
=
\{L_1(z)_\lambda L_2(w)\}^{(0)}_1
} \\
\displaystyle{
\vphantom{\Big(}
+
\Omega
\Big(
-
\iota_w(w+\lambda+\partial)^{-1}
\big(L_1(z)
-
L_2^*(\lambda-z)
\big)
} \\
\displaystyle{
\vphantom{\Big(}
+
\iota_z(z-\lambda)^{-1}
\big(
L_2(w)
-
L_1(w+\lambda)
\big)
\Big)
\,.}
\end{array}
\end{equation}

In a similar way one can compute the Adler identities corresponding to the $R$-matrix $R=R^{(2)}$
from \eqref{eq:affineR} (in the scalar case $A=\mb F$).
We leave this exercise to the interested reader.

\subsection{The $\epsilon$-Adler identities and the corresponding continuous 
PVA $\lambda$-brackets on $\mc V_\infty\!\!\widehat{}$\,\,}
\label{sec:4.3}

As before, throughout this section we let $R$ be one of the $R$-matrices $R^{(0)},\,R^{(1)},\,R^{(2)}$ 
defined in \eqref{eq:affineR}.
In fact, apart for the proof of the continuity of the $\lambda$-bracket in Proposition \ref{prop:cont}
(where we use the explicit expression for $R_z(\delta(z-w))$),
all other arguments only use properties \eqref{eq:black2+}-\eqref{eq:black4}.
\begin{proposition}\label{prop:cont}
The $\epsilon$-Adler identity \eqref{20180414:eq4} associated to $R$ defines a continuous
$\lambda$-bracket on $\mc V_\infty\yhat$, for every $\epsilon\in\mb F$.
\end{proposition}

\begin{proof}
We need to prove the continuity condition \eqref{20180418:eq11}
(or, equivalently, \eqref{20180418:eq11b}). 
Fix $N\in\mb Z$.
Note that, by the definition of the projection maps 
$\pi_N:\,\mc V_\infty\to\mc V_N$ in \eqref{20180417:eq4},
$$
\pi_N(L(z))
=
\sum_{p=-N-1}^\infty u_{p,\alpha}z^{-p-1}E^\alpha
\,,
$$
has powers $z^{\leq N}$.
Applying $\pi_N$ to the RHS of \eqref{eq:3adler-eps}, we get
\begin{equation}\label{20180419:eq1}
\begin{array}{l}
\displaystyle{
\vphantom{\Big(}
\Omega
\bigg(
(\pi_N(L_1(w+\lambda+\partial))+\epsilon\id)
(z\!-\!w\!-\!\lambda\!-\!\partial)^{-1}
\pi_N(L_2^*(\lambda-z))(\pi_N(L_2(w))+\epsilon\id)
\big)
} \\
\displaystyle{
\vphantom{\Big(}
-(\pi_N(L_1(w+\lambda+\partial))+\epsilon\id)
\pi_N(L_1(z))
(z\!-\!w\!-\!\lambda\!-\!\partial)^{-1}
(\pi_N(L_2(w))+\epsilon\id)
\big)
} \\
\displaystyle{
\vphantom{\Big(}
+
\pi_N(L_1(w+\lambda+\partial))
(\pi_N(L_1(z))+\epsilon\id)
(z\!-\!w\!-\!\lambda\!-\!\partial)^{-1}
(\pi_N(L_2^*(\lambda-z))+\epsilon\id)
}
\\
\displaystyle{
\vphantom{\Big(}
-
(\pi_N(L_1(z))+\epsilon\id)
(z\!-\!w\!-\!\lambda\!-\!\partial)^{-1}
(\pi_N(L_2^*(\lambda-z))+\epsilon\id)
\pi_N(L_2(w))
\bigg)
\,.}
\end{array}
\end{equation}
If we expand $(z-w-\lambda-\partial)^{-1}$ in negative powers of $z$,
we observe that the powers of $z$ in \eqref{20180419:eq1}
are bounded above by $M=2N-1$.
If instead we expand $(z-w-\lambda-\partial)^{-1}$ in negative powers of $w$,
we get that also the powers of $w$ in \eqref{20180419:eq1}
are bounded above by $M=2N-1$.
Recall that, by the observation after formula \eqref{eq:3adler-0}, 
the RHS of \eqref{eq:3adler-0} is unchanged
if we expand $(z-w-\lambda-\partial)^{-1}$ in either negative powers of $z$
or negative powers of $w$.
Hence, the continuity condition \eqref{20180418:eq11b}
for the $\epsilon$-Adler identity \eqref{eq:3adler-eps} of $R^{(0)}$ holds.
The proof for $R^{(1)}$ and $R^{(2)}$ is similar.
\end{proof}

\begin{proposition}\label{prop:skew}
The $\epsilon$-Adler identity \eqref{20180414:eq4} associated to $R$ 
implies the skewsymmetry condition
\eqref{eq:pva5c}.
\end{proposition}
\begin{proof}
First, note that the skewsymmetry condition \eqref{eq:pva5c} can be rewritten as
\begin{equation}\label{eq:pva5c-bis}
\{L_1(z)_\lambda L_2(w)\}=-\big|_{x=\partial}\{L_1(w)_{-\lambda -x}L_2(z)\}^\sigma\,,
\end{equation}
where $\sigma$ is the endomorphism of $A^{\otimes2}$ defined by $(X\otimes Y)^\sigma=Y\otimes X$.
Using the $\epsilon$-Adler identity \eqref{20180414:eq4}, by a straightforward 
computation we have
\begin{equation}\label{20181007:eq2}
\begin{array}{l}
\displaystyle{
\vphantom{\Big(}
-\big|_{x=\partial}\{L_{1}(w)_{-\lambda-x}L_2(z)\}^{R,\epsilon}
} \\
\displaystyle{
\vphantom{\Big(}
=
\frac12
\Omega
\bigg(
-(\big|_{\xi=z-\lambda-\partial}L_1^*(\lambda -z)+\epsilon\id)
(\big|_{\zeta=w+\partial}L_1(w))
R_{\xi}\big(
\delta(\zeta-\xi)\big)
\big(
L_2(z)+\epsilon\id\big)
} \\
\displaystyle{
\vphantom{\Big(}
+
(\big|_{\zeta=\partial} L_1^*(\lambda-z)+\epsilon\id)
R_w(\delta(w-\xi))\big|_{\xi=z-\lambda-\zeta}
L_2(w+\lambda+\zeta+\partial)
(L_2(z)+\epsilon\id)
} \\
\displaystyle{
\vphantom{\Big(}
-
\big(\big|_{\xi=\partial} L_1^*(\lambda-z)(L_1(w)+\epsilon\id)\big)
R_{\zeta}\big(\delta(\zeta-z)\big)
(\big|_{\zeta=w+\lambda+\xi}
(L_2(w+\lambda+\xi)+\epsilon\id)
} \\
\displaystyle{
\vphantom{\Big(}
+
(\big|_{\zeta=w+\lambda+\partial+z-\xi} L_1(w)+\epsilon\id)
R_\xi\big(
\delta(\zeta-\xi)\big)
(L_2(\zeta)+\epsilon\id)
\big|_{\xi=z+\partial}L_2(z)
\bigg)
\,.}
\end{array}
\end{equation}
The skewsymmetry condition \eqref{eq:pva5c-bis} follows by applying $\sigma$ in both sides of
\eqref{20181007:eq2} and by using the facts that $(XY)^\sigma=X^\sigma Y^\sigma$, for every $X,Y\in A^{\otimes 2}$, and that $\Omega^\sigma=\Omega$.
\end{proof}

\begin{lemma}\label{20181105:lem1}
The modified Yang-Baxter equation \eqref{eq:mod-YB} for $R$ is equivalent to the following identity
\begin{equation}\label{mybe2}
\begin{split}
&\big(
R_v(\delta(w-v))R_{\xi}(\delta(z-\xi))\big|_{\xi=v+\mu}
-R_v(\delta(w-v))R_{\zeta}(\delta(z-\zeta))\big|_{\zeta=w+\mu}
\\
&
-R_v(\delta(z-\mu-v))R_{\eta}(\delta(w-\eta))\big|_{\eta=z-\mu}
+\delta(z-v-\mu)\delta(w-v)
\big)\Omega_{12}\Omega_{23}
\\
&
-\big(
R_v(\delta(z-v))R_{\xi}(\delta(w-\xi))\big|_{v+\lambda}
-R_v(\delta(z-v))R_{\zeta}(\delta(w-\zeta))\big|_{\zeta=z+\lambda}
\\
&
-R_v(\delta(w-\lambda-v))R_{\eta}(\delta(z-\eta))\big|_{\eta=w-\lambda}
+\delta(w-v-\lambda)\delta(z-v)
\big)\Omega_{23}\Omega_{12}=0
\,.
\end{split}
\end{equation}
\end{lemma}
\begin{proof}
Let us compute the modified Yang-Baxter equation \eqref{eq:mod-YB} for
$a=A(\partial)\delta(z-\partial)E_\alpha$ and $b=B(\partial)\delta(w-\partial)E_\beta$, for
$A(\partial),B(\partial)\in\mc V_\infty\yhat((\partial^{-1}))$ and $\alpha,\beta\in I$.
Using equation \eqref{20180416:eq2}, and recalling that $\partial\left(R_w(\delta(z-w))\right)=0$, we get the identity
\begin{equation}\label{20180712:eq1}
\begin{split}
&A(z)\big[
R_v(\delta(w-v))R_{\xi}(\delta(z-\xi))\big|_{\xi=v+\mu}
-R_v(\delta(w-v))R_{\zeta}(\delta(z-\zeta))\big|_{\zeta=w+\mu}
\\
&
\qquad-R_v(\delta(z-\mu-v))R_{\eta}(\delta(w-\eta))\big|_{\eta=z-\mu}
\\
&\qquad
+\delta(z-v-\mu)\delta(w-v)
\big]\left(\big|_{\mu=\partial}B(w)\right)\big|_{v=\partial}E_{\alpha}E_{\beta}
\\
&
-B(w)\big[
R_v(\delta(z-v))R_{\xi}(\delta(w-\xi))\big|_{\xi=v+\lambda}
-R_v(\delta(z-v))R_{\zeta}(\delta(w-\zeta))\big|_{\zeta=z+\lambda}
\\
&
\qquad
-R_v(\delta(w-\lambda-v))R_{\eta}(\delta(z-\eta))\big|_{\eta=w-\lambda}
\\
&\qquad
+\delta(w-v-\lambda)\delta(z-v)
\big]\left(\big|_{\lambda=\partial}A(z)\right)\big|_{v=\partial}E_{\beta}E_{\alpha}=0
\,.
\end{split}
\end{equation}
Tensoring both sides of identity \eqref{20180712:eq1} on the left by $E^\alpha\otimes E^\beta$ and taking the sum
over $\alpha,\beta\in I$ we get the identity
\begin{equation}\label{20180712:eq2}
\begin{split}
&A(z)\big[
R_v(\delta(w-v))R_{\xi}(\delta(z-\xi))\big|_{\xi=v+\mu}
-R_v(\delta(w-v))R_{\zeta}(\delta(z-\zeta))\big|_{\zeta=w+\mu}
\\
&
\qquad-R_v(\delta(z-\mu-v))R_{\eta}(\delta(w-\eta))\big|_{\eta=z-\mu}
\\
&\qquad
+\delta(z-v-\mu)\delta(w-v)
\big]\left(\big|_{\mu=\partial}B(w)\right)\big|_{v=\partial}\Omega_{12}\Omega_{23}
\\
&
-B(w)\big[
R_v(\delta(z-v))R_{\xi}(\delta(w-\xi))\big|_{\xi=v+\lambda}
-R_v(\delta(z-v))R_{\zeta}(\delta(w-\zeta))\big|_{\zeta=z+\lambda}
\\
&
\qquad
-R_v(\delta(w-\lambda-v))R_{\eta}(\delta(z-\eta))\big|_{\eta=w-\lambda}
\\
&\qquad
+\delta(w-v-\lambda)\delta(z-v)
\big]\left(\big|_{\lambda=\partial}A(z)\right)\big|_{v=\partial}\Omega_{23}\Omega_{12}=0
\,.
\end{split}
\end{equation}
Since identity \eqref{20180712:eq2} holds for arbitrary $A(\partial), B(\partial)\in\mc V_\infty\yhat((\partial^{-1}))$
it implies identity \eqref{mybe2}.
\end{proof}
\begin{remark}
If $A$ is non commutative, the elements $\Omega_{12}\Omega_{23}$ and $\Omega_{23}\Omega_{12}$ are linearly independent.
Hence, by Lemma \ref{20181105:lem1}, it follows that the modified
Yang-Baxter equation \eqref{eq:mod-YB} is equivalent to the identity
\begin{equation}\label{mybe3}
\begin{split}
&
R_v(\delta(z-v))R_{\xi}(\delta(w-\xi))\big|_{\xi=v+\lambda}
-R_v(\delta(z-v))R_{\zeta}(\delta(w-\zeta))\big|_{\zeta=v+\lambda}
\\
&
-R_v(\delta(w-\lambda-v))R_{\eta}(\delta(z-\eta))\big|_{\eta=w-\lambda}
+\delta(w-v-\lambda)\delta(z-v)
=0
\,.
\end{split}
\end{equation}
\end{remark}
\begin{proposition}\label{prop:jacobi}
The $\epsilon$-Adler identity \eqref{20180414:eq4} associated to $R=R^{(0)}$, $R^{(1)}$ or $R^{(2)}$ 
implies the Jacobi identity \eqref{eq:pva3c} for every $\epsilon\in\mb F$.
In particular, the $1$-st and $3$-rd Adler identities \eqref{eq:1adler} and \eqref{eq:3adler} associated to $R$ imply the Jacobi identity.
For $R=R^{(0)}$,
then also the $2$-nd Adler identity \eqref{eq:2adler} associated to $R$ implies the Jacobi identity,
and the corresponding continuous $\lambda$-brackets 
defined by \eqref{eq:3adler-0}, \eqref{eq:2adler-0} and \eqref{eq:1adler-0}
are compatible, in the sense that any their linear combination satisfies the Jacobi identity.
\end{proposition}
\begin{proof}
The proof follows by a very long but straightforward computation. We outline it in the case of the $1$-st Adler
identity \eqref{eq:1adler}. 

Recall the generating series \eqref{Rdelta}. Let us introduce the shorthand
$$
R(z,w):=R_w(\delta(z-w))\,.
$$
Using sesquilinearity \eqref{eq:pva4}, Leibniz rules \eqref{eq:pva2} and the
identity \eqref{20181105:eq1}, the Jacobi identity \eqref{eq:pva3c} for
the $\lambda$-bracket $\{L(z)_\lambda L(w)\}_1^R$ given by the  $1$-st Adler type identity
\eqref{eq:1adler} becomes
\begin{align}
\begin{split}\label{20181207:jacobi1}
&\Big[
\Omega_{12}\Omega_{23}
R(w-\mu,z)R(v+\lambda+\zeta+\mu,w)
-\Omega_{13}\Omega_{23}
R(v+\lambda+\zeta+\mu,z)R(w-\mu,v)
\\
&-\Omega_{23}\Omega_{12}
R(v+\lambda+\zeta+\mu,w+\lambda+\zeta)
R(w+\lambda+\zeta,z)
\\
&+\Omega_{23}\Omega_{13}
R(w-\mu,v+\lambda+\zeta)
R(v+\lambda+\zeta,z)
\\
&
+\Omega_{12}\Omega_{13}
R(w-\mu,z-\lambda-\mu-\zeta)R(v+\mu+\lambda+\zeta,z)
\\
&-\Omega_{13}\Omega_{12}
R(v+\lambda+\zeta+\mu,z+\mu)R(w-\mu,z)
\Big]
\big|_{\zeta=\partial}L_1^*(\lambda-z)
\\
&+\Big[
-\Omega_{12}\Omega_{23}
R(z-\lambda,w-\lambda-\mu-\xi)R(v+\lambda+\mu+\xi,w)
\\
&-\Omega_{12}\Omega_{13}
R(z-\lambda,w)R(v+\lambda+\mu+\xi,z)
+\Omega_{23}\Omega_{13}
R(v+\lambda+\mu+\xi,w)R(z-\lambda,v)
\\
&+\Omega_{13}\Omega_{12}
R(v+\lambda+\mu+\xi,z+\mu+\xi)
R(z+\mu-+\xi,w)
\\
&-\Omega_{13}\Omega_{23}
R(z-\lambda,v+\mu+\xi)R(v+\mu+\xi,w)
\\
&+\Omega_{23}\Omega_{12}
R(v+\lambda+\mu+\xi,w+\lambda)R(z-\lambda,w)
\Big]
\big|_{\xi=\partial}L_2^*(\mu-w)
\\
&+\Big[
\Omega_{13}\Omega_{23}
R(z-\lambda,v+\mu)R(w-\mu,v)
-\Omega_{23}\Omega_{13}
R(w-\mu,v+\lambda)R(z-\lambda,v)
\\
&+\Omega_{13}\Omega_{12}
R(z-\lambda,v)
R(w-\mu,z)
-\Omega_{23}\Omega_{12}
R(w-\mu,v)
R(z-\lambda,w)
\\
&-\Omega_{12}\Omega_{13}
R(w-\mu,z-\lambda-\mu)R(z-\lambda-\mu,v)
\\
&
+\Omega_{12}\Omega_{23}
R(z-\lambda,w-\lambda-\mu)R(w-\lambda-\mu,v)
\Big]
L_3(v+\lambda+\mu)
\\
&+\Big[-
R(w-\mu,z+\zeta)R(v+\lambda+\zeta+\mu,w)\Omega_{12}\Omega_{23}
\\
&
+R(v+\lambda+\zeta+\mu,z+\zeta)R(w-\mu,v)\Omega_{13}\Omega_{23}
\\
&
+R(v+\lambda+\zeta+\mu,w+\lambda+\zeta)
R(w+\lambda+\zeta,z+\zeta)\Omega_{23}\Omega_{12}
\\
&-
R(w-\mu,v+\lambda+\zeta)
R(v+\lambda+\zeta,z+\zeta)\Omega_{23}\Omega_{13}
\\
&-R(w-\mu,z-\lambda-\mu)R(v+\lambda+\zeta+\mu,z+\zeta)\Omega_{12}\Omega_{13}
\\
&+R(v+\lambda+\zeta+\mu,z+\mu+\zeta)R(w-\mu,z+\zeta)\Omega_{13}\Omega_{12}
\Big]\big|_{\zeta=\partial}L_1(z)
\\
&+\Big[
R(z-\lambda,w-\lambda-\mu)R(v+\lambda+\mu+\xi,w+\xi)
\Omega_{12}\Omega_{23}
\\
&+
R(z-\lambda,w+\xi)R(v+\lambda+\mu+\xi,z)
\Omega_{12}\Omega_{13}
\\
&-R(v+\lambda+\mu+\xi,w+\xi)R(z-\lambda,v)
\Omega_{23}\Omega_{13}
\\
&-
R(v+\lambda+\mu+\xi,z+\mu+\xi)
R(z+\mu+\xi,w+\xi)
\Omega_{13}\Omega_{12}
\\
&+
R(z-\lambda,v+\mu+\xi)R(v+\mu+\xi,w+\xi)
\Omega_{13}\Omega_{23}
\\
&-
R(v+\lambda+\mu+\xi,w+\lambda+\xi)R(z-\lambda,w+\xi)
\Omega_{23}\Omega_{12}
\Big]\big|_{\xi=\partial}L_2(w)
\\
&+\Big[
-R(z-\lambda,v+\mu+\nu)R(w-\mu,v+\nu)
\Omega_{13}\Omega_{23}
\\
&+
R(w-\mu,v+\lambda+\nu)R(z-\lambda,v+\nu)
\Omega_{23}\Omega_{13}
\\
&-
R(z-\lambda,v+\nu)
R(w-\mu,z)
\Omega_{13}\Omega_{12}
+
R(w-\mu,v+\nu)
R(z-\lambda,w)
\Omega_{23}\Omega_{12}
\\
&+
R(w-\mu,z-\lambda-\mu)R(z-\lambda-\mu,v+\nu)
\Omega_{12}\Omega_{13}
\\
&-
R(z-\lambda,w-\lambda-\mu)R(w-\lambda-\mu,v+\nu)
\Omega_{12}\Omega_{23}
\Big]\big|_{\nu=\partial}L_3(v)
=0
\,.
\end{split}
\end{align}
Note that
\begin{equation}\label{20181208:eq1}
\Omega_{12}\Omega_{23}=\Omega_{13}\Omega_{12}=\Omega_{23}\Omega_{13}\,,
\qquad
\Omega_{23}\Omega_{12}=\Omega_{13}\Omega_{23}=\Omega_{12}\Omega_{13}\,.
\end{equation}
Furthermore, let us rewrite the identity \eqref{mybe2} as
\begin{equation}\label{mybe2bisse}
\begin{split}
\Gamma(z,w,v,\lambda,\mu)
=
\delta(w-v-\lambda)\delta(z-v)\Omega_{23}\Omega_{12}
-\delta(z-v-\mu)\delta(w-v)
\Omega_{12}\Omega_{23}
\,,
\end{split}
\end{equation}
where
\begin{equation}\label{20181208:eq2}
\begin{split}
\Gamma(z,w,v,\lambda,\mu)
&=
\big(
R(w,v)R(z,v+\mu)
-R(w,v)R(z,w+\mu)
\\
&
-R(z-\mu,v)R(w,z-\mu)
\big)\Omega_{12}\Omega_{23}
\\
&-\big(
R(z,v)R(w,v+\lambda)
-R(z,v)R(w,z+\lambda)
\\
&
-R(w-\lambda,v)R(z,w-\lambda)
\big)\Omega_{23}\Omega_{12}
\,.
\end{split}
\end{equation}
Using the identities \eqref{20181208:eq1} and equation \eqref{20181208:eq2}
we can rewrite equation \eqref{20181207:jacobi1} as follows
\begin{align}
\begin{split}\label{20181207:jacobi1b}
&\Gamma(v+\lambda +\zeta+\mu,w-\mu,z,-\lambda-\zeta-\mu,\mu)
\big|_{\zeta=\partial}L_1^*(\lambda-z)\\
&+\Gamma(z-\lambda,v+\lambda+\mu+\xi,w,\lambda,-\lambda-\mu-\xi)
\big|_{\xi=\partial}L_2^*(\mu-w)
\\
&+
\Gamma(w-\mu,z-\lambda,v,\mu,\lambda)
L_3(v+\lambda+\mu)\\
&=
\Gamma(v+\lambda+\zeta+\mu,w-\mu,z+\zeta,-\lambda-\zeta-\mu,\mu)
\big|_{\zeta=\partial}L_1(z)
\\
&+
\Gamma(z-\lambda,v+\lambda+\mu+\xi,w+\xi,\lambda,-\lambda-\mu-\xi)
\big|_{\xi=\partial}L_2(w)
\\
&
+\Gamma(w-\mu,z-\lambda,v+\nu,\mu,\lambda)
\big|_{\nu=\partial}L_3(v)
\,.
\end{split}
\end{align}
For any $X,Y,Z\in A$ we have
\begin{equation}\label{20181210:eq1}
X_1Y_2Z_3\Omega_{23}\Omega_{12}=\Omega_{23}\Omega_{12}X_2Y_3Z_1
\,,\quad
X_1Y_2Z_3\Omega_{12}\Omega_{23}=\Omega_{12}\Omega_{23}X_3Y_1Z_2
\,.
\end{equation}
Hence, using equations \eqref{mybe2bisse}, \eqref{20181210:eq1} and \eqref{eq:delta1}, we have
\begin{equation}\label{20181210:eq2}
\begin{split}
&\Gamma(v+\lambda +\zeta+\mu,w-\mu,z,-\lambda-\zeta-\mu,\mu)
\big|_{\zeta=\partial}L_1^*(\lambda-z)\\
&=\delta(z-w-\lambda-\partial)\delta(w-v-\mu)\Omega_{23}\Omega_{12}L_2(w)
\\
&
-\delta(z-v-\lambda-\partial)\delta(w-z-\mu)\Omega_{12}\Omega_{23}L_3(v)
\,,
\\
&\Gamma(z-\lambda,v+\lambda+\mu+\xi,w,\lambda,-\lambda-\mu-\xi)
\big|_{\xi=\partial}L_2^*(\mu-w)\\
&
=\delta(w-v-\mu-\partial)\delta(z-w-\lambda)\Omega_{23}\Omega_{12}L_3(v)
\\
&-\delta(z-v-\lambda)\delta(w-z-\mu-\partial)\Omega_{12}\Omega_{23}L_1(z)
\,,
\\
&
\Gamma(w-\mu,z-\lambda,v,\mu,\lambda)
L_3(v+\lambda+\mu)\\
&
=\delta(z-w-\lambda)\delta(w-v-\mu)\Omega_{23}\Omega_{12}L_1(z)
\\
&-\delta(w-z-\mu)\delta(z-v-\lambda)\Omega_{12}\Omega_{23}L_2(w)
\,.
\end{split}
\end{equation}
By equations \eqref{20181210:eq2} and \eqref{mybe2bisse} we immediately get that both sides of
\eqref{20181207:jacobi1b} coincide thus showing that the Jacobi identity \eqref{eq:pva3c}
holds for the bracket $\{\cdot\,_\lambda\,_\cdot\}_1^R$.
Similar (but longer) computation shows that the Jacobi identity \eqref{eq:pva3c} holds for the $\epsilon$-Adler
identity \eqref{20180414:eq4} and the $3$-rd Adler identity \eqref{eq:3adler}.
As in the proof of Theorem \ref{20180408:thm}, in order to show that any linear combination of the 3-rd, 2-nd and 1-st Adler identities satisfy Jacobi identity one has to check that the  $2$-nd Adler type identity \eqref{eq:2adler} satisfies \eqref{eq:pva3c}. This is again similar (but longer) to the analogous computation
for the $1$-st Adler type identity. The interest reader can check that, in this case, the Jacobi identity
\eqref{eq:pva3c} holds for $R^{(0)}$.
\end{proof}
\begin{remark}
If we apply Lemma 2.3(g)-(h) from \cite{DSKV17} we have
\begin{equation}
\begin{split}\label{20181210:eq3}
&\{L^{-1}(z)_\lambda L^{-1}(w)\}_3^R
=\big(\big|_{x_1=\partial}(L_1^{-1})^*(\lambda-z)\big)L_2^{-1}(w+\lambda+x_1+x_2+y_2+u)
\\
&\times\big(\big|_{u=\partial} \{L_1(z+x_2)_{\lambda+x_1+x_2} L_2(w+y_2)\}_3^R\big)\big(\big|_{x_2=\partial}L_1^{-1}(z)\big)
\big(\big|_{y_2=\partial}L_2^{-1}(w)\big)\,.
\end{split}
\end{equation}
By using the $3$-Adler identity \eqref{eq:3adler} we rewrite the RHS of \eqref{20181210:eq3}
as
\begin{equation}
\begin{split}\label{20181210:eq4}
&\frac12\Omega\Big(R^*_\xi(\delta(z-\xi))|_{\xi=w+\lambda+\partial}(L_2^{-1})^*(\lambda-z)
-\big|_{\zeta=\partial}L_1^{-1}(z)R^*_{\xi}(\delta(z+\zeta-\xi))\big|_{\xi=\zeta+w+\lambda}
\\
&
-R_w(\delta(z-w-\lambda))L_1(w+\lambda)+R_\xi(\delta(z-\lambda-\xi))\big|_{\xi=w+\partial}L_2(w)\Big)
\,.
\end{split}
\end{equation}
Here, we used the fact that $X_1Y_2\Omega=\Omega X_2Y_1$ and the identities
\begin{align*}
&L(z+x)\big|_{x=\partial}L^{-1}(z)=1\,,
&
&
L(w+\lambda+y)\big|_{y=\partial}L^{-1}(w+\lambda)=1
\,,
\\
&
L^*(\lambda+x-z)\big|_{x=\partial}(L^{-1})^*(\lambda-z)=1\,,
&
&
L(w+y)\big|_{y=\partial}L^{-1}(w)=1
\,.
\end{align*}
Note that equation \eqref{20181210:eq4} is the the RHS of \eqref{eq:1adler} with opposite sign.
This shows that if $L(\partial)$ satisfies the $3$-rd Adler type identity \eqref{eq:3adler} for an $R$-matrix $R$, then $L^{-1}(\partial)$ satisfies the $1$-st Adler type identity \eqref{eq:1adler} for the $R$-matrix $-R$.
Moreover, using equation (3.5) in \cite{DSK}, sesquilinearity and Leibniz rules, it is straightforward to check that, in any PVA, we have the identity
\begin{equation}\label{20181210:eq5}
\begin{split}
&\{a_\lambda\{b_\mu c\}-\{b_\mu\{a_\lambda c\}-\{\{a_\lambda b\}_{\lambda+\mu} c\}
=\Big(\{{a^{-1}}_{\lambda+x}\{{b^{-1}}_{\mu+y} c^{-1}\}
\\
&
-\{{b^{-1}}_{\mu+y}\{{a^{-1}}_{\lambda+x} c^{-1}\}-\{\{{a^{-1}}_{\lambda+x} {b^{-1}}\}_{\lambda+x+\mu+y} c^{-1}\}\Big)\big(\big|_{x=\partial}a^2\big)\big(\big|_{y=\partial}b^2\big)c^2\,.
\end{split}
\end{equation}
Jacobi identity \eqref{eq:pva3c} for the $\lambda$-bracket $\{L(z)_\lambda L(w)\}_3^R$ 
then follows by equation \eqref{20181210:eq5}
and the fact that $L^{-1}(\partial)$ satisfies the $1$-st Adler type identity \eqref{eq:1adler} for the $R$-matrix $-R$.
\end{remark}
\begin{theorem}\label{thm:main1}
Let $R=R^{(0)}$, $R^{(1)}$ or $R^{(2)}$ from \eqref{eq:affineR}.
Then the $\epsilon$-Adler identity \eqref{20180414:eq4} 
defines a continuous PVA $\lambda$-bracket on $\mc V_\infty\yhat$,
for every $\epsilon\in\mb F$.
In particular,
for all three $R$-matrices, the $3$rd and $1$st Adler identities \eqref{eq:3adler} and \eqref{eq:1adler}
define continuous PVA $\lambda$-brackets on $\mc V_\infty\yhat$,
while for $R=R^{(0)}$ 
all $3$rd, $2$nd and $1$st Adler identities \eqref{eq:3adler-0}--\eqref{eq:1adler-0}
define compatible continuous PVA $\lambda$-brackets on $\mc V_\infty\yhat$.
\end{theorem}
\begin{proof}
It is an immediate consequence of Propositions \ref{prop:cont}, \ref{prop:skew} and \ref{prop:jacobi}.
\end{proof}

\section{Hamiltonian equations and integrability}\label{sec:ham}
As in the usual PVA case, see \cite{BDSK09}, given a continuous $\lambda$-bracket on $\mc V_{\infty}\yhat$,
the space of local functionals
$\mc V_\infty\yhat/\partial\mc V_\infty\yhat$ acts on $\mc V_\infty\yhat$ by derivations, commuting with $\partial$, as follows
$$
\{\tint h,u\}=\{h_\lambda u\}|_{\lambda=0}\,,
\qquad h,u\in\mc V_\infty\yhat\,.
$$
A \emph{Hamiltonian equation} on $\mc V_\infty\yhat$ associated to a \emph{Hamiltonian functional} 
$\tint h\in\mc V_\infty\yhat/\partial\mc V_\infty\yhat$ is the evolution equation 
\begin{equation}\label{ham-eq}
\frac{du}{dt}=\{\tint h,u\}\,\,, \,\,\,\, u\in\mc V\,.
\end{equation}
An \emph{integral of motion} for the Hamiltonian equation \eqref{ham-eq}
is a local functional $\tint f\in\mc V_\infty\yhat/\partial\mc V_\infty\yhat$ such that $\{\tint h,\tint f\}=0$,
and two integrals of motion $\tint f,\tint g$ are \emph{in involution} if $\{\tint f,\tint g\}=0$. Here, the Lie
bracket in $\mc V_\infty\yhat/\partial\mc V_\infty\yhat$ is the one defined in Section \ref{sec:locPoisson}.
The minimal requirement for \emph{integrability} is to have an infinite collection
$\tint h_0=\tint h,\,\tint h_1,\,\tint h_2,\,\dots\,$
of linearly independent integrals of motion in involution.
In this case, we have the \emph{integrable hierarchy} of Hamiltonian equations
\begin{equation}\label{eq:int-hier}
\frac{du}{dt_n}=\{\tint h_n,u\}\,\,, \,\,\,\, n\in\mb Z_{\geq0}
\,.
\end{equation}

\begin{theorem}\label{thm:hn}
Let $R=R^{(0)}$, $R^{(1)}$ or $R^{(2)}$ from \eqref{eq:affineR}.
For $n\in\mb Z_{\geq0}$, 
define the elements $h_{n}\in\mc V_\infty\yhat$ by ($\Tr=1\otimes\Tr$)
\begin{equation}\label{eq:hn}
h_{n}=
\frac{-1}{n}
\Res_z\Tr(L^n(z))
\text{ for } n\neq0
\,,\,\,
h_0=0\,.
\end{equation}
Then: 
\begin{enumerate}[(a)]
\item
All the elements $\tint h_{n}$ are Hamiltonian functionals in involution, for every $\epsilon\in\mb F$:
\begin{equation}\label{eq:invol}
\{\tint h_{m},\tint h_{n}\}^{R,\epsilon}=0
\,\text{ for all } m,n\in\mb Z_{\geq0}
\,.
\end{equation}
\item
The corresponding compatible hierarchy of Hamiltonian equations satisfies
\begin{equation}\label{eq:hierarchy}
\frac{dL(w)}{dt_{n}}
=
\{\tint h_{n},L(w)\}^{R,\epsilon}
=\frac12
[R((L+\epsilon\id)\circ L^{n-1}\circ (L+\epsilon\id)),L](w)
\,,\,\,n\in\mb Z_{\geq0}
\,,
\end{equation}
(in the RHS we are taking the symbol of the commutator of pseudodifferential operators),
and the Hamiltonian functionals $\tint h_{n}$, $n\in\mb Z_{\geq0}$,
are integrals of motion of all these equations.
\end{enumerate}
\end{theorem}
It follows immediately from part (b) and equation \eqref{20180408:eq2b}
that we have the following triple Lenard-Magri
relations ($n\in\mb Z_{\geq1}$)
\begin{equation}\label{20220621:eq2}
\begin{split}
\{\tint h_{n-1},L(w)\}^{R,3}&=\{\tint h_{n},L(w)\}^{R,2}
=\{\tint h_{n+1},L(w)\}^{R,1}
\\
&=\frac12[R (L^n),L](w)\,.
\end{split}
\end{equation}
Equation \eqref{20220621:eq2} is the affine analogue of equation \eqref{eq:ham6}.

In the remainder of the section we will give a proof of Theorem \ref{thm:hn}.
We will use the following results for which we omit the proofs, see \cite[Lemmas 6.3, 6.4 and 6.5]{DSKV17}.
\begin{lemma}\label{lemma:29032017}
Let $X$, $Y$ be in $A$. Then
\begin{enumerate}[(a)]
\item $(\Tr\otimes \id)(\Omega (X\otimes Y))=XY\in A$;
\item $(\Tr\otimes \Tr)(\Omega(X\otimes Y))=\Tr(X Y)\in\mb F$.
\end{enumerate}
\end{lemma}
\begin{lemma}\label{lem:hn1}
Given two operators $P(\partial),Q(\partial)\in\mc V_\infty((\partial^{-1}))\xhat\otimes A$, we have
\begin{enumerate}[(a)]
\item
$\Res_z P(z)Q^*(\lambda-z)=\Res_zP(z+\lambda+\partial)Q(z)$;
\item
$\tint \Res_z \Tr(P(z+\partial)Q(z))=\tint \Res_z\Tr(Q(z+\partial)P(z))$.
\end{enumerate}
\end{lemma}
\begin{lemma}\label{lem:hn2}
For $n\in\mb Z_{\geq0}$, let $h_{n}\in\mc V_\infty\yhat$ be given by \eqref{eq:hn}.
Then, for $a\in\mc V_\infty\yhat$, we have
\begin{equation}\label{eq:hn2}
\begin{split}
& \{{h_{n}}_\lambda a\}^{R,\epsilon}\big|_{\lambda=0}
=
-\Res_z\Tr 
\{L(z+x)_x a\}^{R,\epsilon}\big(\big|_{x=\partial}L^{n-1}(z)\big)
\,,
\\
& \tint \{a_\lambda h_{n}\}^{R,\epsilon}\big|_{\lambda=0}
=
-\int \Res_w\Tr 
\{a_\lambda L(w+x)\}^{R,\epsilon}\big|_{\lambda=0}\big(\big|_{x=\partial}L^{n-1}(w)\big)
\,.
\end{split}
\end{equation}
\end{lemma}
\begin{proof}[Proof of Theorem \ref{thm:hn}]
Applying the second equation in \eqref{eq:hn2} first,
and then the first equation in \eqref{eq:hn2}, we get
\begin{equation}\label{eq:hn-pr1}
\begin{split}
& \{\tint h_{m},\tint h_{n}\}
= 
\int \Res_z \Res_w (\Tr\otimes\Tr)
\{L(z+x)_x L(w+y)\}^{R,\epsilon}
\\
& \,\,\,\,\,\,\,\,\,\,\,\,\,\,\,\,\,\, \times
\Big(
\big(\big|_{x=\partial}L_1^{m-1}(z)\big)
\big(\big|_{y=\partial}L_2^{n-1}(w)\big)
\Big)\,.
\end{split}
\end{equation}
We can now use the $\epsilon$-Adler identity \eqref{20180414:eq4}
associated to $R$ 
to rewrite the RHS of \eqref{eq:hn-pr1} as
\begin{equation}\label{eq:hn-pr2}
\begin{array}{l}
\displaystyle{
\vphantom{\Big(}
\frac12
\int \Res_z \Res_w (\Tr\otimes\Tr)\Omega
(L_1(w+\partial)+\epsilon\id)
\big(\big|_{\zeta=z+\partial}L_1^m(z)\big)
} \\
\displaystyle{
\vphantom{\Big(}
\,\,\,\times R_\zeta(\delta(\zeta-\xi))\big|_{\xi=\zeta-z+w+\partial}
(L_2(w+\partial)+\epsilon\id)L_2^{n-1}(w)
} \\
\displaystyle{
\vphantom{\Big(}
-\frac12
\int \Res_z \Res_w (\Tr\otimes\Tr)\Omega(L_1(w+\partial)+\epsilon\id)
\big(\big|_{x=z+\partial}L_1^{m-1}(z)\big)
} \\
\displaystyle{
\vphantom{\Big(}
\,\,\,\times
R_x(\delta(x-\xi))\big|_{\xi=x-z+w+\partial}L_2^{*}(-z)(L_2(w+\partial)+\epsilon\id)
L_2^{n-1}(w)
} 
\end{array}
\end{equation}
\begin{equation}\label{eq:hn-pr2b}
\begin{array}{l}
\displaystyle{
\vphantom{\Big(}
+\frac12
\int \Res_z \Res_w (\Tr\otimes\Tr)\Omega L_1(w+\partial)\Big((L_1(z+\partial)+\epsilon\id)L_1^{m-1}(z)\Big)} \\
\displaystyle{
\vphantom{\Big(}
\,\,\,
\times
R_\xi(\delta(z-\zeta-\xi))
\big(\big|_{\zeta=\partial}L_2^*(-z)+\epsilon\id\big)
\big(\big|_{\xi=w+\partial}L_2^{n-1}(w)\big)
} \\
\displaystyle{
\vphantom{\Big(}
-\frac12
\int \Res_z \Res_w (\Tr\otimes\Tr)\Omega\Big((L_1(z+\partial)+\epsilon\id)L_1^{m-1}(z)\Big)
} \\
\displaystyle{
\vphantom{\Big(}
\,\,\,\times
R_\xi(\delta(z-\zeta-\xi))
\big(\big|_{\zeta=\partial}L_2^*(-z)+\epsilon\id\big)
\big(\big|_{\xi=w+\partial}L_2^{n}(w)\big)
\,.}
\end{array}
\end{equation}
We can use
Lemma \ref{lemma:29032017}(b),
to rewrite the term \eqref{eq:hn-pr2} as
\begin{equation}\label{eq:hn-pr3}
\begin{array}{l}
\displaystyle{
\vphantom{\Big(}
\frac12
\int \Res_z \Res_w \Tr\,
(L(w+\partial)+\epsilon\id)}
\Big(\big(\big|_{\zeta=z+\partial}L^{m}(z)\big)R_\zeta(\delta(\zeta-\xi))|_{\xi=\zeta-z+w+\eta}
\\
\displaystyle{
\vphantom{\Big(}
\,\,\,
-\big(\big|_{x=z+\partial}L^{m-1}(z)\big)
R_x(\delta(x-\xi))\big(\big|_{\zeta=\partial}L^*(-z)\big)\big|_{\xi=x-z+w+\zeta+\eta}
\Big)
}
\\
\displaystyle{
\vphantom{\Big(}
\,\,\,\times
\big(\big|_{\eta=\partial}(L(w+\partial)+\epsilon\id)L^{n-1}(w)\big)
\,.}
\end{array}
\end{equation}
By Lemma \ref{lem:hn1}(a) and the fact that $L^m(\partial)=L^{m-1}(\partial)\circ L(\partial)$ we have that
\begin{equation}\label{eq:hn-pr3b}
\begin{split}
&\res_z\Big( \big(\big|_{x=z+\partial}L^{m-1}(z)\big)R_x(\delta(x-\xi))\big(\big|_{\zeta=\partial}L^*(-z)\big)\big|_{\xi=x-z+w+\zeta+\eta}\Big)
\\
&=
\res_z\Big( \big(\big|_{\zeta=z+\partial}L^{m}(z)\big)R_\zeta(\delta(\zeta-\xi))\big|_{\xi=\zeta-z+w+\eta}\Big)
\,.
\end{split}
\end{equation}
Hence, from equations \eqref{eq:hn-pr3} and \eqref{eq:hn-pr3b} it follows that the term
\eqref{eq:hn-pr2} vanishes.

Next, we can use Lemma \ref{lemma:29032017}(b) and  Lemma \ref{lem:hn1}(b)
to rewrite the first term in \eqref{eq:hn-pr2b}
as
\begin{equation}\label{eq:hn-pr2bb}
\begin{array}{l}
\displaystyle{
\vphantom{\Big(}
+\frac12
\int \Res_z \Res_w \Tr\,L(w+\partial)\Big((L(z+\partial)+\epsilon\id)L^{m-1}(z)\Big)} \\
\displaystyle{
\vphantom{\Big(}
\,\,\,
\times
R_\xi(\delta(z-\zeta-\xi))
\big(\big|_{\zeta=\partial}L^*(-z)+\epsilon\id\big)
\big(\big|_{\xi=w+\partial}L^{n-1}(w)\big)
} \\
\displaystyle{
\vphantom{\Big(}
=\frac12
\int \Res_z \Res_w \Tr\,\Big((L(z+\partial)+\epsilon\id)L^{m-1}(z)\Big)
} \\
\displaystyle{
\vphantom{\Big(}
\,\,\,\times
R_\xi(\delta(z-\zeta-\xi))
\big(\big|_{\zeta=\partial}L^*(-z)+\epsilon\id\big)
\big(\big|_{\xi=w+\partial}L^{n}(w)\big)
\,.}
\end{array}
\end{equation}
On the other hand, by Lemma \ref{lemma:29032017}(b) the second term in \eqref{eq:hn-pr2b} is equal to
\begin{equation}\label{eq:hn-pr2bbb}
\begin{array}{l}
\displaystyle{
\vphantom{\Big(}
-\frac12
\int \Res_z \Res_w (\Tr\otimes\Tr)\Omega\Big((L_1(z+\partial)+\epsilon\id)L_1^{m-1}(z)\Big)
} \\
\displaystyle{
\vphantom{\Big(}
\,\,\,\times
R_\xi(\delta(z-\zeta-\xi))
\big(\big|_{\zeta=\partial}L_2^*(-z)+\epsilon\id\big)
\big(\big|_{\xi=w+\partial}L_2^{n}(w)\big)
}
\\
\displaystyle{
\vphantom{\Big(}
=-\frac12
\int \Res_z \Res_w \Tr\,\Big((L(z+\partial)+\epsilon\id)L^{m-1}(z)\Big)
} \\
\displaystyle{
\vphantom{\Big(}
\,\,\,\times
R_\xi(\delta(z-\zeta-\xi))
\big(\big|_{\zeta=\partial}L^*(-z)+\epsilon\id\big)
\big(\big|_{\xi=w+\partial}L^{n}(w)\big)
\,.}
\end{array}
\end{equation}
From equations \eqref{eq:hn-pr2bb} and \eqref{eq:hn-pr2bbb} we have that the term \eqref{eq:hn-pr2b} vanishes thus proving (a).

We are left to prove part (b).
We have
\begin{equation}\label{eq:proofb-1}
\begin{array}{l}
\displaystyle{
\vphantom{\Big(}
\{\tint h_{n},L(w)\}^{R,\epsilon}
=
\{{h_{n}}_\lambda L(w)\}^{R,\epsilon}\big|_{\lambda=0}
} \\
\displaystyle{
\vphantom{\Big(}
=
-\Res_z (\Tr\otimes\id)
\{L(z+x)_x L(w)\}
\big(\big|_{x=\partial}L_1^{n-1}(z)\big)
} \\
\displaystyle{
\vphantom{\Big(}
=
-\frac12
\Res_z(\Tr\otimes\id)
\Omega
(L_1(w+\partial)+\epsilon\id)
(\big|_{\zeta=z+\partial}L_1^m(z))
} \\
\displaystyle{
\vphantom{\Big(}
\qquad\qquad\times R_\zeta}(\delta(\zeta-\xi))(\big|_{\xi=\zeta-z+w+\partial}L_2(w)+\epsilon\id) \\
\displaystyle{
\vphantom{\Big(}
\,\,\,\,
+\frac12
\Res_z (\Tr\otimes\id)
\Omega(L_1(w+\partial)+\epsilon\id)
(\big|_{x=z+\partial}L_1^{m-1}(z))
} \\
\displaystyle{
\vphantom{\Big(}
\qquad\qquad\times
R_z(\delta(x-\xi))
(\big|_{\zeta=\partial}L_2^*(-z))(\big|_{\xi=x-z+\zeta+w+\partial}L_2(w)+\epsilon\id)
} \\
\displaystyle{
\vphantom{\Big(}
\,\,\,\,
-\frac12
\Res_z (\Tr\otimes\id)\Omega L_1(w+\partial)
\Big((L_1(z+\partial)+\epsilon\id)L_1^{n-1}(z)\Big)
} \\
\displaystyle{
\vphantom{\Big(}
\qquad\qquad\times
R_w(\delta(z-w-\zeta))
(\big|_{\zeta=\partial}L_2^*(-z)+\epsilon\id)} \\
\displaystyle{
\vphantom{\Big(}
\,\,\,\,
+\frac12
\Res_z (\Tr\otimes\id)\Omega
\Big((L_1(z+\partial)+\epsilon\id)L_1^{n-1}(z)\Big)
} \\
\displaystyle{
\vphantom{\Big(}
\qquad\qquad\times
R_\xi(\delta(z-\zeta-\xi))
(\big|_{\zeta=\partial}L_2^*(-z)+\epsilon\id)
(\big|_{\xi=w+\partial}L_2(w))
} \\
\displaystyle{
\vphantom{\Big(}
=
-\frac12
(L(w+\partial)+\epsilon\id)
\Res_z\Big((\big|_{\zeta=z+\partial}L^m(z))
R_\zeta(\delta(\zeta-\xi))\Big)
}\\
\displaystyle{
\vphantom{\Big(}
\qquad\qquad\times(\big|_{\xi=\zeta-z+w+\partial}L(w)+\epsilon\id)
}\\
\displaystyle{
\vphantom{\Big(}
\,\,\,\,
+\frac12
(L(w+\partial)+\epsilon\id)
\Res_z \Big((\big|_{x=z+\partial}L^{m-1}(z))
R_z(\delta(x-\xi))
(\big|_{\zeta=\partial}L^*(-z))\Big)
} \\
\displaystyle{
\vphantom{\Big(}
\qquad\qquad\times(\big|_{\xi=x-z+\zeta+w+\partial}L(w)+\epsilon\id)
} \\
\displaystyle{
\vphantom{\Big(}
\,\,\,\,
-\frac12
L(w+\partial)
\Res_z \Big[\Big((L(z+\partial)+\epsilon\id)L^{n-1}(z)\Big)
} \\
\displaystyle{
\vphantom{\Big(}
\qquad\qquad\times
R_w(\delta(z-w-\zeta))
(\big|_{\zeta=\partial}L^*(-z)+\epsilon\id)} \Big]\\
\displaystyle{
\vphantom{\Big(}
\,\,\,\,
+\frac12
\Res_z 
\Big[\Big((L(z+\partial)+\epsilon\id)L^{n-1}(z)\Big)
} \\
\displaystyle{
\vphantom{\Big(}
\qquad\qquad\times
R_\xi(\delta(z-\zeta-\xi))
(\big|_{\zeta=\partial}L^*(-z)+\epsilon\id)\Big]
(\big|_{\xi=w+\partial}L(w))
} \\
\displaystyle{
\vphantom{\Big(}
=
-\frac12
L(w+\partial)
\Res_z (L(z+\partial)+\epsilon\id)L^{n-1}(z+\partial)(L(z)+\epsilon\id) R_w(\delta(z-w))
}\\
\displaystyle{
\vphantom{\Big(}
\,\,\,\,
+\frac12
\Res_z 
(L(z+\partial)+\epsilon\id)L^{n-1}(z+\partial)(L(z)+\epsilon\id)
R_\xi(\delta(z-\xi))
(\big|_{\xi=w+\partial}L(w))
} \\
\displaystyle{
\vphantom{\Big(}
=
-\frac12
L(w+\partial)
(L(w+\partial)+\epsilon\id)L^{n-1}(w+\partial)(L^*(-z)+\epsilon\id)
}\\
\displaystyle{
\vphantom{\Big(}
\,\,\,\,
+\frac12
(L(w+\partial)+\epsilon\id)L^{n-1}(w+\partial)(L(w+\partial)+\epsilon\id)
L(w)
\,.}
\end{array}
\end{equation}
In the second equality we used 
the first equation in \eqref{eq:hn2},
in the third equality we used the $\epsilon$-Adler identity \eqref{20180414:eq4}
associated to the $R$-matrix $R$,
in the fourth equality we used Lemma \ref{lemma:29032017}(a),
in the fifth equality we used equation \eqref{eq:hn-pr3b} and Lemma \ref{lem:hn1}(a),
in the last equality we used the identity
$$
\Res_z (P(z)R_w(\delta(z-w)))=R_w(P(w))\,,
$$
which can be easily verified for any $P(z)\in\mc V_\infty((z^{-1}))\xhat$.
This proves \eqref{eq:hierarchy} and completes the proof of the Theorem.
\end{proof}

For $k\in\mb Z$, recall the projection maps $\Pi_{\geq k}$ defined in Section \ref{sec:R}. For 
$P(\partial)\in\mc V_\infty((\partial^{-1}))\xhat\otimes A$ we clearly have
$$
P(\partial)=\Pi_{\geq k}(P(\partial))+\Pi_{<k}(P(\partial))
\,.
$$
Moreover, we clearly have $[L^n(\partial),L(\partial)]=0\,,$
for every $n\in\mb Z_{\geq0}$. Hence, for the $R$-matrices $R^{(k)}$, $k=0,1,2$, defined in \eqref{eq:affineR},
the hierarchy \eqref{20220621:eq2} can be rewritten as
$$
\frac{d L(\partial)}{dt_{n}}=[\Pi_{\geq k}(L^n(\partial)),L(\partial)]
\,,\,\,n\in\mb Z_{\geq0}\,.
$$
(The case $k=2$ occurs only if $A=\mb F$.)

\section{Example: the KP hierarchy}\label{sec:exa}

In this section we employ the machinery developed in Section \ref{sec:4} to provide a description of the
tri-Hamiltonian structure of the $A$-valued Kadomtsev-Petviashvili (KP) hierarchy,
where, as before, $A$ is a finite-dimensional unital associative algebra over $\mb F$
with a non-degenerate trace form.

Let $R=R^{(0)}$, as defined in \eqref{eq:affineR}(i). Throughout the section let us use the shorthand
$\{\cdot\,_\lambda\,\cdot\}^\epsilon:=\{\cdot\,_\lambda\,\cdot\}^{R^{(0)},\epsilon}$
to denote the $\lambda$-bracket on $\mc V_\infty\yhat$ defined by equation \eqref{eq:3adler-eps}.
We also simply denote
by $\{\cdot\,_\lambda\,\cdot\}_i:=\{\cdot\,_\lambda\,\cdot\}^{(0)}_i$, $i=1,2,3$,
the compatible $\lambda$-brackets on $\mc V_\infty\yhat$ defined by equations \eqref{eq:1adler-0}, \eqref{eq:2adler-0}
and \eqref{eq:3adler-0} respectively.

For $N\in\mb Z$, recall the differential algebra homomorphism $\pi_N:\mc V_\infty\yhat\to\mc V_N$,
defined in \eqref{20180417:eq4}, sending $u_{p,\alpha}^{(n)}$ to zero if $p<-N-1$.
By construction we have $\mc V_N\cong \mc V_\infty\yhat/\Ker\pi_N$. Note that,
in general $\Ker\pi_N$ is not a PVA ideal. Hence, generally, we do not have an induced PVA structure on the quotient space $\mc V_N$. 

For $p\in\mb Z$ we denote
\begin{equation}\label{Up}
U_{p}=\sum_{\alpha\in I}u_{p,\alpha}E^\alpha\in\mc V_{\infty}\yhat\otimes A\,,
\end{equation}
so that, 
from \eqref{20180417:eq9}, we have
\begin{equation}\label{20220719:eq1}
L(z)=
\sum_{p\in\mb Z}U_pz^{-p-1}
\,.
\end{equation}
Moreover, $\pi_N(U_p)=0$ if
$p<-N-1$.
In fact, $\Ker\pi_N$ is the differential algebra ideal of $\mc V_\infty\yhat$ generated by the coefficients of
$U_p$, $p<-N-1$.

We have, from \eqref{eq:3adler-eps} and \eqref{20220719:eq1}
\begin{equation}\label{eq:3adler-Up}
\begin{array}{l}
\displaystyle{
\vphantom{\Big(}
\{(U_p)_1{}_\lambda L_2(w)\}^{\epsilon}
=\res_z z^p\{L_1(z)_\lambda L_2(w)\}^{\epsilon}
} \\
\displaystyle{
\vphantom{\Big(}
=
\Omega
\bigg(
(L_1(w+\lambda+\partial)+\epsilon\id)
\big((w+\lambda+\partial)^{p}
L_2(w+\partial)\big)_+(L_2(w)+\epsilon\id)
\big)
} \\
\displaystyle{
\vphantom{\Big(}
-(L_1(w+\lambda+\partial)+\epsilon\id)
\big(L_1(w+\lambda+\partial)
(w+\lambda+\partial)^{p}\big)_+
(L_2(w)+\epsilon\id)
\big)
} \\
\displaystyle{
\vphantom{\Big(}
+
L_1(w+\lambda+\partial)
\big((L_1(w+\lambda+\partial)+\epsilon\id)
(w+\lambda+\partial)^{p}
(L_2(w)+\epsilon\id)\big)_+
}
\\
\displaystyle{
\vphantom{\Big(}
-
\big((L_1(w+\lambda+\partial)+\epsilon\id)
(w+\lambda+\partial)^{p}
(L_2(w+\partial)+\epsilon\id)\big)_+
L_2(w)
\bigg)
\,.}
\end{array}
\end{equation}
In the second equality of \eqref{eq:3adler-Up} we used the identity 
\begin{equation}\label{20220623:eq1}
\res_z a(z)i_z(z-w)^{-1}=a(w)_+
\,.
\end{equation}

\begin{proposition}\phantomsection\label{20220719:prop1}
\begin{enumerate}[(a)]
\item
If $N\leq1$, then $\Ker\pi_N$ is a PVA ideal for the continuous PVA structure on $\mc V_\infty\yhat$
defined by the 3-Adler identity \eqref{eq:3adler-0}.
\item
For every $N\in\mb Z$, $\Ker\pi_N$ is a PVA ideal for the continuous PVA structure on
$\mc V_\infty\yhat$ defined by the 2-Adler idenitty \eqref{eq:2adler-0}.
\item
If $N\geq-1$, then $\Ker\pi_N$ is a PVA ideal for the continuous PVA structure on $\mc V_\infty\yhat$
defined by the 1-Adler identity \eqref{eq:1adler-0}.
\item
The $\epsilon$-Adler identity \eqref{eq:3adler-eps} defines a PVA structure on the differential algebras 
$\mc V_{N}$ for $N=-1,0,1$.
\end{enumerate}
\end{proposition}
\begin{proof}
Recall that the coefficients of $U_p$, $p<-N-1$, generate $\Ker\pi_N$ as a differential algebra ideal.
Hence, if we show that $\pi_N\{{(U_p)_1}_\lambda L_2(w)\}_i^{(0)}=0$,
$i=1,2,3$, for every $p<-N-1$, then $\Ker\pi_N$ is a PVA ideal of $\mc V_\infty\yhat$.
Note that, if $p<-N-1$, then
$$
\pi_N(L(w+\lambda+\partial))(w+\lambda+\partial)^p
$$
has order $p+N<-N-1+N=-1$. Hence,
\begin{equation}\label{20220719:eq2}
\big(\pi_N(L(w+\lambda+\partial))(w+\lambda+\partial)^p\big)_+=0\,.
\end{equation}
Similarly
\begin{equation}\label{20220719:eq3}
\big((w+\lambda+\partial)^p\pi_N(L(w+\partial))\big)_+=0\,.
\end{equation}
For $p<-N-1$, we thus have, from equations \eqref{eq:3adler-Up}, \eqref{20220719:eq2} and
\eqref{20220719:eq3},
\begin{equation}\label{eq:3adler-Upker}
\begin{array}{l}
\displaystyle{
\vphantom{\Big(}
\pi_N\{(U_p)_1{}_\lambda L_2(w)\}^{\epsilon}
} \\
\displaystyle{
\vphantom{\Big(}
=
\Omega
\bigg(
\pi_N(L_1(w+\lambda+\partial))
\big(\pi_N(L_1(w+\lambda+\partial))
(w+\lambda+\partial)^{p}
\pi_N(L_2(w))\big)_+
} \\
\displaystyle{
\vphantom{\Big(}
-
\big(\pi_N(L_1(w+\lambda+\partial))
(w+\lambda+\partial)^{p}
\pi_N(L_2(w+\partial))\big)_+
\pi_N(L_2(w))\bigg)
}
\\
\displaystyle{
\vphantom{\Big(}
+\epsilon^2 \Omega\bigg(\pi_N(L_1(w+\lambda))
\big(
(w+\lambda)^{p}
\big)_+
-\big(
(w+\lambda+\partial)^{p}\big)_+
\pi_N(L_2(w))
\bigg)
\,.}
\end{array}
\end{equation}
Recall the expansion \eqref{20180408:eq2b}.
Since there is no coefficient of $2\epsilon$ in \eqref{eq:3adler-Upker} we have that
$\pi_N\{(U_p)_1{}_\lambda L_2(w)\}_2=0$, for every $N\in\mb Z$, proving part (b).
Moreover, $\pi_N(L(w+\lambda+\partial))(w+\lambda+\partial)^{p}\pi_N(L(w))$ has order $2N+p$.
If $p<-N-1$, then $2N+p<N-1$. Hence, if $N\leq1$,
from equation \eqref{eq:3adler-Upker} we have
$\pi_N\{(U_p)_1{}_\lambda L_2(w)\}_3=0$ which proves part (a).
Finally, if $p<0$, which happens when $N>1$, from equation \eqref{eq:3adler-Upker} we have
$\pi_N\{(U_p)_1{}_\lambda L_2(w)\}_1=0$ proving part (c).
Part (d) follows from parts (a),(b) and (c) and the fact that $\mc V_{N}\cong\mc V_\infty\yhat/\Ker\pi_{N}$.
\end{proof}

By an abuse of notation we simply denote
\begin{equation}\label{L:quasiKP}
L(\partial):=\pi_1(L(\partial))=\sum_{p\leq 1}U_{-p-1}\partial^{p}\in\mc V_1((\partial^{-1}))\otimes A\,,
\end{equation}
where $U_p$ is as in \eqref{Up}. From Proposition \ref{20220719:prop1}(d)
we have that the $\epsilon$-Adler identity \eqref{eq:3adler-eps}, for $L(\partial)$ as in
\eqref{L:quasiKP}, defines a PVA structure on $\mc V_1$.
\begin{lemma}\label{20220719:lem1}
In the PVA $\mc V_1$, we have the following $\lambda$-brackets relations:
\begin{enumerate}[(a)]
\item
$\{{(U_{-2})_1}_\lambda L_2(w)\}^\epsilon=
\Omega\Big(
L_1(w+\lambda+\partial)(U_{-2}\otimes U_{-2}) - (U_{-2}\otimes U_{-2})L_2(w)
\Big)$.
\item
$\{L_1(z)_\lambda (U_{-2})_2\}^\epsilon=
\Omega\Big(
 (U_{-2}\otimes U_{-2})L_1(z) -L_2^*(\lambda-z)(U_{-2}\otimes U_{-2}) 
\Big)$.
\item
\begin{align*}
&\{{(U_{-1})_1}_\lambda L_2(w)\}^\epsilon=
\Omega\Big(
(L_1(w+\lambda+\partial)+\epsilon\id)(\id\otimes U_{-2}-U_{-2}\otimes\id)(L_2(w)+\epsilon\id)
\\
&+L_1(w+\lambda+\partial)
\big((U_{-2}\otimes U_{-2})w+(U_{-1}+\epsilon\id)\otimes U_{-2}+U_{-2}\otimes(U_{-1}+\epsilon\id)\big)
\\
&-\big((U_{-2}\otimes U_{-2})(w+\partial)+(U_{-1}+\epsilon\id)\otimes U_{-2}+U_{-2}\otimes(U_{-1}+\epsilon\id)\big)L_2(w)
\Big)\,.
\end{align*}
\item
\begin{align*}
&\{L_1(z)_\lambda(U_{-1})_2\}^\epsilon=
\Omega\Big(
(L_2^*(\lambda-z)+\epsilon\id)(\id\otimes U_{-2}-U_{-2}\otimes\id)(L_1(z)+\epsilon\id)
\\
&+\big((U_{-2}\otimes U_{-2})(z+\partial)+(U_{-1}+\epsilon\id)\otimes U_{-2}+U_{-2}\otimes(U_{-1}+\epsilon\id)\big)L_1(z)
\\
&-L_2^*(\lambda-z)
\big((U_{-2}\otimes U_{-2})z+(U_{-1}+\epsilon\id)\otimes U_{-2}+U_{-2}\otimes(U_{-1}+\epsilon\id)\big)
\Big)\,.
\end{align*}
\end{enumerate}
Furthermore, we have
\begin{align}
\begin{split}\label{C11}
&\{{(U_{-2})_1}_\lambda (U_{-2})_2\}^\epsilon=
\Omega\Big(U_{-2}^2\otimes U_{-2}-U_{-2}\otimes U_{-2}^2\Big)
\,,
\end{split}
\\
\begin{split}\label{C12}
&\{{(U_{-2})_1}_\lambda (U_{-1})_2\}^\epsilon
=
\Omega\Big(U_{-2}(\lambda+\partial)(U_{-2}\otimes U_{-2})
+U_{-1}U_{-2}\otimes U_{-2}
\\
&\qquad\qquad\qquad\qquad\!-U_{-2}\otimes U_{-2}U_{-1}\Big)
\,,
\end{split}
\\
\begin{split}\label{C21}
&\{{(U_{-1})_1}_\lambda (U_{-2})_2\}^\epsilon
=
\Omega\Big(U_{-2}\otimes \big((\lambda+\partial)U_{-2}\big) U_{-2})
+U_{-2}U_{-1}\otimes U_{-2}
\\
&\qquad\qquad\qquad\qquad\!-U_{-2}\otimes U_{-1}U_{-2}\Big)
\,,
\end{split}
\\
\begin{split}\label{C22}
&\{{(U_{-1})_1}_\lambda (U_{-1})_2\}^\epsilon
=
\Omega\Big(U_{-2}(\lambda+\partial)(U_{-1}\otimes U_{-2})
-U_{-2}\otimes\Big((\lambda+\partial)U_{-2}\Big)U_{-1}
\\
&\qquad\qquad\qquad\qquad\!+U_{-1}^2\otimes U_{-2}-U_{-2}\otimes U_{-1}^2
+U_{0}\otimes U_{-2}^2-U_{-2}^2\otimes U_{0}\Big) 
\\
&\qquad\qquad\qquad\qquad\!+2\epsilon\,\Omega\Big(
U_{-2}\otimes (\lambda+\partial)U_{-2}+U_{-1}\otimes U_{-2}-U_{-2}\otimes U_{-1}\Big)
\\
&\qquad\qquad\qquad\qquad\!+\epsilon^2\,\Omega\Big(\id\otimes U_{-2}-U_{-2}\otimes\id\big)
\,.
\end{split}
\end{align}
\end{lemma}
\begin{proof}
It follows by a straightforward $\lambda$-bracket computation from \eqref{eq:3adler-Up}.
\end{proof}

Consider the differential algebra homomorphism
$$
\phi:\mc V_1\to\mc V_{1}
$$
defined on generators by ($p\geq-2,\alpha\in I, n\in\mb Z_{\geq0}$)
\begin{equation}\label{20220720:eq1}
\phi(u_{p,\alpha}^{(n)})=
\left\{
\begin{array}{ll}
\delta_{n0}\Tr(E_\alpha)\,,
& p=-2\,,
\\
0\,, & p=-1\,,
\\
u_{p,\alpha}^{(n)}\,,& p\geq0\,,
\end{array}
\right.
\end{equation}
and extended using the Leibniz rule. Extending $\phi$ to a homomorphism
$\phi:\mc V_1\otimes A\to\mc V_{-1}\otimes A$ acting as the identity on $A$, and using
\eqref{Up},
we can rewrite equations \eqref{20220720:eq1} in a more compact form as follows $(n\in\mb Z_{\geq0})$:
$$
\phi(U_{-2}^{(n)})=\delta_{n0}\id\,,\qquad
\phi(U_{-1}^{(n)})=0\,,\qquad
\phi(U_p^{(n)})=U_p^{(n)}\,,\,\, p\geq0\,.
$$
We then get, from Lemma \ref{20220719:lem1}(a)-(c) the following identities
\begin{equation}\label{20220720:eq2}
\begin{split}
&\phi\left(\{{(U_{-2})_1}_\lambda L_2(w)\}^\epsilon\right)=
\Omega\Big(
\phi(L_1(w+\lambda))-\phi(L_2(w))
\Big)
\\
&\phi\left(\{{(U_{-1})_1}_\lambda L_2(w)\}^\epsilon\right)=
\Omega\Big(
\phi(L_1(w+\lambda))(w+2\epsilon)
-(w+\partial+2\epsilon)\phi(L_2(w))
\Big)\,.
\end{split}
\end{equation}
Note that in both equations in \eqref{20220720:eq2} there is no $\epsilon^2$ term. Hence, we have
$$
\phi\left(\{{(U_{-2})_1}_\lambda L_2(w)\}_3\right)=0
=\{{\phi(U_{-2})_1}_\lambda \phi(L_2(w))\}_3
$$
and
$$
\phi\left(\{{(U_{-1})_1}_\lambda L_2(w)\}_3\right)=0
=\{{\phi(U_{-1})_1}_\lambda \phi(L_2(w))\}_3
\,,
$$
namely $\phi$ is a PVA homomorphism with respect to the PVA structure defined by the 1-Adler identity
\eqref{eq:1adler-0}. However, from equation \eqref{C22} we have
$$
\phi\left(\{{(U_{-1})_1}_\lambda (U_{-1})_2\}^\epsilon\right)
=
\Omega\Big(U_{0}\otimes \id-\id\otimes U_{0}+2\epsilon(\id\otimes\id)\lambda\Big)
\,,
$$
which has  non-zero terms in $\epsilon^0$ and $\epsilon$. This implies that $\phi$ is not a PVA homomorphism
with respect to the 2-nd and 3-rd PVA structures on $\mc V_1$ defined by the identities
\eqref{eq:2adler-0} and \eqref{eq:3adler-0} respectively.

In order to make $\phi$ a PVA homomorphism we can consider on $\mc V_1$ the PVA structure
defined by the Dirac modification
$\{L_1(z)_\lambda L_2(w)\}^{\epsilon,D}$ (see \cite{DSKVDirac} for details on the Dirac reduction for PVA)
of the $\epsilon$-Adler identity \eqref{eq:3adler-eps} with respect to the constraints
($\alpha\in I$)
$$
\theta_{1,\alpha}=u_{-2,\alpha}-\Tr(E_\alpha)
\,
\qquad
\theta_{2,\alpha}=u_{-1,\alpha}
\,.
$$
Letting $\theta_i=\sum_{\alpha\in I}\theta_{i,\alpha}E^\alpha$, $i=1,2$, the constraints can be rewritten in compact 
form as $\theta_1=U_{-2}-\id$ and $\theta_2=U_{-1}$.
Let $C(\lambda)=\left(C_{(i,\alpha),(j,\beta)}(\lambda)\right)$ be the matrix, with coefficients in $\mc V_{-1}[\lambda]$, 
defined by
$$
C_{(i,\alpha),(j,\beta)}(\lambda)=\{{\theta_{j,\beta}}_\lambda \theta_{i,\alpha}\}^{\epsilon}
\,,
$$
and let $C^{-1}(\lambda)$ be its inverse. The Dirac modification $\{L_1(z)_\lambda L_2(w)\}^{\epsilon,D}$
of the $\epsilon$-Adler identity \eqref{eq:3adler-eps} is defined by the formula
\begin{equation}
\begin{split}\label{20220720:eq3}
&\{L_1(z)_\lambda L_2(w)\}^{\epsilon,D}
=\{L_1(z)_\lambda L_2(w)\}^{\epsilon}
\\
&-\sum_{\substack{i,j=1,2\\\alpha,\beta\in I}}
\{{\theta_{i,\alpha}}_{\lambda+\partial}L_2(w)\}^\epsilon
(C^{-1})_{(i,\alpha),(j,\beta)}(\lambda+\partial)\{L_1(z)_\lambda\theta_{j,\beta}\}^\epsilon
\,.
\end{split}
\end{equation}
The most important facts for us are that the modified $\epsilon$-Adler identity \eqref{20220720:eq3} defines a PVA structure on
$\mc V_1$
and that $\ker\phi$ is a PVA ideal for this structure (proofs can be found in \cite{DSKVDirac}).
It is clear from \eqref{20220720:eq1} that $\im\phi=\mc V_{-1}$. Hence the modified
Adler identity \eqref{20220720:eq3} induces a PVA structure on the quotient space
$\mc V_{-1}\cong\mc V_{1}/\Ker\phi$. 

Using the identities ($X,Y,Z\in A$)
\begin{equation}\label{eq:Tr}
(\Tr(Z\cdot)\otimes\id)(\Omega (X\otimes Y))=X Z Y
\,,\quad
(\id\otimes\Tr(Z\cdot))(\Omega (X\otimes Y))=Y Z X\,,
\end{equation}
and noticing that
$\{{u_{p,\alpha}}_\lambda L_2(w)\}^\epsilon=(\Tr(E_\alpha\cdot)\otimes\id)\{{(U_{p,\alpha})_1}_\lambda L_2(w)\}$, we get
from equations \eqref{20220720:eq2}:
\begin{equation}\label{20220720:eq4}
\begin{split}
&\phi\left(\{{\theta_{1,\alpha}}_\lambda L_2(w)\}^\epsilon\right)=
\id\otimes\phi(L_(w+\lambda))E_\alpha-\id\otimes E_\alpha\phi(L(w))\,,
\\
&\phi\left(\{{\theta_{2,\alpha}}_\lambda L_2(w)\}^\epsilon\right)=
\id\otimes\phi(L(w+\lambda))E_\alpha(w+2\epsilon)
-\id\otimes E_\alpha(w+\partial+2\epsilon)\phi(L(w))
\,,
\end{split}
\end{equation}
for every $\alpha\in I$. By skewsymmetry \eqref{eq:pva5} we also get
\begin{equation}\label{20220720:eq5}
\begin{split}
&\phi\left(\{L_1(z)_\lambda\theta_{1,\beta}\}^\epsilon\right)=
E_\beta \phi(L(z))\otimes\id-\phi(L^*(\lambda-z))E_\beta\otimes1\,,
\\
&\phi\left(\{L_1(z)_\lambda \theta_{2,\beta}\}^\epsilon\right)=
E_{\beta}(z+\partial+2\epsilon)\phi(L(z))\otimes\id
-\phi(L^*(\lambda-z))E_\beta(z+2\epsilon)\otimes\id
\,.
\end{split}
\end{equation}
Finally, using \eqref{eq:Tr} and equations \eqref{C11}-\eqref{C22}, we find the following explicit expressions for the entries of the matrix 
$\phi(C(\lambda))$ ($\alpha,\beta\in I$):
\begin{equation}\label{phiC}
\begin{split}
&\phi\left(C_{(1,\alpha),(1,\beta)}(\lambda)\right)=0\,,
\quad
\phi\left(C_{(1,\alpha),(2,\beta)}(\lambda)\right)
=\phi\left(C_{(2,\alpha),(1,\beta)}(\lambda)\right)=\Tr(E_\alpha E_\beta)\lambda\,,
\\
&\phi\left(C_{(2,\alpha),(2,\beta)}(\lambda)\right)
=\Tr([E_\alpha,E_{\beta}]U_0)+2\epsilon\, \Tr(E_\alpha E_\beta)\lambda
\,.
\end{split}
\end{equation}
From \eqref{phiC} it is immediate to get the expression for the entries of the matrix $\phi(C^{-1}(\lambda))$.
We have ($\alpha,\beta\in I$)
\begin{equation}\label{phiCinv}
\begin{split}
&\phi\left((C)^{-1}_{(1,\alpha),(1,\beta)}(\lambda)\right)=
-(\lambda+\partial)^{-1}\left(\Tr([E^\alpha,E^\beta]U_0)+2\epsilon\Tr(E^\alpha E^\beta)\lambda\right)
\lambda^{-1}\,,
\\
&\phi\left((C)^{-1}_{(1,\alpha),(2,\beta)}(\lambda)\right)
=\phi\left((C)^{-1}_{(2,\alpha),(1,\beta)}(\lambda)\right)=\Tr(E^\alpha E^\beta)\lambda^{-1}\,,
\\
&
\phi\left((C^{-1})_{(2,\alpha),(2,\beta)}(\lambda)\right)
=0
\,.
\end{split}
\end{equation}
Applying $\phi$ to both sides of the Dirac modified $\epsilon$ Adler-identity and using \eqref{20220720:eq4},
\eqref{20220720:eq5} and \eqref{phiCinv} we get the explicit form of the 
Dirac modified $\epsilon$ Adler-identity defining the induced PVA structure on the quotient space.
We summarize this in the next result
\begin{theorem}
Let $A$ be a finite dimensional unital associative algebra over $\mb F$ 
with a non-degenerate trace form, and fix dual bases $\{E_{\alpha}\}_{\alpha\in I}$
and $\{E^\alpha\}_{\alpha\in I}$ satisfying \eqref{eq:motiv7}. Let $\mc V:=\mc V_0$ be the algebra of differential polynomials in infinitely many variables $u_{p,\alpha}$,
$p\geq0$, $\alpha\in I$. Let
$$
L(\partial)=\id \partial+\sum_{p\geq0}U_p\partial^{-p-1}\in\mc V((\partial^{-1}))\otimes A\,,
$$
where $U_p=\sum_{\alpha\in I}u_{p,\alpha}E^\alpha\in\mc V\otimes A$.
The following identity
\begin{equation}\label{eq:3adler-eps-Dirac}
\begin{array}{l}
\displaystyle{
\vphantom{\Big(}
\{L_1(z)_\lambda L_2(w)\}^{\epsilon}
} \\
\displaystyle{
\vphantom{\Big(}
=
\Omega
\bigg(
(L_1(w+\lambda+\partial)+\epsilon\id)
(z-w-\lambda-\partial)^{-1}
L_2^*(\lambda-z)(L_2(w)+\epsilon\id)
} \\
\displaystyle{
\vphantom{\Big(}
-(L_1(w+\lambda+\partial)+\epsilon\id)
L_1(z)
(z-w-\lambda-\partial)^{-1}
(L_2(w)+\epsilon\id)
} \\
\displaystyle{
\vphantom{\Big(}
+
L_1(w+\lambda+\partial)
(L_1(z)+\epsilon\id)
(z-w-\lambda-\partial)^{-1}
(L_2^*(\lambda-z)+\epsilon\id)
}
\\
\displaystyle{
\vphantom{\Big(}
-
(L_1(z)+\epsilon\id)
(z-w-\lambda-\partial)^{-1}
(L_2^*(\lambda-z)+\epsilon\id)
L_2(w)
\bigg)
}
\\
\displaystyle{
\vphantom{\Big(}
+\sum_{\alpha,\beta\in I}
\left(\id\otimes L(w+\lambda+\partial)E_\alpha-
\id\otimes E_\alpha L(w)\right)
(\lambda+\partial)^{-1}\times
}\\
\displaystyle{
\vphantom{\Big(}
\qquad\qquad\times\Tr([E^\alpha,E^\beta]U_0)(\lambda+\partial^{-1})
\left(E_\beta L(z)\otimes\id-L^*(\lambda-z)E_\beta\otimes\id\right)
}
\\
\displaystyle{
\vphantom{\Big(}
-\sum_{\alpha\in I}
\left(
\id\otimes L(w+\lambda+\partial)E_\alpha-\id\otimes E_\alpha L(w)\right)
(\lambda+\partial)^{-1}\times
}\\
\displaystyle{
\vphantom{\Big(}
\qquad\qquad\times
\left(E^\alpha (z+\partial)L(z)\otimes\id-L^*(\lambda-z)E^\alpha z\otimes\id\right)
}
\\
\displaystyle{
\vphantom{\Big(}
-\sum_{\alpha\in I}
\left(\id\otimes L(w+\lambda+\partial)E_\alpha w-\id\otimes E_\alpha\left((w+\partial) L_2(w)\right)\right)
(\lambda+\partial)^{-1}\times
}\\
\displaystyle{
\vphantom{\Big(}
\qquad\qquad\times
\left(E^\alpha L(z)\otimes\id-L^*(\lambda-z)E^\alpha\otimes\id\right)
}
\\
\displaystyle{
\vphantom{\Big(}
-2\epsilon \sum_{\alpha\in I}
\left(\id\otimes L(w+\lambda+\partial)E_\alpha-\id\otimes E_\alpha L(w)\right)
(\lambda+\partial)^{-1}\times
}\\
\displaystyle{
\vphantom{\Big(}
\qquad\qquad\times
\left(E^\alpha L(z)\otimes\id-L^*(\lambda-z)E^\alpha\otimes\id\right)
\,,}
\end{array}
\end{equation}
defines a PVA structure on $\mc V$. Furthermore, expanding
$\{\cdot\,_\lambda\,\cdot\}^\epsilon=\{\cdot\,_\lambda\,\cdot\}_3+2\epsilon\{\cdot\,_\lambda\,\cdot\}_2
+\epsilon^2\{\cdot\,_\lambda\,\cdot\}_1$, we get three compatible PVA $\lambda$-brackets on $\mc V$.
\end{theorem}
\begin{proof}
The statement follows from the computations outlined in the present section and the results in
\cite{DSKVDirac} about Dirac reduction for PVA.
\end{proof}

As in Section \ref{sec:ham} define elements in $\mc V$ by $h_0=0$ and $h_{n}=\frac{-1}{n}\Res_z\Tr(L^n(z))$, for $n>0$. Then,
and the corresponding Hamiltonian equations \eqref{20220621:eq2} are
\begin{equation}\label{eq:KP}
\frac{d L(w)}{dt_n}=[(L^n)_+,L](w)\,,
\qquad n\geq0\,.
\end{equation}
These equations form  the $A$-valued KP hierarchy and equation \eqref{eq:3adler-eps-Dirac} gives its tri-Hamiltonian structure.
\begin{remark}
When $A=\mb F$, equation \eqref{eq:3adler-eps-Dirac} reads
\begin{equation}\label{scalarKP}
\begin{array}{l}
\displaystyle{
\vphantom{\Big(}
\{L(z)_\lambda L(w)\}^{\epsilon}
=(L(w+\lambda+\partial)+\epsilon)
(z-w-\lambda-\partial)^{-1}
L^*(\lambda-z)(L(w)+\epsilon)
} \\
\displaystyle{
\vphantom{\Big(}
-(L(w+\lambda+\partial)+\epsilon)
L(z)
(z-w-\lambda-\partial)^{-1}
(L(w)+\epsilon)
\big)
} \\
\displaystyle{
\vphantom{\Big(}
+
L(w+\lambda+\partial)
(L(z)+\epsilon)
(z-w-\lambda-\partial)^{-1}
(L^*(\lambda-z)+\epsilon)
}
\\
\displaystyle{
\vphantom{\Big(}
-
(L(z)+\epsilon)
(z-w-\lambda-\partial)^{-1}
(L^*(\lambda-z)+\epsilon)
L(w)
\bigg)
}
\\
\displaystyle{
\vphantom{\Big(}
-(L(w+\lambda+\partial)-L(w))(\lambda+\partial)^{-1}
((z+\partial)L(z)-L^*(\lambda-z) z)
}
\\
\displaystyle{
\vphantom{\Big(}
-(L(w+\lambda+\partial)w-\left((w+\partial) L(w)\right))(\lambda+\partial)^{-1}
(L(z)-L^*(\lambda-z))
}
\\
\displaystyle{
\vphantom{\Big(}
-2\epsilon\,(L(w+\lambda+\partial)-L(w))(\lambda+\partial)^{-1}
(L(z)-L^*(\lambda-z))
\,.}
\end{array}
\end{equation}
The coefficient of $2\epsilon$ and $\epsilon^2$ agree, up to an overall minus sign, with the Adler type formulas used in 
\cite{AGD} to define the bi-Hamiltonian structure of the KP hierarchy. Moreover, note that 
the constant term in $\epsilon$ in \eqref{scalarKP} defines a local PVA. Hence, the same computations as in \cite{AGD} show that \eqref{eq:KP} is a tri-Hamiltonian integrable hierarchy.
We expect this to be true for arbitrary associative algebras $A$.
\end{remark}

\begin{remark}
Let $V$ be a finite dimensional vector space and $A=\End(V)$. Then,
the coefficient of $2\epsilon$ and $\epsilon^2$ in \eqref{eq:3adler-eps-Dirac} agree, up to an overall minus sign, with the Adler type formulas used in 
\cite{AGD} to define the bi-Hamiltonian structure of the matrix KP hierarchy.
\end{remark}

\begin{remark}
Replacing $R^{(0)}$ with $R^{(1)}$ similar computations as in the present section lead to the Hamiltonian
formalism for the modified KP hierarchy, see \cite{Kup85}.
\end{remark}

\begin{remark}
In \cite{AGD}, for every $N\geq1$, a compatible pair of $\lambda$-brackets describing a bi-Hamiltonian structure for the KP hierarchy has been found. These are obtained from Proposition \ref{20220719:prop1}(b)
and (c). However Proposition \ref{20220719:prop1} shows that it is not possible to get the third
compatible $\lambda$-bracket if $N>1$. To overcome this problem one has to consider a (non-local) Dirac modification
of the $\epsilon$-Adler identity \eqref{eq:3adler-eps} by the constraints $\theta_i=U_{-N-1-i}$, $i\geq1$. 
\end{remark}

%
%
%
%
%
%
%
%

\end{document}